\documentclass[11pt]{article}

\usepackage{booktabs} 
\usepackage{graphicx}
\usepackage{subcaption}
\usepackage{microtype}

\usepackage{hyperref}
\hypersetup{hypertexnames=false}

\let\MR\relax

\usepackage{mathtools}

\usepackage{adjustbox}

\usepackage{amsmath}
\usepackage{amssymb}
\usepackage{amsthm}
\usepackage{thmtools}
\usepackage{bm}
\usepackage{xcolor,colortbl}
\usepackage{dsfont}
\usepackage{authblk}
\usepackage{nicefrac}
\usepackage{complexity}

\usepackage{tikz}
\usepackage{forest}
\usetikzlibrary{shapes.misc,positioning,matrix,arrows.meta}
\usepackage{bbm}
\usepackage{complexity}
\usepackage{float}
\usepackage{mathrsfs}
\usepackage{nicematrix}

\usepackage{pgfplots}
\pgfplotsset{compat=newest}

\usepackage{mleftright}
\usepackage{thm-restate}
\usepackage{physics}
\usepackage[capitalize,noabbrev,nameinlink]{cleveref}
\usepackage[shortlabels]{enumitem}
\usepackage{xspace}

\usepackage{fullpage}
\usepackage{parskip}
\usepackage{MnSymbol}

\crefname{item}{Item}{Item}

\newcounter{qst}
\crefname{qst}{Question}{Questions}

\newcommand{\declarecolor}[2]{\definecolor{#1}{RGB}{#2}\expandafter\newcommand\csname #1\endcsname[1]{\textcolor{#1}{##1}}}
\declarecolor{White}{255, 255, 255}
\declarecolor{Black}{0, 0, 0}
\declarecolor{Maroon}{128, 0, 0}
\declarecolor{Coral}{255, 127, 80}
\declarecolor{Red}{182, 21, 21}
\declarecolor{LimeGreen}{50, 205, 50}
\declarecolor{DarkGreen}{0, 80, 0}
\declarecolor{Purple}{146, 42, 158}
\declarecolor{Navy}{0, 0, 128}
\declarecolor{LightBlue}{84, 101, 202}
\definecolor{mydarkblue}{rgb}{0,0.08,0.45}

\makeatletter
\patchcmd\algocf@Vline{\vrule}{\vrule \kern-0.4pt}{}{}
\patchcmd\algocf@Vsline{\vrule}{\vrule \kern-0.4pt}{}{}
\makeatother

\makeatletter
\let\cref@old@stepcounter\stepcounter
\def\stepcounter#1{%
  \cref@old@stepcounter{#1}%
  \cref@constructprefix{#1}{\cref@result}%
  \@ifundefined{cref@#1@alias}%
    {\def\@tempa{#1}}%
    {\def\@tempa{\csname cref@#1@alias\endcsname}}%
  \protected@edef\cref@currentlabel{%
    [\@tempa][\arabic{#1}][\cref@result]%
    \csname p@#1\endcsname\csname the#1\endcsname}}
\makeatother

\theoremstyle{plain}
\newtheorem{theorem}{Theorem}[section]
\newtheorem*{theorem*}{Theorem}
\newtheorem{lemma}[theorem]{Lemma}

\newtheorem{proposition}[theorem]{Proposition}

\newtheorem{property}[theorem]{Property}

\theoremstyle{definition}
\newtheorem{definition}[theorem]{Definition}

\theoremstyle{remark}
\newtheorem{remark}[theorem]{Remark}
\newtheorem{example}[theorem]{Example}

\let\E\relax

\newcommand{\tlow}[1]{\underline{t_{#1}}}
\newcommand{\thigh}[1]{\overline{t_{#1}}}

\newcommand{\period}[1]{[\tlow{#1}, \thigh{#1}]}
\newcommand{\perm}[1]{\mathfrak{S}_{#1}}

\newcommand{\vx}{\vec{x}}
\newcommand{\tvx}{\tilde{\vec{x}}}

\newcommand{\vy}{\vec{y}}
\renewcommand{\vu}{\vec{u}}

\newcommand{\proj}{\Pi}

\newcommand{\regcon}{R}
\newcommand{\etacon}{1}
\newcommand{\vecx}[2]{\vx_{#1}^{(#2)}\xspace}

\usepackage{natbib}
\newcommand{\xstar}{\vec{x}^{\star}\xspace}

\newcommand{\innerprod}[2]{
\left\langle  #1, #2 \right\rangle \xspace
}

\newcommand*{\N}{{\mathbb{N}}}

\let\R\relax
\newcommand*{\R}{{\mathbb{R}}}
\newcommand*{\E}{\operatornamewithlimits{\mathbb E}}

\newcommand*{\cX}{{\mathcal{X}}}

\newcommand*{\cR}{{\mathcal{R}}}

\newcommand*{\cA}{{\mathcal{A}}}

\newcommand{\defeq}{\coloneqq}

\newcommand{\reg}{\mathsf{Reg}}

\mathchardef\mathhyphen="2D

\NewDocumentCommand{\treeset}{o}{\mathbb{T}\IfNoValueF{#1}{_{#1}}}

\newcommand{\mwu}{\texttt{MWU}\xspace}

\newcommand{\md}{\texttt{MD}\xspace}
\newcommand{\ftrl}{\texttt{FTRL}\xspace}

\newcommand{\fictp}{\texttt{FP}}

\newcommand{\gap}[3]{\texttt{Gap}_{#1}^{(#2)}[#3]}

\DeclarePairedDelimiterX{\infdivx}[2]{(}{)}{%
  #1\;\delimsize\|\;#2%
}

\renewcommand{\vec}[1]{\bm{#1}}
\newcommand{\mat}[1]{\mathbf{#1}}

\DeclarePairedDelimiterX{\card}[1]{\lvert}{\rvert}{#1}
\DeclarePairedDelimiterX{\tuple}[1]{\lparen}{\rparen}{#1}
\DeclarePairedDelimiterX{\parens}[1]{\lparen}{\rparen}{#1}
\DeclarePairedDelimiterX{\brackets}[1]{\lbrack}{\rbrack}{#1}
\DeclarePairedDelimiterX{\set}[1]\{\}{#1}
\let\Pr\relax
\DeclarePairedDelimiterXPP{\Pr}[1]{\mathbb{P}}[]{}{#1}
\DeclarePairedDelimiterXPP{\PrX}[2]{\mathbb{P}_{#1}}[]{}{#2}
\DeclarePairedDelimiterXPP{\Ex}[1]{\mathbb{E}}[]{}{#1}
\DeclarePairedDelimiterXPP{\ExX}[2]{\mathbb{E}_{#1}}[]{}{#2}

\hypersetup{ %
    pdftitle={},
    pdfkeywords={},
    pdfborder=0 0 0,
    pdfpagemode=UseNone,
    colorlinks=true,
    linkcolor=Maroon,
    citecolor=Navy,
    filecolor=Purple,
    urlcolor=Purple,
}

\usetikzlibrary{fit,shapes.misc,arrows.meta}
\makeatletter
\tikzset{
  fitting node/.style={
    inner sep=0pt,
    fill=none,
    draw=none,
    reset transform,
    fit={(\pgf@pathminx,\pgf@pathminy) (\pgf@pathmaxx,\pgf@pathmaxy)}
  },
  reset transform/.code={\pgftransformreset}
}
\tikzset{cross/.style={path picture={
  \draw[black]
(path picture bounding box.south east) -- (path picture bounding box.north west) (path picture bounding box.south west) -- (path picture bounding box.north east);
}}}
\tikzstyle{ox}=[semithick,draw=black,circle,cross,inner sep=1.2mm]
\makeatother

\newcommand{\nc}{\newcommand}
\newcount\Comments
\nc\io[1]{\ifnum\Comments=1 {\textcolor{purple}{[ioannis: #1]}}\fi}
\nc\br[1]{\ifnum\Comments=1 {\textcolor{teal}{[brian: #1]}}\fi}
\nc\et[1]{\ifnum\Comments=1 {\textcolor{orange}{[ET: #1]}}\fi}

\nc{\Opthedge}{OMWU\xspace}

\nc{\DMO}{\DeclareMathOperator}
\nc\old[1]{\textcolor{brown}{[old: #1]}}
\nc{\BR}{\mathbb{R}}
\nc{\BC}{\mathbb{C}}
\DMO{\Bin}{Bin}
\nc{\BN}{\mathbb{N}}
\nc{\distrs}[1]{\Delta({#1})}
\nc{\BZ}{\mathbb{Z}}
\nc{\ep}{\epsilon}
\nc{\ra}{\rightarrow}
\nc{\st}{\star}
\nc{\Reg}[2]{\REG_{{#1},{#2}}}
\nc{\til}{\tilde}
\nc{\kld}[2]{\KL({#1};{#2})}
\nc{\chisq}[2]{\chi^2({#1};{#2})}
\DMO{\POLYLOG}{polylog}

\nc{\matx}[1]{\left(\begin{matrix}#1\end{matrix}\right)}
\DMO{\VAR}{Var}
\DMO{\COV}{Cov}
\nc{\Var}[2]{\VAR_{{#1}}\left({#2}\right)}
\nc{\Cov}[3]{\COV_{{#1}}\left({#2},{#3}\right)}
\DMO{\DD}{D}
\nc{\fd}[2]{\DD_{#1}{#2}}
\nc{\fds}[3]{\left(\fd{#1}{#2}\right)\^{#3}}
\nc{\fdc}[2]{\DD^\circ_{#1}{#2}}
\nc{\fdcs}[3]{\left(\fdc{#1}{#2}\right)\^{#3}}
\nc{\shf}[2]{\EEE_{#1}{#2}}
\nc{\shfs}[3]{\left(\shf{#1}{#2}\right)\^{#3}}
\nc{\normst}[2]{\left\| {#2} \right\|_{#1}^\st}
\renewcommand{\^}[1]{^{(#1)}}
\DeclareMathOperator*{\argmax}{arg\,max}

\nc{\lng}{\langle}
\nc{\rng}{\rangle}
\nc{\bbone}{\mathbf{1}}
\nc{\bbzero}{\mathbf{0}}
\nc{\MD}{\mathcal{D}}
\nc{\MM}{\mathcal{M}}
\nc{\MZ}{\mathcal{Z}}
\nc{\MU}{\mathcal{U}}
\nc{\MR}{\mathcal{R}}
\nc{\MC}{\mathcal{C}}
\nc{\MT}{\mathbb{T}^{n}}
\nc{\MS}{\mathcal{S}}
\nc{\MX}{\mathcal{X}}
\nc{\MY}{\mathcal{Y}}
\nc{\MB}{\mathcal{B}}
\nc{\MJ}{\mathcal{J}}
\nc{\MF}{\mathcal{F}}
\nc{\MG}{\mathcal{G}}
\nc{\ML}{\mathcal{L}}
\nc{\MQ}{\mathcal{Q}}

\nc{\ba}{\mathbf{A}}
\nc{\bx}{\mathbf{x}}
\nc{\by}{\mathbf{y}}
\nc{\bz}{\mathbf{z}}
\nc{\bs}{\mathbf{s}}
\nc{\bt}{\mathbf{t}}

\nc{\ME}{\mathcal{E}}
\DMO{\View}{View}
\DMO{\KL}{KL}
\nc{\MW}{\mathcal{W}}
\nc{\CS}{\mathscr{S}}
\nc{\CI}{\mathscr{I}}
\nc{\CQ}{\mathscr{Q}}
\nc{\CL}{\mathscr{L}}
\nc{\CM}{\mathscr{M}}
\nc{\CG}{\mathscr{G}}
\nc{\CR}{\mathscr{R}}

\nc{\wh}{\widehat}

\nc{\BM}{BM\xspace}
\nc{\ALG}{\texttt{ALG}}
\nc{\MCT}{{\rm MCT}}

\nc{\matrixLL}{\mat{L}}
\nc{\vectorecks}{\vec{x}}
\nc{\vectorLL}{\vec{\ell}}
\nc{\matrixKYU}{\mat{Q}}
\nc{\rowdot}{\cdot}

\nc{\X}{\mathcal{X}}
\nc{\Y}{\mathcal{Y}}

\newcommand{\delimit}[3]{\newcommand{#1}[1]{\left#2##1\right#3}}
\delimit \ceil \lceil \rceil
\delimit \floor \lfloor \rfloor

\definecolor{p0color}{RGB}{0,0,0}
\definecolor{p1color}{RGB}{31,119,180}
\definecolor{p2color}{RGB}{214,39,40}

\def\ve{{\bm{e}}}

\def\vu{{\bm{u}}}

\def\vx{{\bm{x}}}
\def\vy{{\bm{y}}}

\renewcommand\vec\bm

\def\mA{{\mathbf{A}}}
\def\mB{{\mathbf{B}}}

\newcommand{\Phirange}{\Phi_{\mathsf{range}}}

\newcommand{\idx}{p}

\newcommand{\calT}{\mathcal{T}}

\title{(Doubly) Exponential Lower Bounds for Follow the Regularized Leader in Potential Games}

\author[1]{Ioannis Anagnostides\thanks{These authors contributed equally.}}
\author[2]{Ioannis Panageas}
\author[2]{Nikolas Patris\protect\footnotemark[1]}
\author[1,3]{Tuomas Sandholm}

\affil[1]{Carnegie Mellon University}
\affil[2]{University of California, Irvine}
\affil[3]{\small Additional affiliations: Strategy Robot, Inc., Strategic Machine, Inc., Optimized Markets, Inc.}
\affil[ ]{\texttt{\{ianagnos,sandholm\}}\texttt{@cs.cmu.edu}, \texttt{\{ipanagea,npatris\}}\texttt{@uci.edu}}

\begin{document}

\maketitle

\begin{abstract}
    \emph{Follow the regularized leader $(\ftrl)$} is the premier algorithm for online optimization. However, despite decades of research on its convergence in constrained optimization---and potential games in particular---its behavior remained hitherto poorly understood. In this paper, we establish that $\ftrl$ can take exponential time to converge to a Nash equilibrium in two-player potential games for any (permutation-invariant) regularizer and potentially vanishing learning rate. By known equivalences, this translates to an exponential lower bound for certain mirror descent counterparts, most notably multiplicative weights update. On the positive side, we establish the potential property for $\ftrl$ and obtain an exponential upper bound $\exp(O_{\epsilon}(1/\epsilon^2))$ for any no-regret dynamics executed in a lazy, alternating fashion, matching our lower bound up to factors in the exponent. Finally, in multi-player potential games, we show that fictitious play---the extreme version of $\ftrl$---can take \emph{doubly} exponential time to reach a Nash equilibrium. This constitutes an exponentially stronger lower bound for the foundational learning algorithm in games.
\end{abstract}

\section{Introduction}
\label{sec:introduction}

\emph{Multiplicative weights update ($\mwu$)} is the quintessential online algorithm~\citep{Littlestone94:Weighted}, attaining the optimal \emph{regret} bound in many fundamental online learning problems~\citep{Fan25:Universal}. It has found applications in diverse areas ranging from approximation algorithms to boosting in machine learning~\citep{Arora12:Multiplicative}. $\mwu$ is an instantiation of the seminal online learning paradigm \emph{follow the regularized leader $(\ftrl)$}~\citep{Kalai05:Efficient}. Viewed differently, $\mwu$ can also be cast as an instance of online \emph{mirror descent $(\md)$}~\citep{Nemirovski83:Problem}.

A celebrated connection inextricably links the no-regret property, attained by algorithms such as $\ftrl$ and $\md$, to \emph{coarse correlated equilibria (CCEs)}, a seminal game-theoretic solution concept~\citep{Moulin78:Strategically,Aumann74:Subjectivity}. Specifically, if each player in a multi-player game employs a no-regret algorithm to update its strategy, the \emph{average} correlated distribution of play converges to the set of CCEs. This connection, however, provides no information about the \emph{last-iterate} convergence of the dynamics.

Characterizing the convergence of $\mwu$ and other online algorithms has received extensive attention. \citet{Kleinberg09:Multiplicative} first established that $\mwu$ converges to Nash equilibria for almost all \emph{potential} games. In this setting, players are effectively ascending a global potential function; thus, the dynamics can be viewed through the lens of constrained optimization, with Nash equilibria corresponding to first-order stationary points. Non-asymptotic rates soon emerged for some versions of $\md$, most notably gradient descent, but despite decades of research on this topic, a major question remained hitherto wide open:
\begin{quote}
    \centering
    \emph{How many iterations are needed for $\ftrl$ algorithms to converge in potential games?}
\end{quote}
Viewed differently, can $\mwu$ or $\ftrl$ efficiently find first-order stationary points in constrained optimization?

Although $\mwu$ has been shown to asymptotically converge in potential games---at least when the set of Nash equilibria comprises isolated points, a condition that holds for almost all potential games~\citep{Kleinberg09:Multiplicative,Palaiopanos17:Multiplicative}---this result has not been extended to the general class of $\ftrl$. $\ftrl$ presents distinct challenges that complicate the analysis compared to algorithms such as gradient descent. For example, it can exhibit \emph{spurious fixed points}: due to its history dependence, the dynamics may stagnate even when the current strategy is highly suboptimal relative to the observed utility (\Cref{example:spuriousFPs}).

\subsection{Our results}

We provide new upper and lower bounds for the convergence of $\ftrl$ in potential games and constrained optimization.


Our first positive result goes beyond $\ftrl$ and covers the entire class of no-regret dynamics. We establish a non-asymptotic rate of convergence to Nash equilibria, provided the updates are alternating and \emph{lazy}---the player performs an update only when the proposed strategy guarantees a sufficient improvement. 

\begin{theorem}
    \label{theorem:positive}
    In any potential game, alternating $\epsilon$-lazy no-regret dynamics converge to an $\epsilon$-Nash equilibrium after at most $\exp(O_\epsilon(1/\epsilon^2))$ iterations.
\end{theorem}

\begin{figure*}[t]
    \centering
    \includegraphics[scale=0.6]{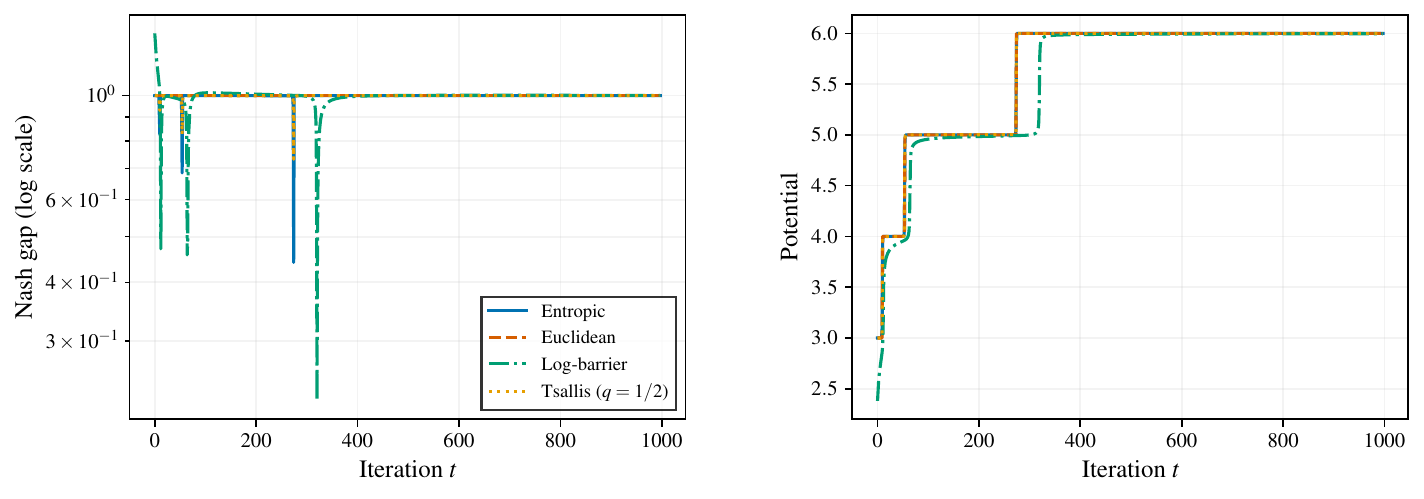}
    \caption{Illustration of our results: $\ftrl$ algorithms have the potential property (right), but can take exponential time to converge (left).}
    \label{fig:algs}
\end{figure*}

The proof proceeds by bounding the number of updates that can occur and then deriving a recursion for the time required between successive updates, paving the way to~\Cref{theorem:positive}. To put this into context, it is well-known that no-regret dynamics can fail to converge to Nash equilibria in potential games~\citep{Kleinberg09:Multiplicative}; \Cref{theorem:positive} offers a natural way to circumvent those impossibility results.

Turning to the special case of $\ftrl$, we show that, for sufficiently small learning rate, the potential value is nondecreasing (\Cref{fig:algs}, right), and this does not require alternation nor laziness (\Cref{prop:pot-improvement}). As a result, $\ftrl$ is bound to converge to a set in which the value of the potential is constant.

While~\Cref{theorem:positive} only establishes an exponential upper bound in $1/\epsilon^2$, our next result justifies this by establishing a fundamental barrier for $\ftrl$ dynamics, which manifests itself with or without alternation or laziness.

\begin{theorem}
\label{theorem:negative}
    In two-player $m \times m$ potential games, $\ftrl$ can take $2^{\Omega(m \log m)}$ iterations to converge to an $\epsilon$-Nash equilibrium for $\epsilon = 1/\poly(m)$. This holds for any permutation-invariant regularizer and learning rate $\eta^{(t)} = 1/t^\alpha$ for $\alpha \in [0, 1)$.
\end{theorem}

Crucially, this theorem applies to \emph{any} member of $\ftrl$ with a permutation invariant regularizer, which encompasses most common incarnations of $\ftrl$ (\Cref{fig:algs}, left). By virtue of known equivalences, \Cref{theorem:negative} thereby circumscribes specific $\md$ algorithms as well, most notably $\mwu$. Furthermore, \Cref{theorem:negative} holds even under a vanishing learning rate $\eta^{(t)} = 1/t^\alpha$ for $\alpha \in [0, 1)$. The analytical challenge in this regime stems from the propensity of $\ftrl$ to mix actions, especially when the learning rate is close to zero. Another notable aspect of our lower bound construction is that it tightly mirrors the mechanism driving our upper bound.

Finally, we turn our attention to $n$-player potential games. We analyze \emph{fictitious play ($\fictp$)}~\citep{Robinson51:iterative,Brown51:Iterative}, the foundational learning algorithm in games, which can be viewed as the extreme version of $\ftrl$ when $\eta \to \infty$. We show that $\fictp$ can take \emph{doubly} exponentially many iterations to converge to a Nash equilibrium.

\begin{theorem}
    There is an $n$-player binary-action potential game in which $\fictp$ takes $2^{\Omega(2^n)}$ iterations to reach an $\epsilon$-Nash equilibrium for $\epsilon = 1/2^n$.
\end{theorem}

This constitutes an exponentially stronger lower bound relative to prior work~\citep{Panageas23:Exponential,Brandt13:Rate}. From a technical standpoint, our multi-player construction distills the essential property behind the two-player lower bound by making a connection to a classic graph-theoretic problem connected to coding theory (\Cref{sec:multiplayer}).

\subsection{Related work}

\paragraph{Learning in potential games} There is a rich literature on the convergence of learning dynamics in potential games. Unlike general games, potential games always admit a \emph{pure} Nash equilibrium, and better-response dynamics converge to one~\citep{Rosenthal73:Class,Monderer96:Potential}. Fictitious play is also known to converge in this setting, though it was recently shown that it may require exponential time to do so~\citep{Panageas23:Exponential}. Our lower bound in two-player game adapts their construction; unlike fictitious play, $\ftrl$ has, by definition, a propensity to mix, which significantly complicates the analysis. The construction of~\citet{Panageas23:Exponential} was also recently used by~\citet{Anagnostides25:Convergence} to show that \emph{regret matching}~\citep{Hart00:Simple} can also require exponential time to converge.

Beyond these classic dynamics, significant effort has been devoted to analyzing other learning algorithms in potential games~\citep{Heliou17:Learning,Palaiopanos17:Multiplicative,Anagnostides22:Last,Hart03:Regret,Marden07:Regret,Candogan13:Dynamics,Maheshwari25:Independent}. From a complexity standpoint, computing Nash equilibria in potential games is likely intractable when the precision is exponentially small~\citep{Babichenko21:Settling,Babichenko20:Communication}, but amenable to algorithms such as gradient descent when the precision is inverse polynomial. Our results reveal that $\ftrl$ algorithms are poor first-order optimizers even in the inverse polynomial regime.

From a broader vantage point, the convergence of $\ftrl$ in games is a topic of active research~\citep{Lotidis25:Multiagent,Lotidis25:Robust}.

\paragraph{Fictitious play} Following the pioneering work of~\citet{Robinson51:iterative}, the convergence of $\fictp$ has received extensive attention in two-player zero-sum games. \citet{Daskalakis14:Counter} proved an exponential lower bound, albeit under adversarial tie breaking. Characterizing its convergence beyond adversarial tie breaks remains open, although~\citet{Wang25:TieBreaking} recently made progress. Our construction holds regardless of how ties are broken.

\paragraph{Forgetfulness and convergence speed} Our finding that $\ftrl$ is exponentially slower than $\md$ in potential games mirrors a recent result by~\citet{Cai24:Fast}, albeit in a fundamentally different problem and class of algorithms. Specifically, their result pertains to (two-player) zero-sum games and \emph{optimistic} algorithms.
\section{Preliminaries}
\label{sec:prels}

Our paper revolves around the $\ftrl$ algorithm, which is reviewed in~\Cref{sec:ftrl}. \Cref{sec:potential} provides some basic background on potential games.

\subsection{Online learning and $\ftrl$}
\label{sec:ftrl}

We examine the convergence of \emph{follow the regularized leader ($\ftrl$)}. This is a classic online algorithm introduced by~\citet{Kalai05:Efficient}, which became a mainstay in online optimization~\citep{Shalev-Shwartz12:Online}. Taking a step back, in the usual online learning framework, a learner interacts with an environment over a sequence of $T$ rounds. In each round $t \in [T]$, the learner selects a strategy $\vx^{(t)} \in \cX$, where $\cX$ is a convex and compact set such as the probability simplex. The environment then specifies a utility vector $\vu^{(t)}$, so that the utility of the learner in round $t$ reads $\langle \vx^{(t)}, \vu^{(t)} \rangle$. $\ftrl$ is a specific update rule that maps the history of observed utilities to the next strategy of the learner. Specifically, upon observing a sequence of utilities $(\vu^{(\tau)})_{\tau=1}^t$, it computes
\begin{equation}
    \label{eq:FTRL}
    \vx^{(t+1)} \defeq \argmax_{\vx \in \cX} \left\{ \left\langle \vx, \sum_{\tau=1}^t \vec{u}^{(\tau)} \right\rangle - \frac{1}{\eta^{(t)}} \cR(\vx) \right\}.
\end{equation}
Here, $\cR$ is a $1$-strongly convex and continuously differentiable\footnote{Differentiability is assumed over some open set $\tilde{\cX} \supseteq \cX$.} regularizer with respect to some norm $\|\cdot\|$: $\cR(\vx) \geq \cR(\vx') + \langle \nabla \cR(\vx'), \vx - \vx' \rangle + \frac{1}{2} \|\vx - \vx' \|^2$ for any $\vx, \vx'$; and $\eta^{(t)} > 0$ is the (nonincreasing) \emph{learning rate} sequence. When $\eta^{(t)} = +\infty$, $\ftrl$ approaches \emph{fictitious play}~\citep{Robinson51:iterative}---also known as \emph{follow the leader} in the online learning literature---formally defined in the sequel. We will denote by $\regcon$ the \emph{range} of the regularizer $\cR$, that is, $| \cR(\vx) - \cR(\vx') | \leq \regcon$ for any $\vx, \vx' \in \cX$. Also, $\|\cdot\|_*$ denotes the dual norm of $\|\cdot\|$; that is, $\|\vec{u}\|_* = \sup_{\|\vec{x} \| \leq 1} \langle \vec{x}, \vec{u} \rangle $.

A classic result due to~\citet{Kalai05:Efficient} shows that, unlike fictitious play, $\ftrl$ has the \emph{no-regret} property even against an adversarially produced sequence of utilities. Regret is the most common way of measuring performance in an online learning environment. It is defined as
\begin{equation}
    \reg^{(T)} = \max_{\vx' \in \cX} \sum_{t=1}^T \langle \vx' - \vx^{(t)}, \vu^{(t)} \rangle.
\end{equation}

We are now ready to formally state the regret guarantee of $\ftrl$, as defined in~\eqref{eq:FTRL}.

\begin{proposition}[\citealp{Kalai05:Efficient,Shalev-Shwartz12:Online}]
    \label{prop:FTRL-regret}
    For any sequence of utilities $(\vu^{(t)})_{t=1}^T$, the regret of $\ftrl$ can be bounded as
    \begin{equation*}
        \reg^{(T)} \leq \frac{\regcon}{\eta^{(T)}} +  \sum_{t=1}^T \eta^{(t)} \|\vu^{(t)} \|^2_*
    \end{equation*}
    for any nonincreasing learning rate sequence $\eta^{(t)}$.
    In particular, if $\| \vu^{(t)} \|_* \leq B$ for all $t$ and $\eta^{(t)} = \nicefrac{C(m, B)}{t^\alpha}$ for some $\alpha \in (0, 1)$,
    \begin{equation*}
        \reg^{(T)} \leq \frac{\regcon T^{\alpha}}{C(m, B)} + B^2 C(m, B) \left( 1 + \frac{T^{1 - \alpha}}{1 - \alpha} \right).
    \end{equation*}
\end{proposition}
Above, we used the fact that $\sum_{t=1}^T \nicefrac{1}{t^\alpha} \leq 1 + \nicefrac{T^{1 - \alpha}}{(1 - \alpha)}$. The key takeaway of~\Cref{prop:FTRL-regret} is that $\reg^{(T)} = o(T)$ for any $\alpha \in (0, 1)$. We proceed to review some common instantiations of $\ftrl$, most notably $\mwu$.

\begin{example}[$\mwu$]
    A canonical member of $\ftrl$ is $\mwu$. It is the instantiation of $\ftrl$ on the probability simplex $\cX = \Delta(\cA)$ when $\cR$ is the (negative) entropy regularizer, $\cR(\vx) = \sum_{a \in \cA} \vx[a] \log \vx[a]$, which can be shown to be $1$-strongly convex with respect to $\|\cdot\|_1$. In that case, $\ftrl$ admits the closed-form solution
    \begin{equation*}
        \tag{\texttt{MWU}}
        \vx^{(t+1)}[a] = \frac{e^{\eta^{(t)} \sum_{\tau=1}^t \vu^{(\tau)}[a]}}{\sum_{a' \in \cA} e^{\eta^{(t)} \sum_{\tau=1}^t \vu^{(\tau)}[a']}} \quad \forall a \in \cA. 
    \end{equation*}
    When the learning rate is a constant, this can be equivalently written as
    \begin{equation*}
        \vx^{(t+1)}[a] = \vx^{(t)}[a] \frac{ e^{\eta \vu^{(t)}[a]}}{\sum_{ a' \in \cA } \vx^{(t)}[a'] e^{\eta \vu^{(t)}[a']} } \quad \forall a \in \cA.
    \end{equation*}
    
    Since the range of the entropy regularizer is $\log |\cA|$, \Cref{prop:FTRL-regret} implies that the regret of $\mwu$ is bounded by $\sqrt{T \log|\cA|}$ when $\|\vu^{(t)} \|_\infty \leq 1$ for all $t$.
\end{example}

\begin{example}[Euclidean regularization]
    Another canonical member of the $\ftrl$ family arises when $\cR: \vx \mapsto \frac{1}{2} \|\vx \|_2^2$. In that case, the update rule can be cast as
    \begin{equation*}
        \vx^{(t+1)} = \proj_{\cX} \left( \eta \sum_{\tau=1}^t \vu^{(\tau)} \right),
    \end{equation*}
    where $\proj_{\cX}(\cdot)$ denotes the Euclidean projection onto $\cX$.
\end{example}
Two other notable regularizers are i) the \emph{logarithmic regularizer} $\cR: \vx \mapsto - \sum_{a \in \cA} \log \vx[a]$, and ii) \emph{Tsallis entropy} $\cR: \vx \mapsto - \frac{1}{q(1 - q)} \sum_{a \in \cA} (\vx[a])^q$ for $q \in (0, 1)$, both of which have found many applications in, among others, learning with bandit feedback~\citep{Zimmert21:Tsallis,Abernethy08:Competing,Wei18:More}. All of these regularizers are permutation invariant (\Cref{def:permutation_invariance}).

A regularizer $\cR$ is called a \emph{Legendre regularizer} if, in addition to the previous properties, for any sequence of points $\vx^{(1)}, \vx^{(2)}, \dots $ converging to a boundary point of $\cX$, $\nabla \cR(\vx^{(T)}) \to \infty$ as $T \to \infty$~\citep{Cesa-Bianchi06:Prediction}. A well-known fact is that $\ftrl$ and $\md$ are equivalent under any Legendre regularizer; for example, we refer to~\citet{McMahan11:Follow,Mcmahan17:Survey}. This means that all our results concerning $\ftrl$---both upper and lower bounds---immediately translate to certain $\md$ variants as well.

\paragraph{Fictitious play} Fictitious play~\citep{Robinson51:iterative,Brown51:Iterative}, also known as \emph{follow the leader} in the online learning literature, can be viewed as the extreme case of $\ftrl$ with $\eta \to \infty$. Specifically, 
\begin{equation*}
    \vx^{(t+1)} \defeq \argmax_{\vx \in \cX} \left\{ \left\langle \vx, \sum_{\tau=1}^t \vec{u}^{(\tau)} \right\rangle \right\}.
\end{equation*}
When $\cX = \Delta(\cA)$, it is assumed that fictitious play selects a pure strategy.

\subsection{Games and solution concepts}
\label{sec:potential}

In an $n$-player normal-form game, each player $i \in [n]$ selects as strategy a probability distribution $\vx_i \in \Delta(\cA_i) \eqqcolon \cX_i$ from a finite set of actions $\cA_i$. We denote by $u_i(\vx)$ the expected utility of player $i \in [n]$ under the joint strategy $(\vx_1, \dots, \vx_n)$. We also use the notation $\vx_{-i} = (\vx_1, \dots, \vx_{i-1}, \vx_{i+1}, \dots, \vx_n)$. The most standard solution concept in games is the \emph{Nash equilibrium}, recalled below.

\begin{definition}[\citealp{Nash50:Equilibrium}]
    \label{def:NE}
    A joint strategy $(\vx_1, \dots, \vx_n)$ is an \emph{$\epsilon$-Nash equilibrium} if for any player $i \in [n]$ and deviation $\vx_i' \in \cX_i$, 
    \[
        u_i(\vx_i', \vx_{-i}) \leq u_i(\vx_i, \vx_{-i}) + \epsilon.
    \]
\end{definition}

As we alluded to in our introduction, there is an intimate connection between the no-regret property and a game-theoretic solution concept known as \emph{coarse correlated equilibrium (CCE)}~\citep{Aumann74:Subjectivity,Moulin78:Strategically}. CCEs relax~\Cref{def:NE} by allowing for a \emph{correlated} distribution.

\begin{definition}
    A distribution $\mu \in \Delta(\cA_1 \times \dots \times \cA_n)$ is an $\epsilon$-CCE if for any player $i \in [n]$ and deviation $\vx_i' \in \cX_i$,
    \[
        \E_{\vec{a} \sim \mu }[u_i(\vx_i', \vec{a}_{-i} )] \leq \E_{\vec{a} \sim \mu }[u_i(\vec{a})] + \epsilon.
    \]
\end{definition}
By virtue of~\Cref{prop:FTRL-regret}, it follows that $\ftrl$ dynamics converge to the set of CCEs. Specifically, if the cumulative regret of each player is bounded as $O_T(\sqrt{T})$, convergence to the set of CCEs occurs at a rate of $O_\epsilon(1/\epsilon^2)$.

\paragraph{Potential games} A game is called a \emph{potential game} if there exists a function $\Phi: \cX \to \R$ such that for any player $i$, $\vx_{-i}$, and $\vx_i, \vx_i' \in \Delta(\cA_i)$,
\begin{equation*}
    u_i(\vx_i, \vx_{-i}) - u_i(\vx_i', \vx_{-i}) = \Phi(\vx_i, \vx_{-i}) - \Phi(\vx_i', \vx_{-i}).
\end{equation*}
This implies that the utility gradient of each player can be expressed as $\vu_i(\vx) = \nabla_{\vx_i} \Phi(\vx)$. In a normal-form potential game, the potential function is multilinear, so it is smooth, but not necessarily concave. Our upper bounds apply more broadly beyond normal-form games to constrained optimization whenever the potential is \emph{$L$-smooth}: $\| \nabla \Phi(\vx) - \nabla \Phi(\vx') \|_* \leq L \|\vx - \vx' \|$ for any $\vx, \vx'$, where $\|\cdot\|_*$ is the dual norm. For convenience, we define $L$-smoothness with respect to the global norm $\|\vx\| \defeq \sqrt{\sum_{i=1}^n \|\vx_i \|^2_{(i)} }$, where $\| \cdot \|_{(i)}$ is the norm induced by the regularizer $\mathcal{R}_i$ employed by player $i$. To ease the notation, we will assume that all players employ the same regularizer, although our results apply more broadly.

\paragraph{Further notation} We denote by $D_i$ the diameter of $\cX_i$ with respect to some norm $\|\cdot\|$; the choice of norm will be clear from the context. We use $B$ for an upper bound on the norm of utilities in terms of the dual norm; that is, $\|\vu_i(\vx) \|_* \leq B$ for any $i \in [n]$ and joint strategy $\vx$.
\section{Exponential upper bound and the potential property for FTRL}
\label{sec:upperbounds}

We begin by establishing exponential upper bounds for $\ftrl$ in potential games. Our first goal is to show that when multiple players perform an $\ftrl$ update simultaneously, the value of the potential function cannot decrease. The following lemma takes the perspective of a single player, and shows that $\ftrl$ always yields an improvement relative to its current utility.

\begin{restatable}{lemma}{FTRLimprovement}
    \label{lemma:FTRL-improvement}
If the sequence $(\vx_i^{(t)})_{t\geq 1}$ is updated using $\ftrl$
under a $1$-strongly convex regularizer $\mathcal{R}$ with respect to some norm $\|\cdot\|$,
\[
\langle \vx_i^{(t+1)}, \vu_i^{(t)} \rangle - \langle \vx_i^{(t)}, \vu_i^{(t)} \rangle \ge \frac{1}{\eta} \|\vx_i^{(t+1)} - \vx_i^{(t)} \|^2.
\]
\end{restatable}

In particular, the improvement is proportional to the squared movement $\|\vx_i^{(t+1)} - \vx_i^{(t)} \|$. The proof of~\Cref{lemma:FTRL-improvement} relies on the first-order optimality conditions of the $\ftrl$ update. Together with all the other proofs of this section, it is deferred to~\Cref{appendix:proofs-upperbounds}.

\Cref{lemma:FTRL-improvement} already implies that $\ftrl$ can never decrease the value of the potential function under \emph{alternating updates}, which means that players perform the update one by one; in fact, this holds no matter the choice of the learning rate.

We will now argue about the simultaneous case. Since $\Phi$ is $L$-smooth, we can lower bound the difference $\Phi(\vx^{(t+1)}) - \Phi(\vx^{(t)})$ by
\[
     \langle \nabla \Phi (\vx^{(t)}), \vx^{(t+1)} - \vx^{(t)} \rangle - \frac{L}{2} \|\vx^{(t+1)} - \vx^{(t)} \|^2.
\]
By~\Cref{lemma:FTRL-improvement}, the first term can in turn be lower bounded by
\[
    \sum_{i=1}^n \langle \vx_i^{(t+1)} - \vx_i^{(t)}, \vu_i^{(t)} \rangle \geq \frac{1}{\eta} \sum_{i=1}^n \|\vx_i^{(t+1)} - \vx_i^{(t)} \|^2.
\]
As a result, when $\eta \leq 1/L$, we have shown that
\[
    \Phi(\vx^{(t+1)}) - \Phi(\vx^{(t)}) \geq \frac{1}{2\eta} \|\vx^{(t+1)} - \vx^{(t)} \|^2.
\]
The telescopic summation yields the following implication.

\begin{proposition}
    \label{prop:pot-improvement}
    Suppose that each player in a potential game employs $\ftrl$ with learning rate $\eta \leq 1/L$, where $L$ is the smoothness parameter of the potential. For any $t \geq 1$,
    \begin{equation}
        \label{eq:pot-improvement}
        \Phi(\vx^{(t+1)}) - \Phi(\vx^{(t)}) \geq \frac{1}{2\eta} \|\vx^{(t+1)} - \vx^{(t)} \|^2.
    \end{equation}
    In particular,
    \[
        \sum_{t=1}^T \sum_{i=1}^n \|\vx_i^{(t+1)} - \vx_i^{(t)} \|^2 \leq 2 \eta \Phirange.
    \]
\end{proposition}
The second implication---the fact that $\ftrl$ has a bounded second-order path length---follows from~\eqref{eq:pot-improvement} from a telescopic summation, noting that $\Phirange$ denotes the range of the potential function, which is bounded.

An interesting implication of~\Cref{prop:pot-improvement} is that simultaneous $\ftrl$ asymptotically converges to a set in which the value of the potential is the same (\Cref{prop:constant-pot}). \Cref{prop:pot-improvement} also implies that $\lim_{t \to \infty} \|\vx^{(t+1)} - \vx^{(t)} \| \to 0 $. Specifically, for any $\epsilon > 0$, $O_\epsilon(1/\epsilon^2)$ iterations suffice so that $\|\vx^{(t+1)} - \vx^{(t)} \| \leq \epsilon$. For $\md$ with a smooth regularizer, this in fact implies that $\vx^{(t+1)}$ is an $O_\epsilon(\epsilon)$-Nash equilibrium~\citep{Anagnostides22:Last}; however, this is not the case for $\ftrl$, no matter the choice of the regularizer. In other words, $\ftrl$ can have \emph{spurious fixed points}.

\begin{example}[Spurious fixed points of $\ftrl$]
    \label{example:spuriousFPs}
    Consider a sequence of utilities $\vu_i^{(\tau)} = (1, 0)$ for all $\tau \in [t]$. Under this sequence, $\ftrl$ is producing a strategy that plays the first action deterministically, up to an error proportional to $1/t$. Now, if $\vu_i^{(t+1)} = (0, 1)$, it holds that $\|\vx_i^{(t+1)} - \vx_i^{(t)} \| \leq O_t(1/t)$, and consequently $\vx_i^{(t+1)}$ is highly suboptimal relative to $\vu_i^{(t+1)}$.
\end{example}

This inertia of $\ftrl$ is indeed what drives our lower bounds presented later in~\Cref{sec:lowerbounds}.

\subsection{Non-asymptotic convergence of no-regret dynamics}

Despite this pathological behavior, we will now establish an upper bound for $\ftrl$. In fact, our result applies to any no-regret dynamics, provided that the updates are alternating and $\epsilon$-lazy. To put this into context, we highlight that no-regret dynamics can fail to converge in potential games~\citep{Kleinberg09:Multiplicative}, so our positive results offer a natural way of bypassing that impossibility.

More precisely, let $\mathfrak{R}(\vu_i^{(1)}, \dots, \vu_i^{(t)}) \in \cX$ denote the output of the regret minimization algorithm upon observing the sequence $(\vu_i^{(\tau)})_{\tau=1}^t$, which we assume guarantees $\reg_i^{(t)} \leq O_t(t^{1-\alpha)})$ for some $\alpha \in (0, 1]$. If $\tvx_i^{(t+1)} = \mathfrak{R}(\vu_i^{(1)}, \dots, \vu_i^{(t)})$ at each time $t$, the strategy $\vx_i^{(t+1)}$ in the $\epsilon$-\emph{lazy} version is defined as
\begin{equation}
    \label{eq:lazy}
    \vx_i^{(t+1)} = 
\begin{cases} 
\tilde{\vx}_i^{(t+1)} & \text{if } \langle \tvx_i^{(t+1)} - \vx_i^{(t)} , \vu_i^{(t)} \rangle  \ge \epsilon, \\
\vx_i^{(t)} & \text{otherwise.}
\end{cases}
\end{equation}
That is, the update occurs only if it delivers at least an $\epsilon$ improvement relative to the current utility. In this context, we prove the following result.

\begin{restatable}{theorem}{lazynoregret}
    \label{theorem:positive-precise}
    Alternating $\epsilon$-lazy no-regret dynamics converge to a $2\epsilon$-Nash equilibrium after at most $\exp(O_\epsilon(1/\epsilon^2))$ rounds.
\end{restatable}

Our analysis proceeds by bounding the total number of updates that can occur (\Cref{lem:number_of_updates}) and the duration of the intervals between consecutive updates (\Cref{lem:epoch_recursion}). Specifically, we derive the recursion $T_{k+1} \leq T_k \left( 1 + O_\epsilon\left( \frac{1}{\epsilon} \right) \right)$. What is intriguing is that the mechanism driving our upper bound mirrors the construction of our lower bound.

\section{Exponential lower bound for FTRL}
\label{sec:lowerbounds}

In this section, we show that \ftrl can take exponentially many rounds to reach an approximate Nash equilibrium even in a two-player identical-interest games. Our lower bound applies to either alternating or simultaneous updates; in what follows, we consider simultaneous updates for concreteness. The construction is also robust in terms of using laziness per~\eqref{eq:lazy}, establishing the tightness of~\Cref{theorem:positive-precise} up to factors in the exponent. 

In what follows, we provide a high-level overview of the main steps in our argument; the formal proofs are deferred to \Cref{sec:proofs_lb}.
The class of games we study is based on the construction of \citet{Panageas23:Exponential}, originally introduced to analyze fictitious play. We employ a variant of their class of games, formally defined in~\cref{subsec:init}.
For an odd dimension $m$, we define a payoff matrix $\mA$ whose maximum entry equals $2m-1$, and we take the action sets to be $\cA_1=\cA_2=[m]$.
For each payoff value $k\in\{3,\dots,2m-1\}$, we denote by $a_1(k),a_2(k)\in[m]$ the unique row and column indices, respectively, such that $\mA[a_1(k),a_2(k)]=k$. An illustration of $\mat{A}$ is given in~\Cref{fig:matA}. The role of the extra row and column compared to $\mat{B}$ concerns the initialization; as we shall see, $\ftrl$ initialized uniformly at random on $\mat{A}$ will end up at the beginning of the spiral in the $\mat{B}$ submatrix.\footnote{For simplicity in the exposition, we allow the payoffs to be larger than 1. As we explain in~\Cref{rem:bounded_payoffs}, it is straightforward to adjust our construction even when all payoffs are in $[-1, 1]$.}

To examine the behavior of $\ftrl$ in this class of games, it is essential to keep track of the \emph{cumulative utility gap} of each action, defined as
\[
\gap{i}{t}{a_i}
\defeq
\max_{a_i' \in \cA_i} \sum_{\tau=1}^{t} \vu_i^{(\tau)}[a_i']
-
\sum_{\tau=1}^{t} \vu_i^{(\tau)}[a_i] \geq 0
\]
for each player $i \in \{1, 2\}$. An important observation is that for $\ftrl$ algorithms, an action with large gap will always be played with small probability, as we establish in the lemma below.

\begin{restatable}{lemma}{GapToProb}
\label{lem:regret_gap}
Suppose that $(\vx_i^{(t)})_{t \geq 1}$ is updated using $\ftrl$. If $\gap{i}{t}{a_i} \geq \gamma > 0$ for some action $a_i \in \cA_{i}$, then
\begin{equation}
    \label{eq:gap-prob}
    \vx_{i}^{(t+1)}[a_i] \leq \frac{\regcon}{\eta^{(t)} \gamma},
\end{equation}
where $R$ denotes the range of the regularizer.
\end{restatable}
An important point about~\Cref{lem:regret_gap} is that the bound relating the gap and the probability of playing the corresponding action in~\eqref{eq:gap-prob} becomes \emph{weaker} as the learning rate gets closer to $0$, which happens at a rate of $1/t^\alpha$. In the beginning, $\eta^{(t)}$ is relatively large, so the regularizer has a limited impact. As the learning rate decreases, the propensity of $\ftrl$ to mix over actions intensifies, introducing a considerable obstacle in the lower bound construction; it is harder to control the trajectory of $\ftrl$ under significant mixing. Our key observation is that the gap for the relevant actions will grow \emph{linearly} in $t$ (\Cref{lem:lower_bound_regret_gap}), thereby subsuming the vanishing effect of the learning rate in~\eqref{eq:gap-prob}.



\begin{figure}
\centering
\begin{subfigure}[t]{0.3\textwidth}
  \centering
  \includegraphics[width=\linewidth]{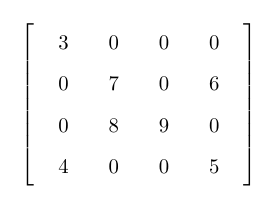}
  \caption{Original matrix $\mB $.}
  \label{fig:matB}
\end{subfigure}
\begin{subfigure}[t]{0.35\textwidth}
  \centering
  \includegraphics[width=\linewidth]{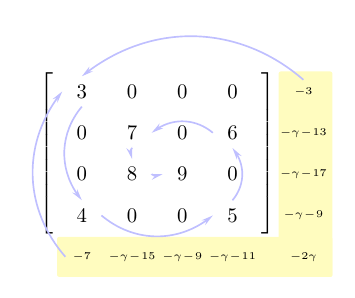}
  \caption{New matrix $\mA$.}
  \label{fig:matA}
\end{subfigure}
\end{figure}

\paragraph{Periods} The key invariance we establish is that the trajectory of \ftrl can be divided into $k$ \emph{periods}, so that in the $k$th period $\ftrl$ places high probability on the profile $(a_1(k),a_2(k))$. Because $\ftrl$ has a tendency to mix, especially under Legendre regularizers, our analysis needs to account for the residue probability. To do so, we fix the threshold $\delta \coloneqq \frac{1}{4m}$ and define the periods as follows.

\begin{restatable}[Transition period]{definition}{transperiod}
\label{def:transition_period}
For each integer $k \in \{3,4,\ldots,2m-1\}$, the \emph{transition period} $k$ is the set of rounds $t$ for which the following conditions hold:
\begin{enumerate}[label=(\roman*)]
\item If $k \ge 4$ is even, then for all $t \in [\tlow{k}, \thigh{k}]$,
\(
\vx_{2}^{(t)}[a_2(k)] \ge 1-\delta
\;
\text{and}
\;
\vx_{1}^{(t)}[a_1(k-1)] + \vx_{1}^{(t)}[a_1(k)] \ge 1-\delta
.
\)

\item If $k \ge 5$ is odd, then for all $t \in [\tlow{k}, \thigh{k}]$,
\(
\vx_{1}^{(t)}[a_1(k)] \ge 1-\delta
\;
\text{and}
\;
\vx_{2}^{(t)}[a_2(k-1)] + \vx_{2}^{(t)}[a_2(k)] \ge 1-\delta
.
\)
\end{enumerate}
\end{restatable}

For $k\ge 3$, we let $\tlow{k}$ and $\thigh{k}$ denote the first and last rounds of period~$k$, respectively. $T_k \coloneqq \thigh{k}-\tlow{k}+1$ denote its length, so that period~$k$ corresponds to the interval $t\in[\tlow{k},\thigh{k}]$.

A key property that we establish is that $\ftrl$ transitions from each period $k$ to its successor period $k+1$.

\begin{restatable}[Period consistency]{property}{MainProp}
\label{property:main}
Suppose both players employ \ftrl. If a round $t$ belongs to period $k$ for some $3 \le k \le 2m-1$, then the subsequent round $t+1$ belongs either to the same period $k$ or to the next period $k+1$.
\end{restatable}

In other words, there are no shortcuts: $\ftrl$ needs to traverse the entire path to reach the maximum payoff.

To argue about~\Cref{property:main}, let us say that $k$ is even. We need to analyze the transition of the row player from $a_1(k-1)$ to $a_1(k)$. When the row player starts playing $a_1(k)$ with higher and higher probability, it triggers a change in the utility observed by the column player. What we need to show is that the transition from $a_1(k-1)$ to $a_1(k)$ will be faster than the time it takes for the column player to react.

This is indeed the case for the following reason. The transition from $a_1(k-1)$ to $a_1(k)$ takes time proportional to $1/\eta^{(t)}$, which is of the order of $t^\alpha$. On the other hand, the time it takes for the column player to react is $\Omega(t)$. This is a high-level simplification; the precise argument is given in~\Cref{sec:valid_period}.

\paragraph{Growth of the cumulative utility gaps} Having established this key invariance, we show that the gap of actions corresponding to future periods grows considerably.

\begin{restatable}{lemma}{gaplb}
\label{lem:lower_bound_regret_gap}
Let $k \ge 6$ be even. Then
\[
\gap{1}{\thigh{k-2}}{a_1(k)}
\;\ge\;
\sum_{\ell = 4}^{k-2}
\Bigl[
(1-\delta)(\ell-1) - \delta(2m-1)
\Bigr] \, T_{\ell}.
\]
Similarly, if $k \ge 5$ is odd, then
\[
\gap{2}{\thigh{k-2}}{a_2(k)}
\;\ge\;
\sum_{\ell = 4}^{k-2}
\Bigl[
(1-\delta)(\ell-1) - \delta(2m-1)
\Bigr] \, T_{\ell}.
\]
\end{restatable}
When $\delta$ is small enough, we use~\Cref{lem:lower_bound_regret_gap} to derive the following recursion.

\begin{restatable}[Recurrence relation of $T_k + T_{k-1}$]{lemma}{recurs}
\label{lem:lower_bound_k_k1}
For a small enough $\delta > 0$,
\[
T_k + T_{k-1}
\;\ge\;
\frac{1}{2}
\sum_{\ell=4}^{k-2}
(\ell - 2)
\,
T_{\ell}
\]
for all $k \ge 6$.
\end{restatable}

The only way for $\ftrl$ to have a small Nash equilibrium gap under the invariance established in~\Cref{property:main} is for it to reach the last period (\Cref{lem:beneficial_deviation}). However, on account of~\Cref{lem:lower_bound_k_k1}, this takes $2^{\Omega(m \log m)}$ rounds, establishing our lower bound, as formalized in \Cref{theorem:negative_formal}.

\paragraph{Basis of the induction} The validity of our previous argument rests on $t$ being large enough, so that the linear gap $\Omega(t)$ subsumes factors that grow as $t^{\alpha}$. To complete the proof, we establish the basis of the induction. Specifically, we construct a suitable gadget that forces $\ftrl$ to i) immediately transition at the beginning of the spiral and ii) spend a considerable amount of time in the initial period. This can be achieved by appending an additional row and column to the matrix with sufficiently negative entries (\Cref{fig:matA}). When $\ftrl$ is initialized uniformly at random, which is the optimal choice for permutation-invariant regularizers (\Cref{lem:uniform_minimizer_regularizer}), all actions except the first will experience large regret. This means that escaping from the initial period will require many rounds so as to overcome the negative regret. Setting the negative entries in the extra row and column appropriately guarantees the desired bound on the length of the initial period,  as we formalize in~\Cref{sec:proofs_lb}.

\section{Doubly exponential lower bound for fictitious play}
\label{sec:multiplayer}

We now turn to multi-player potential games. We will prove an exponentially stronger lower bound for fictitious play ($\fictp$). 

Our construction is based on a connection we make to a graph-theoretic problem known as \emph{snake in the box}. This was first introduced by~\citet{Kautz58:Unit} in the context of coding theory, and has many interesting applications~\citep{Klee70:Maximum,Klee67:Method}. The basic version of the problem is to identify a path along the edges of a high-dimensional hypercube with the following property. The path begins at some vertex of the hypercube and traverses the edges---two vertices are connected if they differ by a single bit---to as many vertices as it can reach, subject to the constraint that every time it arrives at a new vertex, the previous one and \emph{all of its neighbors can never be used} going forward. Such a path is called a \emph{snake}. \Cref{fig:snake_path} (left) portrays an example for the 4-dimensional hypercube. 

This problem can be defined for any graph. In our context, we associate each joint action with a vertex in the graph, and two vertices are connected with an edge if there is a unilateral deviation that goes from one to the other. Under this mapping, \Cref{fig:snake_path} (right) shows a snake in a 3-player game where each player has 4 actions. Our key observation is that what essentially drove our lower bound in two-player games is the defining property of a snake. Indeed, the payoff matrix used in our two-player lower bound (\Cref{fig:matB}) induces a snake that spirals toward the center.

For our purposes, it suffices to restrict our attention to an $n$-dimensional hypercube. A well-known fact is that there exists a snake with exponential length~\citep{Abbott88:Snake,Evdokimov69:Maximal}, although characterizing the precise length is a major open problem in graph theory.

\begin{lemma}[\citealp{Evdokimov69:Maximal}]
    For any $n \in \N$, there exists a snake on the $n$-dimensional hypercube with length $C \cdot 2^{n}$, for some absolute constant $C > 0$.
\end{lemma}

\begin{figure*}
    \centering
    \includegraphics[scale=0.4]{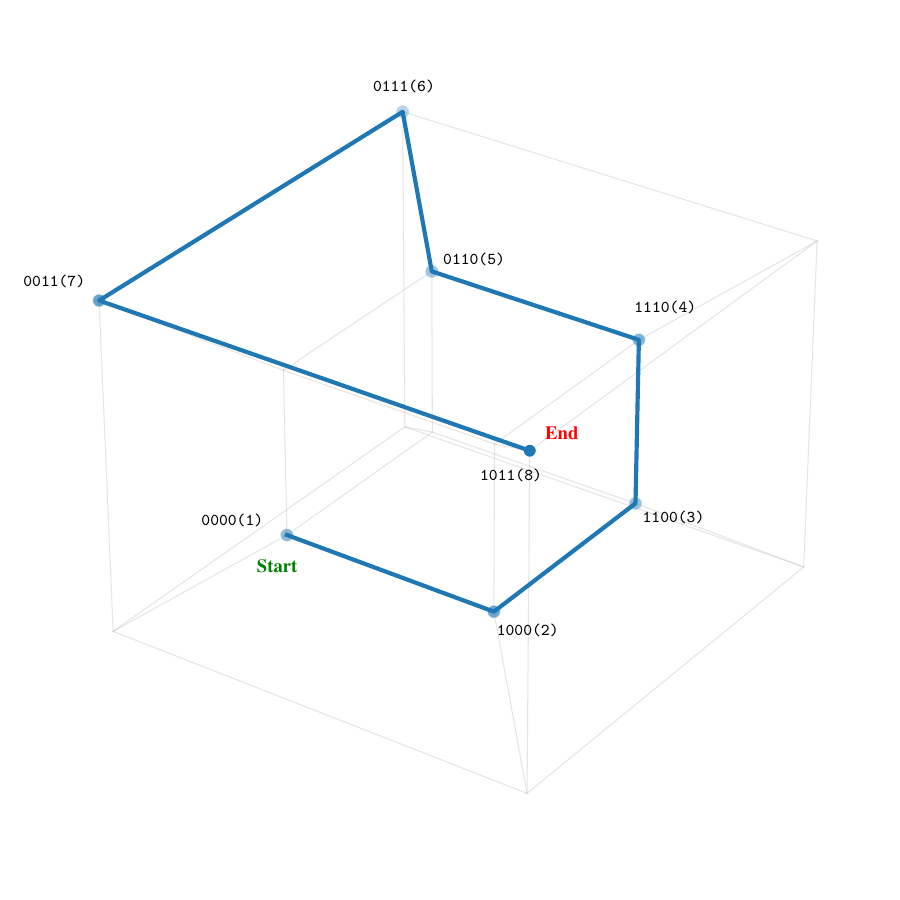}
    \includegraphics[scale=0.35]{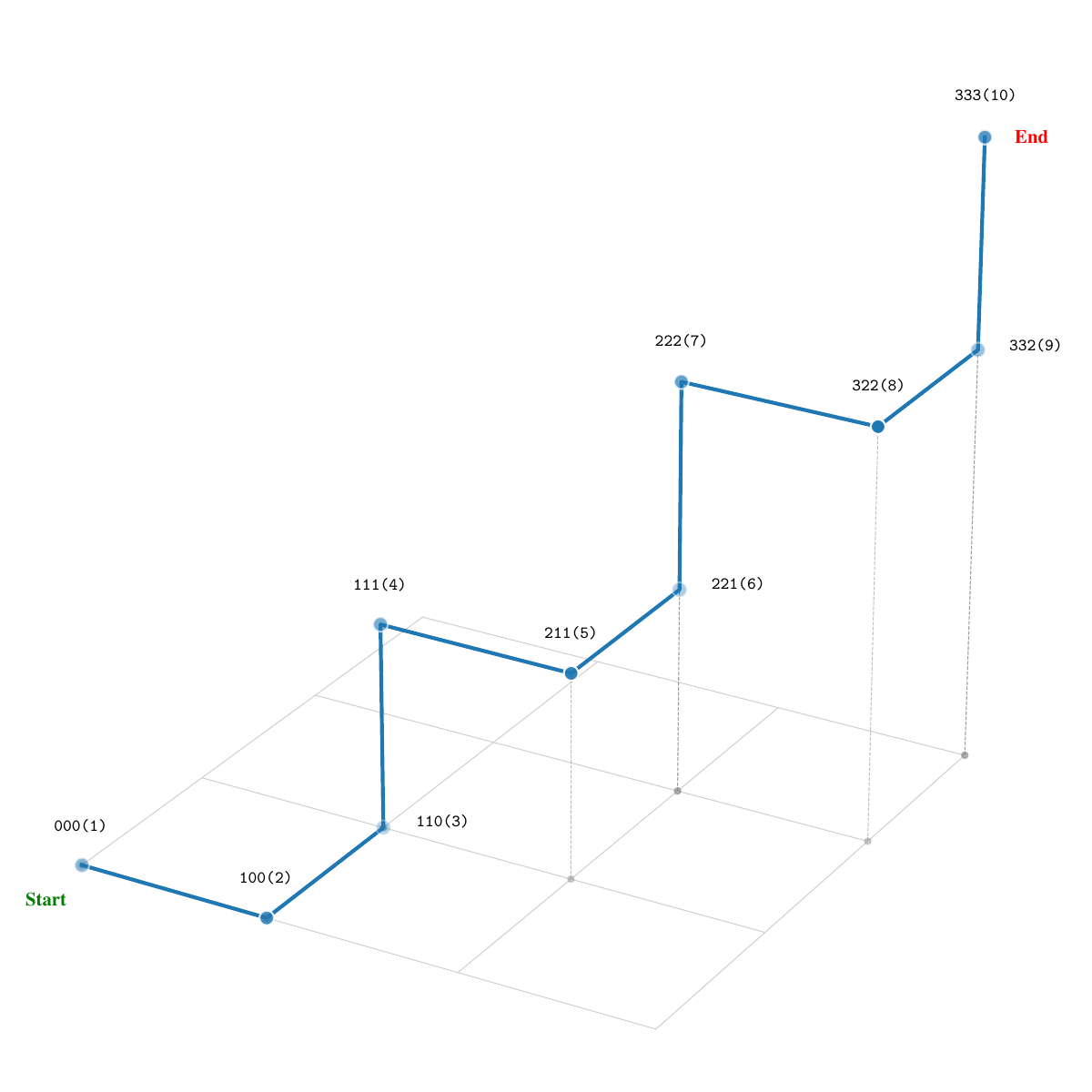}
    \caption{A snake on a 4-dimensional hypercube (left) and a snake corresponding to a 3-player 4-action game (right).}
    \label{fig:snake_path}
\end{figure*}

Let $P$ be such a snake in the $n$-dimensional hypercube. We proceed by embedding this path into a potential game. Specifically, we construct a binary-action $n$-player game with identical interests as follows. A joint action $(a_1, \dots, a_n)$ is in one-to-one correspondence with vertices of the hypercube. For a joint action $\vec{a} = (a_1, \dots, a_n) \in P$, we let $\idx(\vec{a})$ be the position of $\vec{a}$ on the path $P$, starting from $1$. We define the utility, for each player $i \in [n]$, as
\begin{equation}
    u_i(\vec{a}) = 
    \begin{cases}
        0 & \text{if } \vec{a} \notin P,\\
        \idx(\vec{a}) & \text{otherwise}.
    \end{cases}
\end{equation}

We assume that $\fictp$ is initialized at the start of the path $P$. By definition, it will always transition to pure strategies. Furthermore, because of the structure of the game, a basic invariance is that it will always remain on the path. Now, by construction, in every joint action on the path, there is a player who can improve the utility by at least 1 through a unilateral deviation, unless the joint action corresponds to the last vertex in the path, which is the global maximum of the utility function.

The crucial role of the snake property is this: at all iterations before reaching the end of the path, \emph{exactly} one player will have positive best-response gap. This is so because, by definition, all subsequent vertices except the immediate successor cannot be reached through a unilateral deviation. In other words, there are no shortcuts, and each edge traversal is caused by the movement of a single player.

If $T_k$ is the number of iterations that $\fictp$ spent on the action profile corresponding to payoff $k$, we establish the following recursion (the proof is in~\Cref{appendix:multiplayer}).

\begin{restatable}{lemma}{recursionsnake}
    \label{lemma:recursionsnake}
    For any $k \geq 2$, $T_{k} \geq (k -1) T_{k-1} + T_{k-2} + T_{k-3} - \sum_{\kappa = 1}^{k-4} \kappa T_\kappa$. ($T_{k}$ for $k \leq 0$ is to be interpreted as $0$.)
\end{restatable}

We draw attention to the fact that except the first three leading terms, the rest are negative. This is unlike the basic recursion that can be set up in two-player games, and happens because a player continually switches from one action to the other; in two-player games, a player always switches to a new action, which has accumulated negative regret throughout.

We next argue that, despite the negative terms, $T_k \geq (k-1)!$ for $k \geq 1$.

\begin{restatable}{lemma}{factorial}
Any sequence satisfying the recurrence relation of~\Cref{lemma:recursionsnake} grows at least as $(k-1)!$ for $k \geq 1$.
\end{restatable}

Combining, we arrive at the following \emph{doubly} exponential lower bound for $\fictp$ in potential games.

\begin{theorem}
    \label{theorem:FP-doubly}
    $\fictp$ requires at least $(C 2^{n})!$ to converge to a Nash equilibrium of an $n$-player potential game, where $C > 0$ is some absolute constant.
\end{theorem}
This holds even for converging to a $\epsilon$-Nash equilibrium with $\epsilon < 1$, although the payoffs in our construction can be exponentially large. By rescaling---$\fictp$ is scale invariant---it follows that doubly exponentially many iterations are needed to converge to an $\epsilon$-Nash equilibrium with $\epsilon \approx 1/2^n$. This closely matches our upper bound in~\Cref{theorem:positive-precise}.

\Cref{theorem:FP-doubly} improves exponentially over existing lower bounds for fictitious play in potential games~\citep{Panageas23:Exponential,Brandt13:Rate}.
\section{Conclusions and future research}

We have established that $\ftrl$, a celebrated online optimization algorithm, can take exponentially many iterations to converge to approximate Nash equilibria in potential games. We also showed a \emph{doubly} exponential lower bound for fictitious play in multi-player games. An interesting question is whether a doubly exponential lower bound applies to $\ftrl$ when the learning rate is small; the primary challenge is that it is significantly harder to hide the path---as we did for fictitious play---when players are mixing.

\section*{Acknowledgments}

Ioannis Panageas is supported by NSF grant CCF-
2454115. Tuomas Sandholm is supported by NIH award A240108S001, the Vannevar Bush Faculty Fellowship ONR N00014-23-1-2876, and National Science Foundation grant RI-2312342. Ioannis Anagnostides is grateful to Emanuel Tewolde and Brian Hu Zhang for helpful discussions at the initial stage of this project.

\bibliography{main}

\clearpage
\appendix

\section{Proofs from Section~\ref{sec:upperbounds}}
\label{appendix:proofs-upperbounds}

This section contains the proofs omitted from~\Cref{sec:upperbounds}. We begin with the one-step improvement property of $\ftrl$.

\FTRLimprovement*

\begin{proof}
Let $\psi^{(t)}(\vx_i) = \left\langle \vx_i, \sum_{\tau=1}^t \vu_i^{(\tau)} \right\rangle - \frac{1}{\eta} \mathcal{R}(\vx_i)$. By definition, $\vx_i^{(t+1)}$ maximizes $\psi^{(t)}(\vx_i)$ with respect to $\mathcal{X}$. The first-order optimality condition implies
\[
\langle \nabla \psi^{(t)}(\vx_i^{(t+1)}), \vx_i^{(t+1)} - \vx_i' \rangle \ge 0 \quad \forall \vx_i' \in \mathcal{X}_i.
\]
Substituting the gradient $\nabla \psi^{(t)}(\vx_i) = \sum_{\tau=1}^t \vu_i^{(\tau)} - \frac{1}{\eta} \nabla \mathcal{R}(\vx_i)$, we have
\begin{equation} \label{eq:opt_t+1_new}
\left\langle \sum_{\tau=1}^t \vu_i^{(\tau)} - \frac{1}{\eta} \nabla \mathcal{R}(\vx_i^{(t+1)}), \vx_i^{(t+1)} - \vx_i' \right\rangle \ge 0 \quad \forall \vx_i' \in \mathcal{X}_i.
\end{equation}
Similarly, the optimality condition for $\vx_i^{(t)}$ (which maximizes $\psi^{(t-1)}$) implies
\begin{equation} \label{eq:opt_t_new}
\left\langle \sum_{\tau=1}^{t-1} \vu_i^{(\tau)} - \frac{1}{\eta} \nabla \mathcal{R}(\vx_i^{(t)}), \vx_i^{(t)} - \vx_i' \right\rangle \ge 0 \quad \forall \vx_i' \in \mathcal{X}_i.
\end{equation}
Substituting $\vx_i' = \vx_i^{(t)}$ in \eqref{eq:opt_t+1_new} and $\vx_i' = \vx_i^{(t+1)}$ in \eqref{eq:opt_t_new},
\begin{align}
\label{eq:ineq1_new}
\left\langle \sum_{\tau=1}^t \vu_i^{(\tau)} - \frac{1}{\eta} \nabla \mathcal{R}(\vx_i^{(t+1)}), \vx_i^{(t+1)} - \vx_i^{(t)} \right\rangle &\ge 0, \\
\label{eq:ineq2_new}
\left\langle \sum_{\tau=1}^{t-1} \vu_i^{(\tau)} - \frac{1}{\eta} \nabla \mathcal{R}(\vx_i^{(t)}), \vx_i^{(t)} - \vx_i^{(t+1)} \right\rangle &\ge 0.
\end{align}
Summing \eqref{eq:ineq1_new} and \eqref{eq:ineq2_new} yields
\[
\left\langle \sum_{\tau=1}^t \vu_i^{(\tau)}, \vx_i^{(t+1)} - \vx_i^{(t)} \right\rangle + \left\langle \sum_{\tau=1}^{t-1} \vu_i^{(\tau)}, \vx_i^{(t)} - \vx_i^{(t+1)} \right\rangle 
- \frac{1}{\eta} \left\langle \nabla \mathcal{R}(\vx_i^{(t+1)}) - \nabla \mathcal{R}(\vx_i^{(t)}), \vx_i^{(t+1)} - \vx_i^{(t)} \right\rangle \ge 0.
\]
Since
\[
\left\langle \sum_{\tau=1}^t \vu_i^{(\tau)}, \vx_i^{(t+1)} - \vx_i^{(t)} \right\rangle + \left\langle \sum_{\tau=1}^{t-1} \vu_i^{(\tau)}, \vx_i^{(t)} - \vx_i^{(t+1)} \right\rangle = \langle \vu_i^{(t)}, \vx_i^{(t+1)} - \vx_i^{(t)} \rangle,
\]
we have
\[
\langle \vu_i^{(t)}, \vx_i^{(t+1)} - \vx_i^{(t)} \rangle 
- \frac{1}{\eta} \left\langle \nabla \mathcal{R}(\vx_i^{(t+1)}) - \nabla \mathcal{R}(\vx_i^{(t)}), \vx_i^{(t+1)} - \vx_i^{(t)} \right\rangle \ge 0.
\]
Since $\mathcal{R}$ is a 1-strongly convex function, 
\[
\left\langle \nabla \mathcal{R}(\vx_i^{(t+1)}) - \nabla \mathcal{R}(\vx_i^{(t)}), \vx_i^{(t+1)} - \vx_i^{(t)} \right\rangle \geq \|\vx_i^{(t+1)} - \vx_i^{(t)} \|^2,
\]
and the claim follows.
\end{proof}

We next leverage~\Cref{prop:pot-improvement} to establish that simultaneous $\ftrl$ converges to a limit set in which the potential has the same value. We refer to~\citet{Losert83:Dynamics} for a related argument.

\begin{proposition}
    \label{prop:constant-pot}
In an $n$-player potential game, the $\omega$-limit set of simultaneous $\ftrl$ has the same potential value.
\end{proposition}

\begin{proof}
By~\Cref{prop:pot-improvement}, we know that
\[
\Phi(\vec{x}^{(t+1)}) - \Phi(\vec{x}^{(t)}) \ge \frac{1}{2 \eta} \sum_{i=1}^n \| \vec{x}_i^{(t+1)} - \vec{x}_i^{(t)} \|^2 \ge 0.
\]
Since $\Phi$ is bounded and nondecreasing along the trajectory, the sequence $\{\Phi(\vec{x}^{(t)})\}$ converges to a limit $\Phi^*$.
Let $\Omega$ be the $\omega$-limit set---the set of accumulation points---of the sequence $\{\vec{x}^{(t)}\}$. Since $\mathcal{X}$ is compact, $\Omega$ is nonempty. By continuity, $\Phi(\vec{x}) = \Phi^*$ for all $\vec{x} \in \Omega$.
Because the potential is constant on $\Omega$, the improvement at every step must be zero. By~\Cref{lemma:FTRL-improvement}, this implies that any $\vec{x} \in \Omega$ is a \textit{fixed point} of $\ftrl$.
\end{proof}

Furthermore, it is worth pointing out that if the limit point exists, it is a Nash equilibrium. Indeed, let us assume that $\vx = \lim_{t \to \infty} \vx^{(t)}$, so $\ftrl$ converges pointwise. For the sake of contradiction, suppose that $\vx$ is not a Nash equilibrium. This means that there exists player $i$, deviation $\vx_i^*$, and $\epsilon > 0$ such that
\begin{equation}
    \label{eq:subopt}
    \langle \nabla_{\vec{x}_i} \Phi(\vec{x}), \vec{x}_i^* \rangle > \langle \nabla_{\vec{x}_i} \Phi(\vec{x}), \vec{x}_i \rangle + \epsilon.
\end{equation}
Since $\Phi$ is continuously differentiable, $\lim_{t \to \infty} \vu_i^{(t)} = \nabla_{\vx_i} \Phi(\vx)$. This also implies $\lim_{t \to \infty} \frac{1}{t} \sum_{\tau=1}^t \vu_i^{(\tau)} = \nabla_{\vx_i} \Phi(\vx)$. Now, the $\ftrl$ update for player $i$ at time $t+1$ can be expressed as
\[
\vec{x}_i^{(t+1)} = \argmax_{\vec{x}_i' \in \mathcal{X}_i} \left( \left\langle \vec{x}_i', \frac{1}{t} \sum_{\tau=1}^t \vec{u}_i^{(\tau)} \right\rangle - \frac{1}{\eta t} \mathcal{R}_i(\vec{x}_i') \right).
\]
By Berge's maximum theorem, this implies
\[
    \vx_i = \lim_{t \to \infty} \vx_i^{(t)} \in \argmax_{\vx' \in \cX_i} \langle \vx_i', \nabla_{\vx_i} \Phi(\vx) \rangle,
\]
which is a contradiction to~\eqref{eq:subopt}. 

In fact, this argument connecting pointwise convergence to Nash equilibria is significantly more general, and holds under any no-regret algorithm and general-sum game.

We continue with the proof of~\Cref{theorem:positive-precise}, which is recalled below.

\lazynoregret*

\begin{proof}
We first point out that the total number of updates is bounded by $O_\epsilon(1/\epsilon)$. This follows directly from the potential property, the fact that the updates are alternating, and the definition of $\epsilon$-lazy updates.

\begin{lemma}[Number of updates]
\label{lem:number_of_updates}
Let $K$ denote the total number of times an update occurs under $\epsilon$-lazy alternating no-regret dynamics, that is, $\vx_i^{(t+1)} \neq \vx_i^{(t)}$ for some player $i$. Then
\[
K \le \frac{\Phirange }{\epsilon}.
\]
\end{lemma}

What remains is to bound the number of rounds in between two consecutive updates. 

\begin{lemma}[Time between updates]
\label{lem:epoch_recursion}
Let $T_k$ be the time index of the $k$th update, which is assumed to be large enough so that $\reg_i^{(t)} \leq \frac{1}{2} \epsilon t$ for all $t \geq T_k$. Suppose that during $[T_k, T_{k+1}-1]$, the strategy is fixed at $\vx^{(T_k)}$ and is not an $2 \epsilon$-Nash equilibrium. Then there exists a game-dependent parameter $P > 0$ such that
\[
T_{k+1} \le T_k \left( 1 + \frac{P}{\epsilon} \right).
\]
\end{lemma}

\begin{proof}
Let $t < T_{k+1}$. By definition of the no-regret property,
\[
    \sum_{\tau=1}^t \langle \vu_i^{(\tau)}, \vx_i' - \tvx_i^{(\tau)} \rangle \leq \reg_i^{(t)}.
\]
Bounding the difference in the first $T_k$ rounds as $\sum_{\tau=1}^{T_k} \langle \vu_i^{(\tau)}, \vx_i' - \tvx_i^{(\tau)} \rangle \leq \sum_{\tau=1}^{T_k} \| \vx_i' - \tvx_i^{(\tau)} \| \| \vu_i^{(t)} \|_* \leq T_k B D_i$, we have
\[
    \sum_{\tau = T_k+1}^t \langle \vx_i' - \tvx_i^{(\tau)}, \vu_i^{(T_k)} \rangle \leq \reg_i^{(t)} + T_k B D_i.
\]
Here we used the fact that for $T_k < t < T_{k+1}$ no updates occur, so $\vu_i^{(t)} = \vu_i^{(T_k)}$. We now write
\[
    (t-T_k) \langle \vx_i' - \vx_i^{(T_k)}, \vu_i^{(T_k)}  \rangle + \sum_{\tau=T_k+1}^t \langle \vx_i^{(T_k)} - \tvx_i^{(\tau)}, \vu_i^{(T_k)} \rangle \leq T_k B D_i + \reg_i^{(t)}.
\]
Since $\vx_i^{(\tau)}$ is not getting updated before $T_{k+1}$, it follows that $\langle \vx_i^{(\tau)}, \vu_i^{(T_k)} \rangle \geq \langle \tvx_i^{(\tau)}, \vu_i^{(T_k)} \rangle - \epsilon$. Furthermore, since $\vx_i^{(T_k)}$ is not a $2\epsilon$-Nash equilibrium, we have
\[
    \epsilon (t - T_k) \leq T_k B D_i + \reg_i^{(t)}.
\]
Using the fact that $\reg_i^{(t)} \leq \frac{1}{2} \epsilon t$ and rearranging concludes the proof.
\end{proof}

In other words, we have shown that
\[
    T_{k+1} \leq T_k \left( 1 + O_\epsilon \left( \frac{1}{\epsilon} \right) \right).
\]
So,
\[
    T_{k} \leq \left( 1 + O_\epsilon \left( \frac{1}{\epsilon} \right) \right)^k \leq \exp\left( k   O_\epsilon \left( \frac{1}{\epsilon} \right) \right).
\]

Combining with~\Cref{lem:number_of_updates}, the total number of iterations needed is $\exp\left( O_\epsilon \left( \frac{1}{\epsilon^2} \right) \right)$.
\end{proof}

\section{Properties of the regularizer}
\label{sec:regularizer_properties}

This section establishes basic properties of permutation-invariant regularizers. 

\begin{definition}[Permutation invariance]
\label{def:permutation_invariance}
A function $\cR: \Delta_m \to \mathbb{R}$ is called \emph{permutation invariant} if for every permutation $\pi : [m] \to [m]$ and every $\vx \in \Delta_m$,
\[
\cR(\vx) = \cR(\pi(\vx)),
\]
where $\pi(\vx)[i] := \vx[\pi(i)]$.
\end{definition}

\begin{lemma}[Permutation-invariant function has uniform minimizer on the simplex]
\label{lem:uniform_minimizer_regularizer}
Let $\cR : \Delta_m \to \mathbb{R}$ be a convex, permutation-invariant function. Then the uniform distribution
\(
\frac{1}{m} \mathbf{1}
\)
is a minimizer of $\cR$ over $\Delta_m$. If $\cR$ is strictly convex, $\frac{1}{m} \mathbf{1}$ is the unique minimizer.
\end{lemma}

\begin{proof}
Let $\xstar \in \arg\min_{\vx \in \Delta_m} \cR(\vx)$. By permutation invariance (\cref{def:permutation_invariance}), for any permutation $\pi$,
\[
\cR(\pi(\xstar)) = \cR(\xstar),
\]
so $\pi(\xstar)$ is also a minimizer.

Let $\perm{m}$ denote the set of all permutations of $\{1,\dots,m\}$. By convexity of $\cR$ and Jensen's inequality,
\[
\cR
\Biggl( 
\frac{1}{| \perm{m} |} 
\sum_{\pi \in \perm{m}} \pi(\xstar) 
\Biggr)
\le 
\frac{1}{| \perm{m} |} 
\sum_{\pi \in \perm{m}} \cR(\pi(\xstar)) 
= \cR(\xstar).
\]

The average over all permutations of $\xstar$—and more generally of any vector in $\Delta_m$—is the uniform vector $\tfrac{1}{m}\mathbf{1}$. To see this, define
\[
\bar{\vx} 
\defeq 
\frac{1}{|\perm{m}|} 
\sum_{\pi \in \perm{m}} \pi(\xstar),
\]
and denote its $j$th coordinate by $\bar{\vx}[j]$. Decompose $\perm{m}$ as
\(
\perm{m}
=
\bigcup_{i=1}^{m} 
\perm{m}^{(i)},
\)
where $\perm{m}^{(i)}$ is the set of permutations mapping coordinate $j$ to $i$. Clearly, $|\perm{m}^{(i)}| = (m-1)!$ for each $i \in [m]$. Hence,
\begin{align*}
\bar{\vx}[j] 
&= 
\frac{1}{|\perm{m}|} 
\sum_{\pi \in \perm{m}}
\pi(\xstar)[j] 
= 
\frac{1}{|\perm{m}|} 
\sum_{i=1}^{m} |\perm{m}^{(i)}| \, \xstar[i] 
\\
&= 
\frac{(m-1)!}{m!} 
\sum_{i=1}^{m} \xstar[i] 
= 
\frac{1}{m}.
\end{align*}

Consequently, $\bar{\vx} = \frac{1}{m}\mathbf{1}$ and 
\(
\cR(\bar{\vx}) \le \cR(\vx^\star),
\)
so $\bar{\vx}$ is also a minimizer. If $\cR$ is strictly convex, equality holds only if $\vx^\star = \bar{\vx} =  \frac{1}{m} \mathbf{1}$.
\end{proof}

\begin{lemma}[\ftrl with a permutation-invariant regularizer preserves order]
\label{lem:order_preservation}
Let $\cR : \Delta(\cA_i) \to \mathbb{R}$ be a permutation-invariant, strictly convex, and differentiable regularizer, and let $\eta > 0$. Then, for any actions $a_i,a_i' \in \cA_i$ and any round $t$, it holds that
\[
\sum\limits_{\tau=1}^{t} \vu_{i}^{(\tau)}[a_i] 
\ge 
\sum\limits_{\tau=1}^{t} \vu_{i}^{(\tau)}[a_i'] 
\quad 
\Longleftrightarrow
\quad 
\vx_i^{(t+1)}[a_i] 
\ge 
\vx_i^{(t+1)}[a_i'].
\]
\end{lemma}
\begin{proof}
Since $\cR$ is strictly convex and differentiable, the \ftrl optimizer $\vx_1^{(t+1)}$ is unique and satisfies the first-order condition
\[
\nabla \cR(\vx_i^{(t+1)}) = \eta
\sum\limits_{\tau=1}^{t} \vu_i^{(\tau)} + \lambda \mathbf{1}
\]
for some scalar $\lambda$. Hence, for any actions $a_i,a_i' \in \cA_i$,
\begin{equation}
\frac{\partial \cR(\vx_{i}^{(t+1)}) }{\partial \vx_i[a_i] }
-
\frac{\partial \cR(\vx_{i}^{(t+1)}) }{\partial \vx_i[a_i']}
=
\eta \left(
\sum\limits_{\tau=1}^{t} \vu_i^{(\tau)}[a_i] 
-
\sum\limits_{\tau=1}^{t} \vu_i^{(\tau)}[a_i']
\right).
\label{eq:order_preservation}
\end{equation}

Suppose 
\(
\sum_{\tau=1}^{t} \vu_i^{(\tau)}[a_i] 
\ge
\sum_{\tau=1}^{t} \vu_i^{(\tau)}[a_i'].
\)
Then \eqref{eq:order_preservation} gives 
\begin{equation}
\label{eq:order_eq1}
\frac{\partial \cR(\vx_{i}^{(t+1)}) }{\partial \vx_i[a_i]}
\ge
\frac{\partial \cR(\vx_{i}^{(t+1)}) }{\partial \vx_i[a_i']}
.
\end{equation}
We claim this implies 
\[
\vecx{i}{t+1}[a_i] \geq \vecx{i}{t+1}[a_i']
.
\]

For the sake of contradiction, assume $\vx_i^{(t+1)}[a_i] < \vx_i^{(t+1)}[a_i']$.
Let $\pi$ be the permutation that swaps $a_i$ and $a_i'$, and set $\vy \defeq \pi(\vx_i^{(t+1)})$.
By permutation invariance, $\cR(\vy)= \cR(\pi(\vx_i^{(t+1)})) = \cR(\vx_i^{(t+1)})$; moreover, the gradient operator is linear, which implies that the corresponding partial derivatives are swapped:
\begin{equation}
\frac{\partial \cR(\vy)}{\partial \vx_i[a_i]}
=
\frac{\partial \cR(\vx_i^{(t+1)})}{\partial \vx_i[a_i']},
\qquad
\frac{\partial \cR(\vy)}{\partial \vx_i[a_i']}
=
\frac{\partial \cR(\vx_i^{(t+1)})}{\partial \vx_i[a_i]}.
\label{eq:ordering_eq1}
\end{equation}
Since $\cR$ is strictly convex, we have
\begin{equation}
\langle \nabla \cR(\vx_i^{(t+1)})-\nabla \cR(\vy),\, \vx_i^{(t+1)}-\vy\rangle > 0
\quad
\text{for}
\;\;
\vecx{i}{t+1} \neq \vy.
\label{eq:ordering_eq2}
\end{equation}
By definition of $\vy$, it differs from $\vx_i^{(t+1)}$ only in the coordinates $a_i$ and $a_i'$. Therefore, using \eqref{eq:ordering_eq1}, the inner product in \eqref{eq:ordering_eq2} reduces to
\begin{equation}
2\left(
\frac{\partial \cR(\vx_i^{(t+1)})}{\partial \vx_i[a_i]}
-
\frac{\partial \cR(\vx_i^{(t+1)})}{\partial \vx_i[a_i']}
\right)
\left(
\vx_i^{(t+1)}[a_i]-\vx_i^{(t+1)}[a_i']
\right)
>0.
\label{eq:ordering_val_eq1}
\end{equation}
We note that the case $\vx=\vy$ is trivial, since the claim follows immediately from \eqref{eq:ordering_val_eq1}.
Under our assumption $\vx_i^{(t+1)}[a_i]-\vx_i^{(t+1)}[a_i']<0$, the above inequality forces
\[
\frac{\partial \cR(\vx_i^{(t+1)})}{\partial \vx_i[a_i]}
-
\frac{\partial \cR(\vx_i^{(t+1)})}{\partial \vx_i[a_i']}
<0,
\]
contradicting \eqref{eq:order_eq1}.
Therefore, $\vx_i^{(t+1)}[a_i] \ge \vx_i^{(t+1)}[a_i']$.

Conversely, suppose that $\vx_i^{(t+1)}[a_i] \ge \vx_i^{(t+1)}[a_i']$ but
\(
\sum_{\tau=1}^{t} \vu_i^{(\tau)}[a_i]
<
\sum_{\tau=1}^{t} \vu_i^{(\tau)}[a_i'].
\)
Then \eqref{eq:order_preservation} implies
\[
\frac{\partial \cR(\vx_i^{(t+1)})}{\partial \vx_i[a_i]}
<
\frac{\partial \cR(\vx_i^{(t+1)})}{\partial \vx_i[a_i']},
\]
which, by the same argument as above, forces $\vx_i^{(t+1)}[a_i] < \vx_i^{(t+1)}[a_i']$, contradicting the assumption.
Combining both directions establishes the equivalence.
\end{proof}
\section{Proofs from Section~\ref{sec:lowerbounds}}
\label{sec:proofs_lb}

We begin with the following simple but crucial observation. It shows that any pure strategy whose \emph{cumulative utility} up to round~$t$ is smaller than the maximal \emph{cumulative utility} by at least a constant~$\gamma$ is assigned, in round $t+1$, a probability that is inversely proportional to~$ \eta^{(t)} \gamma$. To simplify the exposition, we introduce the \emph{cumulative utility gap}, which we henceforth refer to as the \emph{gap}.

\begin{definition}[Cumulative utility gap]
\label{def:gap}
For a player~$i$, the gap of an action $a_i \in \cA_i$ at time~$t$ is
\[
\gap{i}{t}{a_i}
\defeq
\max_{a_i' \in \cA_i} \sum_{\tau=1}^{t} \vu_i^{(\tau)}[a_i']
-
\sum_{\tau=1}^{t} \vu_i^{(\tau)}[a_i].
\]
\end{definition}




\GapToProb*

\begin{proof}
Define
\[
\vx_i' \coloneqq \vx_i^{(t+1)} - \vx_i^{(t+1)}[a_i]\vec{e}_{a_i}
+ \vx_i^{(t+1)}[a_i]\vec{e}_{a_i'},
\]
where
\[
a_i' \in \argmax_{a \in \cA_{i}} \sum_{\tau=1}^{t} \vu_{i}^{(\tau)}[a],
\]
and $\vec{e}_{a}$ denotes the $a$th standard basis vector.
Equivalently, $\vx_i'$ differs from $\vx_i^{(t+1)}$ only in the coordinates $a_i$ and $a_i'$, with
$\vx_i'[a_i]=0$ and
\(
\vx_i'[a_i'] = \vx_i^{(t+1)}[a_i]+\vx_i^{(t+1)}[a_i'].
\)
In particular, $\vx_i' \in \Delta(\cA_i)$.

Since $\gap{i}{t}{a_i} \ge \gamma$, we obtain
\begin{align}
\biggl\langle 
\vx_i', \sum_{\tau=1}^{t} \vu_i^{(\tau)} 
\biggr\rangle
-
\biggl\langle 
\vx_i^{(t+1)}, \sum_{\tau=1}^{t} \vu_i^{(\tau)}
\biggr\rangle
&= 
\langle 
\vx_i' - \vx_{i}^{(t+1)}, \sum_{\tau=1}^{t} \vu_i^{(\tau)}
\rangle 
\notag 
\\
&= 
\vx_{i}^{(t+1)}[a_i]
\left(
\sum_{\tau=1}^{t} \vu_{i}^{(\tau)}[a_{i}'] 
-
\sum_{\tau=1}^{t} \vu_{i}^{(\tau)}[a_i]
\right) 
\notag
\\
&= 
\vx_i^{(t+1)}[a_i]
\;
\gap{i}{t}{a_i}
\notag 
\\[3pt]
&\ge 
\vx_{i}^{(t+1)}[a_i]
\;
\gamma.
\label{eq:gap_eq1}
\end{align}
Moreover, the regularizer has bounded range, so
\(
\cR(\vx_i') \le \cR(\vx_i^{(t+1)}) + \regcon.
\)
On the other hand, $\vx_i^{(t+1)}$ is the \ftrl maximizer of the objective
\(
\langle \vx_i, \sum_{\tau=1}^{t} \vu_i^{(\tau)} \rangle - \frac{1}{\eta^{(t)}}\cR(\vx_i),
\)
and therefore its value is at least that of  feasible $\vx_i'$:
\begin{equation}
\left\langle 
\vx_{i}^{(t+1)}, \sum_{\tau=1}^{t} \vu_{i}^{(\tau)}
\right\rangle
-
\frac{1}{\eta^{(t)}} \cR(\vx_{i}^{(t+1)})
\;\ge\;
\left\langle 
\vx_{i}', \sum_{\tau=1}^{t} \vu_{i}^{(\tau)}
\right\rangle
-
\frac{1}{\eta^{(t)}} \cR(\vx_{i}').
\label{eq:gap_eq2}
\end{equation}

Combining \eqref{eq:gap_eq1} and \eqref{eq:gap_eq2} with the bound on $\cR(\vx_i')$ yields 
\[ 
\vx_{i}^{(t+1)}[a_i] 
\le
\frac{\regcon}{\eta^{(t)} \gamma}, 
\]
which implies the claimed bound.
\end{proof}

\subsection{Definition of transition period}

As explained in the main body, we partition the \ftrl dynamics into \emph{periods}. 
A period~$k$ is a block of consecutive rounds during which \ftrl concentrates its probability mass on a small set of strategies---the \emph{competing actions}---while all remaining actions receive negligible mass. 
We fix the threshold $\delta \coloneqq \frac{1}{4m}$ and define periods accordingly. 
We now recall the basic setup.

For $k\ge 3$, let $\tlow{k}$ and $\thigh{k}$ denote the first and last rounds of period~$k$, respectively, and let $T_k \coloneqq \thigh{k}-\tlow{k}+1$ denote its length, so that period~$k$ corresponds to the interval $t\in[\tlow{k},\thigh{k}]$.

\transperiod*

Since the actions $a_1(2)$ and $a_2(2)$ are undefined, the above definition does not apply to $k=3$. We therefore let period~$3$ consist only of the first round, \textit{i.e.}, $\tlow{3}=\thigh{3}=1$, and let period~$4$ start immediately thereafter. The remainder of the analysis relies on \cref{property:main}, which we prove by induction in \cref{sec:valid_period}.

\MainProp*

\subsection{Uniform initialization}
\label{subsec:init}

We begin with the initialization, which serves as the base case for the induction in the proof of \cref{property:main}.
A standard initialization for \ftrl is the uniform strategy. In particular, for permutation-invariant regularizers, the first update ($t=1$) selects a minimizer of $\cR$, which is the uniform distribution.
Indeed, at initialization the cumulative utility gaps satisfy
\[
\gap{i}{0}{a_i}=0
\quad
\text{for all } i\in\{1,2\} \text{ and } a_i\in\cA_i.
\]
Thus, at $t=1$ the \ftrl update reduces to minimizing $\cR$; by \cref{lem:uniform_minimizer_regularizer}, we obtain
\[
\vx_1^{(1)}=\mathrm{Uniform}(\cA_1),
\qquad
\vx_2^{(1)}=\mathrm{Uniform}(\cA_2).
\]

As we mentioned in \cref{sec:lowerbounds}, our construction builds on the class of games introduced by \citet{Panageas23:Exponential}. For completeness, \cref{sec:matrix_props} reviews their recursive construction of the matrix $\mB_{m,r}$ (\textit{cf.}~\eqref{mat:payoffB} in \cref{def:payoffB}) and establishes the structural properties we require.
Although this class is well suited for proving exponential lower bounds for fictitious play, it is not directly applicable to \ftrl: under uniform initialization, running \ftrl on $(\mB_{m,r},\mB_{m,r})$ for any $r$ quickly identifies the actions attaining the maximum payoff and converges to that profile within a constant number of rounds.
To preclude this behavior, we modify the payoff matrix as described below.

The underlying idea is straightforward. By \cref{lem:regret_gap}, the gap
\(
\gap{1}{1}{a_1} \;(\text{resp. } \gap{2}{1}{a_2})
\) 
at round 1 controls the probability assigned to action $a_1$ (resp. $a_2$) at round 2.  
Hence, by ensuring that a sufficiently large gap is incurred for all actions except $(1,1)$ after round $1$, we can guarantee that at round 2 the only action profile played with probability nearly 1 (up to $\delta$) is $(1,1)$.

For an odd dimension $m$, we consider the matrix $\mB = \mB_{m-1,2}$ of size $(m-1) \times (m-1)$ as defined in 
\cref{def:payoffB}. Starting from the matrix $\mB$, we embed it into a larger matrix
\[
\mA \in [-\gamma - (4m - 1),\, 2m - 1]^{m \times m},
\]
which is obtained by adding one extra row and one extra column to $\mA$.
The transformation is illustrated in \Cref{fig:matrixA,fig:matrixB}.

\begin{figure}
\centering
\begin{subfigure}[t]{0.3\textwidth}
  \centering
  \includegraphics[width=\linewidth]{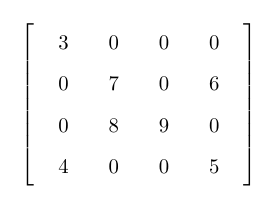}
  \caption{Original matrix $\mB = \mB_{4,2}$ of size $4 \times 4$.}
  \label{fig:matrixB}
\end{subfigure}
\hspace{2em}
\begin{subfigure}[t]{0.35\textwidth}
  \centering
  \includegraphics[width=\linewidth]{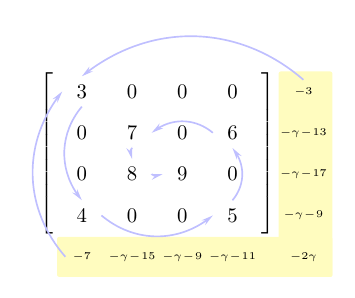}
  \caption{New matrix $\mA$ of size $5 \times 5$.}
  \label{fig:matrixA}
\end{subfigure}
\end{figure}

\begin{equation}
\mA [a_1,a_2]
=
\begin{cases}
\mB[a_1,a_2],
& 1 \le a_1 \le m-1 \ \text{and}\ 1 \le a_2 \le m-1,
\\[2pt]
- \sum_{a_1'=1}^{m-1} \mB[a_1',1],
& a_1 = m \ \text{and}\ a_2 = 1,
\\[2pt]
- \gamma - \sum_{a_1'=1}^{m-1} \mB[a_1',a_2],
& a_1 = m \ \text{and}\ 2 \le a_2 \le m-1,
\\[2pt]
- \sum_{a_2'=1}^{m-1} \mB[1,a_2'],
& a_2 = m \ \text{and}\ a_1 = 1,
\\[2pt]
- \gamma - \sum_{a_2'=1}^{m-1} \mB[a_1,a_2'],
& a_2 = m \ \text{and}\ 2 \le a_1 \le m-1,
\\[2pt]
- 2 \; \gamma,
& a_1 = m \ \text{and}\ a_2 = m,
\\[2pt]
0,
& \text{otherwise.}
\end{cases}
\label{eq:matrixA}
\end{equation}

\begin{lemma}[Structural properties of $\mA$]
\label{lem:matrixA_props}
The following properties hold for $\mA$:
\begin{enumerate}[label=(\roman*)]
\item The positive entries $\{3,4,\ldots,2m-1\}$ each appear exactly once.
\label{prop:matrixA_prop4_unique}
\item $\max_{i,j} \mA[i,j] = (2m-1)$.
\label{prop:matrixA_prop1_max}
\item For even $k \in \{3,4,\dots,2m-2\}$,
\(
a_1(k) = a_1(k+1).
\)
\label{prop:matrixA_prop2_even}
\item For odd $k \in \{3,4,\dots,2m-2\}$,
\(
a_2(k) = a_2(k+1).
\)
\label{prop:matrixA_prop3_odd}
\end{enumerate}
\end{lemma}

\begin{proof}
All claims follow from the properties established in \cref{sec:matrix_props} for the embedded submatrix $\mB_{m-1,2}$ of $\mA$. In particular, \cref{prop:matrixA_prop4_unique} is an immediate consequence of \cref{lem:support_matrixB}. Moreover, by \cref{prop:matrixB_prop1} in \cref{lem:matrixB_props}, the maximum entry of $\mB_{m-1,2}$ equals
\(
2 + \bigl(2(m-1)-1\bigr)=2m-1.
\) proving \cref{prop:matrixA_prop1_max}.
Since the embedding preserves the values of $\mB_{m-1,2}$, and all additional entries in $\mA$ (in particular those in the extra row and column) are chosen to be negative, this entry is also the maximum of $\mA$.
The remaining properties follow directly from \cref{prop:matrixB_prop2,prop:matrixB_prop3} in \cref{lem:matrixB_props}, applied to the same embedded block.
\end{proof}

As illustrated below, this construction ensures that, under uniform initialization, all actions other than $(1,1)$ start with a utility gap of order $\Omega(\gamma / m)$, thereby enforcing that the action profile $(1,1)$ is played with probability nearly $1$.   
The remainder of the analysis relies on \cref{property:main}, which we prove by induction in \cref{sec:valid_period}.
In what follows, we establish \cref{asmpt:initial_conditions}, which will be used throughout the analysis.

\begin{property}
\label{asmpt:initial_conditions}
At round $1$, the utility gaps satisfy
\[
\gap{1}{1}{a_1} \ge \gamma^{(1)}
\quad
\text{and}
\quad
\gap{2}{1}{a_2} \ge \gamma^{(1)},
\]
for all actions $a_1 \in \cA_1$ and $a_2 \in \cA_2$ except for action~$1$, which has zero gap, and $\gamma^{(1)}$ satisfies
\begin{equation}
\gamma^{(1)} \defeq 
\max\Bigl\{  
\frac{4 m^3 \,\regcon}{\etacon},
2
\left(
\frac{64 \regcon m}{\delta}
\right)^{\frac{1}{1-\alpha}},
\left(
\frac{\regcon m}{\delta}
\right)^{\frac{1}{1-\alpha}}
\alpha^{\frac{\alpha}{1-\alpha}}
(1-\alpha)
\Bigr\}
\label{eq:gamma_init}
\end{equation}

Then, it also holds
\[
\vx_{2}^{(2)}[a_2(4)] \ge 1-\delta, \quad
\vx_{1}^{(2)}[a_1(3)] + \vx_{1}^{(2)}[a_1(4)] \ge 1-\delta,
\quad
\text{where}
\;
\delta = \frac{1}{4m}.
\]
In other words, period $k=4$ starts at round $2$, satisfying the \cref{def:transition_period}.
\end{property}




\begin{lemma}
Let $\vx_1^{(1)}=\vx_2^{(1)}=\mathrm{Uniform}(m)$. If both players employ \ftrl on the game $(\mA,\mA)$, then \cref{asmpt:initial_conditions} is satisfied.
\end{lemma}

\begin{proof}
From \cref{lem:uniform_minimizer_regularizer}, the \ftrl update at round $1$ dictates that both players choose the uniform strategy over their action sets, \textit{i.e.},
\[
\vx_1^{(1)} = \vx_2^{(1)} = \mathrm{Uniform}(m).
\]  

Although the effect of \eqref{eq:matrixA} on the utility vectors $\vu_1^{(1)}, \vu_2^{(1)}$ is apparent, we present a detailed derivation for completeness. We describe Player~1's utility vector; Player~2's is analogous.

\begin{enumerate}

\item \textbf{Row $a_1 = 1$}  
\begin{align*}
\vu_{1}^{(1)}[1]
&\defeq 
\sum_{a_2=1}^{m} \vx_{2}^{(1)}[a_2] \, \mA[1, a_2] \\
&
= 
\frac{1}{m} \left( \sum_{a_2=1}^{m-1} \mA[1, a_2] + \mA[1, m] \right) 
\\
&
= 
\frac{1}{m} \left( \sum_{a_2=1}^{m} \mB[1, a_2] + \left( - \sum_{a_2'=1}^{m} \mB[1,a_2'] \right) \right) 
\\
&
= 0.
\end{align*}

\item \textbf{Rows $2 \le a_1 \le m-1$}  
\begin{align*}
\vu_{1}^{(1)}[a_1]
&\defeq
\sum_{a_2=1}^{m} \vx_{2}^{(1)}[a_2] \, \mA[a_1, a_2] 
\\
&= 
\frac{1}{m} \left( \sum_{a_2=1}^{m-1} \mA[a_1, a_2] 
+
\mA[a_1, m]  
\right) \\
&= 
\frac{1}{m} \left( \sum_{a_2=1}^{m-1} \mB[a_1, a_2] 
+
\left(
- \gamma - \sum_{a_2'=1}^{m-1} \mB[a_1,a_2'] \right) \right) \\
&= - \frac{\gamma}{m}.
\end{align*}

\item \textbf{Row $a_1 = m$}  
\begin{align*}
\vu_{1}^{(1)}[m]
&\defeq
\sum_{a_2=1}^{m} 
\vx_{2}^{(1)}[a_2] \, \mA[m, a_2]
\\
&= \frac{1}{m} 
\Bigg(
\mA[m, 1] 
+ 
\sum_{a_2=2}^{m-1} \mA[m, a_2]
+ 
\mA[m, m] 
\Bigg) \\
&= 
\frac{1}{m} 
\Bigg( 
- \sum_{a_1'=1}^{m-1} \mB[a_1',1] 
+
\sum_{a_2=2}^{m-1} 
\big(
- \gamma - \sum_{a_1'=1}^{m-1} \mB[a_1',a_2] 
\big) - 2 \gamma 
\Bigg)
\\
&= 
\frac{1}{m} 
\Bigg( 
-
(m-2) \gamma
 - 2 \gamma 
- \sum_{a_1'=1}^{m-1} \mB[a_1',1] 
-
\sum_{a_2=2}^{m-1} 
\sum_{a_1'=1}^{m-1} \mB[a_1',a_2] 
\Bigg)
\\
&= 
\frac{1}{m} 
\Bigg( - m \gamma - \sum_{a_1=1}^{m-1} \sum_{a_2=1}^{m-1} \mB[a_1,a_2] \Bigg) \\
&= - \gamma - \frac{1}{m} \sum_{a_1=1}^{m-1} \sum_{a_2=1}^{m-1} \mB[a_1,a_2].
\end{align*}
\end{enumerate}

Only action $1$ has non-negative utility at round $1$, while all other actions have utility gaps of order 
\(
\Omega\left(\gamma / m \right)
\), as shown below.
Analogously, the same holds for Player~2's utility vector $\vu_2^{(1)}$. 
Since $\gamma$ is a free parameter that can be chosen arbitrarily large, \cref{lem:regret_gap} and this construction ensure that under uniform initialization, the action profile $(1,1)$ will be played at round $2$ with probability nearly $1$.
\begin{align*}
\gap{1}{1}{a_1}
&=
\begin{cases}
0
&
a_1 = 1,
\\[2pt]
\frac{\gamma}{m}
&
2 \le a_1 \le m-1,
\\[2pt]
\gamma
+
\frac{
\sum_{a_1, a_2}
\mB[a_1, a_2]
}{m}
& a_1 = m.
\end{cases}
\\
\\
\gap{2}{1}{a_2}
&=
\begin{cases}
0
&
a_2 = 1,
\\[2pt]
\frac{\gamma}{m}
&
2 \le a_2 \le m-1,
\\[2pt]
\gamma
+
\frac{
\sum_{a_1, a_2}
\mB[a_1, a_2]
}{m}
& 
a_2 = m.
\end{cases}
\end{align*}


Setting \(\gamma^{(1)} \defeq \frac{\gamma}{m}\), it follows that \(\gap{i}{1}{a} \geq \gamma^{(1)}\) for every \(a\neq 1\).
Now, by \cref{lem:regret_gap} and setting
\(
\gamma^{(1)}
= \frac{4 m^3 \regcon}{\etacon}
\)
it follows that for all \(a_1 \neq 1\),
\[
\vx_1^{(2)}[a_1] 
\;\leq\; 
\frac{\regcon}{\eta^{(1)} \gap{1}{1}{a_1}}
\;\leq\; 
\frac{\regcon}{
\left(\frac{\etacon}{1^\alpha}\right)
\left(\frac{4 m^3 \regcon}{\etacon}
\cdot 
\frac{1}{m}\right)}
\;=\; 
\frac{1}{4 m^2}
\;\leq\;
\frac{1}{4m}
\frac{1}{m}
=
\frac{\delta}{m}.
\]
and therefore
\[
\vx_1^{(2)}[1] 
\;=\;
1 
-
\sum_{a_1 = 2}^{m}
\vx_1^{(2)}[a_1] 
\;\geq\; 1 - \frac{m}{4 m^2} 
\;\geq\; 1 - \frac{1}{4 m}
\;=\; 1 - \delta.
\]

An analogous argument applies to Player~2's strategy $\vx_2^{(1)}$.
Therefore, in accordance with \cref{def:transition_period}, period $k=3$ starts at round $1$ and terminates immediately, so period $4$ begins at round $2$. Indeed, at round $2$ the defining conditions hold:
\[
\vx_1^{(2)}[a_1(3)] + \vx_1^{(2)}[a_1(4)] \ge 1-\delta,
\qquad
\vx_2^{(2)}[a_2(4)] \ge 1-\delta.
\]
Finally, setting
\[
\gamma^{(1)}
\;=\;
\max\Bigl\{  
\frac{4 m^3 \,\regcon}{\etacon},
2
\left(
\frac{64 \regcon m}{\delta}
\right)^{\frac{1}{1-\alpha}},
\left(
\frac{\regcon m}{\delta}
\right)^{\frac{1}{1-\alpha}}
\alpha^{\frac{\alpha}{1-\alpha}}
(1-\alpha)
\Bigr\},
\quad 
\text{where}
\quad
\delta = \frac{1}{4m}.
\]
ensures that \cref{asmpt:initial_conditions} is satisfied.
\end{proof}

\begin{remark}[Normalized payoff matrix $\mA$]
\label{rem:bounded_payoffs}
Our lower bound construction can be implemented with payoffs in $[-1,1]$.
Although \ftrl is not scale invariant---rescaling payoffs changes the updates---we can modify the construction so that it does not rely on large-magnitude entries.

Specifically, assume henceforth that payoffs must lie in $[-1,1]$.
We start from the recursive matrix $\mB_{m-1,2}$ in \cref{def:payoffB}. By \cref{prop:matrixB_prop1} in \cref{lem:matrixB_props}, its maximum entry equals $2m-1$.
Define the rescaled matrix
\(
\mB \;\coloneqq\;\frac{1}{2m-1}\,\mB_{m-1,2},
\)
so that $\mB\in[0,1]^{(m-1)\times(m-1)}$. We then construct our payoff matrix $\mA$ by embedding $\mB$, as in the original construction, and padding it with $\ceil{\gamma^{(1)}}$ additional \emph{rows} and \emph{columns}, where $\gamma^{(1)}$ is as in \eqref{eq:gamma_init}.
All newly introduced entries are set to $-1$, except that in each added row and each added column we set the first entry to $0$.
This should be contrasted with the initial construction in \eqref{eq:matrixA}, which adds only one row and one column but uses entries of magnitude $\Theta(\gamma^{(1)})$.
Since $\ceil{\gamma^{(1)}}=\poly(m)$, the resulting game has polynomial dimension and all payoffs lie in $[-1,1]$. It is easy to verify that this modification does not affect the subsequent arguments; consequently, the same lower bound holds: \ftrl still requires $2^{\Omega(\poly(m))}$ rounds.
\end{remark}


\subsection{Growth of the cumulative utility gaps}

As shown in~\cref{lem:regret_gap}, to control the probability assigned to an action it suffices to track its \emph{utility gap} over time, rather than the full cumulative utility vector.


\begin{lemma}
\label{lem:regret_gap_increase_actions_1}
Let $k \ge 4$ be even. For any action $a_1 \in \cA_1 \setminus \{ a_1(k), a_1(k-2) \}$, the utility gap $\gap{1}{t}{a_1}$ increases by at least $(1-\delta)(k-1) - \delta (2m-1)$ at every round $t \in \period{k}$.  
Similarly, if $k \ge 5$ is odd, then for any action $a_2 \in \cA_2 \setminus \{ a_2(k), a_2(k-2) \}$, the utility gap $\gap{2}{t}{a_2}$ increases by at least $(1-\delta)(k-1) - \delta (2m-1)$ at every round $t \in \period{k}$.
\end{lemma}

\begin{proof}
Let $a_1 \in \cA_1 \setminus \{ a_1(k), a_1(k-1) \}$ for even $k$.
By the definition of period~$k$, we have
\[
\vx_{1}^{(t)}[a_1(k-1)] + \vx_{1}^{(t)}[a_1(k)] \ge 1-\delta,
\qquad
\vx_{2}^{(t)}[a_2(k)] \ge 1-\delta .
\]

Since during period~$k$ Player~1 primarily mixes between the actions $a_1(k-1)$ and $a_1(k)$, it follows from \cref{lemma:aux_prob_comparison} that at either
$\vx_{1}^{(t)}[a_1(k-1)]$ or $\vx_{1}^{(t)}[a_1(k)]$ is maximal among the components of $\vx_{1}^{(t)}$.
By \cref{lem:order_preservation}, the same action also maximizes the cumulative utility at rounds $t-1$ and $t$.
Consequently,
\[
\{ a_1(k-1), a_1(k) \}
\supseteq
\argmax_{a_1' \in \cA_1}
\left\{
\sum_{\tau=1}^{t-1} \vu_1^{(\tau)}[a_1']
\right\}
\neq \emptyset
.
\]

Since $k$ is even, \cref{prop:matrixA_prop3_odd} in \cref{lem:matrixA_props} implies $a_2(k-1)=a_2(k)$, and so the utilities of actions $a_1(k-1)$ and $a_1(k)$ at round $t$ satisfy
\begin{equation}
\vu_1^{(t)}[a_1(k-1)]
\;\ge\;
\vx_{2}^{(t)}[a_2(k-1)]\,(k-1)
\;\ge\;
(1-\delta)\,(k-1),
\qquad
\vu_1^{(t)}[a_1(k)]
\;\ge\;
\vx_{2}^{(t)}[a_2(k)]\,k
\;\ge\;
(1-\delta)\,k
.
\label{eq:evol_gaps_eq3}
\end{equation}

Next, consider the one-step change in the gap of action $a_1$ after round~$t$:
\begin{align}
\Delta \gap{1}{t}{a_1}
&\;\defeq\;
\gap{1}{t}{a_1} - \gap{1}{t-1}{a_1}
\notag\\
&=
\left(
\max_{a_1' \in \cA_1}
\sum_{\tau=1}^{t} \vu_1^{(\tau)}[a_1']
-
\sum_{\tau=1}^{t} \vu_1^{(\tau)}[a_1]
\right)
-
\left(
\max_{a_1' \in \cA_1}
\sum_{\tau=1}^{t-1} \vu_1^{(\tau)}[a_1']
-
\sum_{\tau=1}^{t-1} \vu_1^{(\tau)}[a_1]
\right)
\notag\\
&=
\left(
\max_{a_1' \in \cA_1}
\sum_{\tau=1}^{t} \vu_1^{(\tau)}[a_1']
-
\max_{a_1' \in \cA_1}
\sum_{\tau=1}^{t-1} \vu_1^{(\tau)}[a_1']
\right)
-
\vu_1^{(t)}[a_1]
.
\label{eq:evol_gaps_eq1}
\end{align}

Even though (during period~$k$) one of $a_1(k-1)$ and $a_1(k)$ attains the maximal cumulative utility at rounds $t-1$ and~$t$, the maximizer could in principle change from one round to the next.
Fortunately, \cref{lem:aux_max_min_increment} avoids a case distinction.
\begin{equation}
\max_{a_1' \in \cA_1}
\sum_{\tau=1}^{t} \vu_1^{(\tau)}[a_1']
-
\max_{a_1' \in \cA_1}
\sum_{\tau=1}^{t-1} \vu_1^{(\tau)}[a_1']
\;\ge\;
\min\!\left\{
\vu_1^{(t)}[a_1(k-1)],
\vu_1^{(t)}[a_1(k)]
\right\}.
\label{eq:evol_gaps_eq2}
\end{equation}

Combining \eqref{eq:evol_gaps_eq1}, \eqref{eq:evol_gaps_eq2}, and \eqref{eq:evol_gaps_eq3}, we obtain
\[
\Delta \gap{1}{t}{a_1}
\;\ge\;
(1-\delta)(k-1) - \vu_1^{(t)}[a_1].
\]

For any $a_1 \in \cA_1 \setminus \{a_1(k),a_1(k-1)\}$, the identity $a_2(k-1)=a_2(k)$ (by \cref{prop:matrixA_prop3_odd} in \cref{lem:matrixA_props}) implies $\mA[a_1,a_2(k)]=0$; equivalently, row $a_1$ has no nonzero entry in column $a_2(k)$.
Moreover, by the definition of period~$k$, Player~2 assigns total probability at most $\delta$ to columns other than $a_2(k)$. Finally, \cref{prop:matrixA_prop1_max} in \cref{lem:matrixA_props} implies that every payoff outside column $a_2(k)$ is at most $2m-1$. 
Hence,
\[
\vu_1^{(t)}[a_1] \;\le\; \delta(2m-1).
\]
Therefore, due to $a_1(k-2)=a_1(k-1)$ from \cref{prop:matrixA_prop2_even} in \cref{lem:matrixA_props}, it follows that for any
$a_1 \in \cA_1 \setminus \{a_1(k),a_1(k-1)\}
=
\cA_1 \setminus \{a_1(k),a_1(k-2)\}$,
\[
\Delta \gap{1}{t}{a_1}
\;\ge\;
(1-\delta)(k-1) - \delta(2m-1).
\]
The proof for odd $k$ and for Player~2 is analogous.

\end{proof}

\begin{lemma}
\label{lem:regret_gap_increase_actions_2}
Let $k \ge 4$ be even. For any action $a_2 \in \cA_2 \setminus \{a_2(k), a_2(k+2)\}$, the utility gap $\gap{2}{t}{a_2}$ increases by at least $(1-\delta)(k-1) - \delta(2m-1)$ at every round $t \in [\tlow{k}, \thigh{k}]$. 
Similarly, if $k \ge 5$ is odd, then for any action $a_1 \in \cA_1 \setminus \{a_1(k), a_1(k+2)\}$, the utility gap $\gap{1}{t}{a_1}$ increases by at least $(1-\delta)(k-1) - \delta(2m-1)$ at every round $t \in [\tlow{k}, \thigh{k}]$.
\end{lemma} 

\begin{proof}
Let $a_2 \in \cA_2 \setminus \{a_2(k),a_2(k+2)\}$ and suppose that $k$ is even.
By the definition of transition period~$k$, for every $t \in [\tlow{k},\thigh{k}]$,
\[
\vx_{1}^{(t)}[a_1(k-1)] + \vx_{1}^{(t)}[a_1(k)] \ge 1-\delta,
\qquad
\vx_{2}^{(t)}[a_2(k)] \ge 1-\delta .
\]
Since $\vx_{2}^{(t)}[a_2(k)] \ge 1-\delta$, action $a_2(k)$ has maximal probability among actions in $\cA_2$, and thus \cref{lem:order_preservation} implies that it also attains the maximal cumulative utility up to time $t-1$, \textit{i.e.},
\[
a_2(k) \in \arg\max_{a_2' \in \cA_2} \sum_{\tau=1}^{t-1} \vu_2^{(\tau)}[a_2'].
\]
Consequently, the one-step change in the gap of $a_2$ at round $t$ satisfies
\begin{align}
\Delta \gap{2}{t}{a_2}
&\defeq
\gap{2}{t}{a_2} - \gap{2}{t-1}{a_2}
\nonumber
\\
&=
\left(
\max_{a_2' \in \cA_2}\sum_{\tau=1}^{t}\vu_2^{(\tau)}[a_2']
-
\max_{a_2' \in \cA_2}\sum_{\tau=1}^{t-1}\vu_2^{(\tau)}[a_2']
\right)
-
\vu_2^{(t)}[a_2]
\nonumber
\\
&\ge
\left(
\sum_{\tau=1}^{t}\vu_2^{(\tau)}[a_2(k)]
-
\sum_{\tau=1}^{t-1}\vu_2^{(\tau)}[a_2(k)]
\right)
-
\vu_2^{(t)}[a_2]
\nonumber
\\
&=
\vu_2^{(t)}[a_2(k)]-\vu_2^{(t)}[a_2].
\label{eq:gap_increase_eq1}
\end{align}

For the first term, the defining condition of period~$k$ yields
\(
\vu_2^{(t)}[a_2(k)]
=
\vx_{1}^{(t)}[a_1(k)]\,k
+
\vx_{1}^{(t)}[a_1(k-1)]\,(k-1)
\ge
(1-\delta)(k-1).
\)
For the second term, the identities $a_2(k-1)=a_2(k)$ and $a_2(k+1)=a_2(k+2)$ (by \cref{prop:matrixA_prop3_odd} in \cref{lem:matrixA_props}) imply that, for any
$a_2 \in \cA_2 \setminus \{a_2(k),a_2(k+2)\}$,
\(
\mA[a_1(k-1),a_2]=0
\)
and
\(
\mA[a_1(k),a_2]=0.
\)
Moreover, by the definition of period~$k$, Player~1 assigns total probability at most $\delta$ to actions other than $a_1(k-1)$ and $a_1(k)$, and \cref{prop:matrixA_prop1_max} in \cref{lem:matrixA_props} implies that every payoff outside these two rows is at most $2m-1$.
Therefore,
\begin{align}
\vu_2^{(t)}[a_2] \le \delta(2m-1).
\label{eq:gap_increase_eq3}
\end{align}
Combining \eqref{eq:gap_increase_eq1} and \eqref{eq:gap_increase_eq3}, we conclude that for any
$a_2 \in \cA_2 \setminus \{a_2(k),a_2(k+2)\}$,
\[
\Delta \gap{2}{t}{a_2}
=
\vu_2^{(t)}[a_2(k)]-\vu_2^{(t)}[a_2]
\;\ge\;
(1-\delta)(k-1)-\delta(2m-1).
\]
This establishes the claim. The proof for odd $k$ (and the corresponding case for Player~1) is analogous.
\end{proof}

\subsection{Lower bounds on period $k$ duration}

So far, we have characterized how the gap of an action $a_i\in\cA_i$ evolves over the rounds within period $k$. We now use this characterization to bound the duration of period $k$.

\begin{lemma}
\label{lem:lower_bound_period_k}
Let $k \ge 4$ be even. Then
\[
T_k \;\ge\; \frac{1}{1 + \delta(2m-1)} \, \gap{1}{\thigh{k-1}}{a_1(k)}.
\]
Similarly, if $k \ge 5$ is odd, then
\[
T_k \;\ge\; \frac{1}{1 + \delta(2m-1)} \, \gap{2}{\thigh{k-1}}{a_2(k)}.
\]
\end{lemma}

\begin{proof}
Let $k$ be even (the odd-$k$ case is symmetric).
By definition, period~$k$ ends at the round $t$ when $\vecx{1}{t} [a_1(k)]$ reaches the threshold $1-\delta$ (and, for odd~$k$, once the probability mass on $\vecx{2}{t} [a_2(k)]$ reaches $1-\delta$).
Before this happen, the action $a_1(k)$ needs first to become a maximizer of the cumulative gap over $\cA_1$ (as shown in \cref{property:main})
; this occurs at round
\(
\hat t_k \in [\tlow{k},\thigh{k}].
\)

For every $t \in [\tlow{k}, \hat t_k-1]$, the one-step change in the utility gap of $a_1(k)$ is lower bounded by
\begin{align*}
\Delta \gap{1}{t}{a_1(k)}
&\;=\;
\left(
\max_{a_1' \in \cA_1}\sum_{\tau=1}^{t}\vu_{1}^{(\tau)}[a_1']
-
\max_{a_1' \in \cA_1}\sum_{\tau=1}^{t-1}\vu_{1}^{(\tau)}[a_1']
\right)
- 
\vu_1^{(t)}[a_1(k)]
\\
&\;=\;
\left(
\sum_{\tau=1}^{t}\vu_{1}^{(\tau)}[a_2(k-1)]
-
\sum_{\tau=1}^{t-1}\vu_{1}^{(\tau)}[a_2(k-1)]
\right)
- 
\vu_1^{(t)}[a_1(k)]
\\
&\;=\;
\vu_1^{(t)}[a_1(k-1)] - \vu_1^{(t)}[a_1(k)]
\\
&\;\ge\;
\vx_{2}^{(t)}[a_2(k-1)] (k-1)
-
\vx_{2}^{(t)}[a_2(k)]\,k
-
\vx_{2}^{(t)}[a_2(k+1)] (k+1).
\end{align*}
\cref{prop:matrixA_prop3_odd} in \cref{lem:matrixA_props} implies that $a_2(k-1) = a_2(k)$ for $k$ even, and so 
\begin{align*}
\Delta \gap{1}{t}{a_1(k)}
&\;\ge\;
-
\vx_{2}^{(t)}[a_2(k)]
-
\vx_{2}^{(t)}[a_2(k+1)] (k+1)
\\
&\;=\;
-\vx_{2}^{(t)}[a_2(k)] 
-
\vx_{2}^{(t)}[a_2(k+1)]
(2m-1)
\\
&\;\ge\;
-1 - \delta (2m-1),
\end{align*}
where the last inequality uses $\vx_{2}^{(t)}[a_2(k+1)] \le \delta$ and $\vx_{2}^{(t)}[a_2(k)] \le 1$.
Equivalently, over period $k$ the gap $\gap{1}{t}{a_1(k)}$ can decrease by at most $1+\delta(2m+1)$ per round.
Moreover, since period~$k$ lasts at least until $\hat t_k$, we have $T_k \ge (\hat t_k - 1) - \tlow{k} + 1$.
Because the gap at the beginning of the period is finite, it follows that
\[
T_k 
\;\ge\;
\frac{1}{1+\delta(2m-1)} \,\gap{1}{\thigh{k-1}}{a_1(k)},
\]
which establishes the claim.
The odd-$k$ case follows analogously.

\end{proof}



\gaplb*

\begin{proof}
The claim follows by repeatedly applying \cref{lem:regret_gap_increase_actions_1}.
We present the argument for even $k \ge 6$; the odd-$k$ case is analogous.

For any period index $\ell \in \{4,\dots,k-2\}$, applying \cref{lem:regret_gap_increase_actions_1} within period~$\ell$ yields that the regret gap of action $a_1(k)$ increases by at least
\[
\Bigl[(1-\delta)(\ell-1) - \delta(2m-1)\Bigr]
\]
per round throughout that period.
Therefore, over the entire duration $T_\ell$ of period~$\ell$, we obtain the increment bound
\[
\gap{1}{\thigh{\ell}}{a_1(k)}
\;\ge\;
\gap{1}{\thigh{\ell-1}}{a_1(k)}
\;+\;
\Bigl[(1-\delta)(\ell-1) - \delta(2m-1)\Bigr]\, T_{\ell}.
\]
Summing the above inequality over $\ell = 4,\dots,k-2$ and telescoping gives
\[
\gap{1}{\thigh{k-2}}{a_1(k)}
\;\ge\;
\gap{1}{\thigh{3}}{a_1(k)}
\;+\;
\sum_{\ell = 4}^{k-2}
\Bigl[(1-\delta)(\ell-1) - \delta(2m-1)\Bigr]\, T_{\ell}.
\]
Dropping the nonnegative term $\gap{1}{\thigh{3}}{a_1(k)} = \gap{1}{1}{a_1(k)} \ge 0$ results to the stated lower bound:
\[
\gap{1}{\thigh{k-2}}{a_1(k)}
\;\ge\;
\sum_{\ell = 4}^{k-2}
\Bigl[(1-\delta)(\ell-1) - \delta(2m-1)\Bigr]\, T_{\ell}.
\]
The odd-$k$ case follows analogously.
\end{proof}

Although \cref{lem:lower_bound_period_k,lem:lower_bound_regret_gap} appear to admit a recursive application, a mismatch in their reference points prevents this directly. Specifically, \cref{lem:lower_bound_regret_gap} provides a lower bound on the regret gap only up to the end of period \(k-2\), whereas \cref{lem:lower_bound_period_k} concerns the subsequent period \(k-1\). To bridge this gap, we require an auxiliary lemma that tracks the evolution of the regret gap of \(a_1(k)\) from the end of \(k-2\) through the end of \(k-1\).

\begin{lemma}
\label{lem:regret_gap_increase_prev_action}
Let $k \ge 4$ be even. Then,
\[
\gap{1}{\thigh{k-1}}{a_1(k)}
\;\geq\;
\gap{1}{\thigh{k-2}}{a_1(k)} 
-
\left(
1 + \delta(2m-1)
\right)
T_{k-1}
.
\]
Similarly, if $k \ge 5$ is odd, then 
\[
\gap{2}{\thigh{k-1}}{a_2(k)}
\;\geq\;
\gap{2}{\thigh{k-2}}{a_2(k)} 
-
\left(
1 + \delta(2m-1)
\right)
T_{k-1}
.
\]
\end{lemma}

\begin{proof}
We track the evolution of $\gap{1}{t}{a_1(k)}$ during period~$k-1$ for even $k$.
Fix any $t \in [\tlow{k-1},\thigh{k-1}]$.
By the definition of transition period~$k-1$, it holds that
\[
\vx_{2}^{(t)}[a_2(k-2)] + \vx_{2}^{(t)}[a_2(k-1)] \ge 1-\delta
\quad
\text{and}
\quad
\vx_{1}^{(t)}[a_1(k-1)] \ge 1-\delta.
\]

Since Player~1 assigns probability at least $1-\delta$ to $a_1(k-1)$, \cref{lem:order_preservation} implies that $a_1(k-1)$ attains the maximal cumulative utility in $\cA_1$ throughout period $k-1$.
Therefore, the maximizer in the definition of the gap is $a_1(k-1)$ and
\begin{align*}
\Delta \gap{1}{t}{a_1(k)}
&
\;\defeq\;
\gap{1}{t}{a_1(k)}-\gap{1}{t-1}{a_1(k)}
\\
&
\;=\;
\left(
\max_{a_1' \in \cA_1}\sum_{\tau=1}^{t}\vu_1^{(\tau)}[a_1']
-
\max_{a_1' \in \cA_1}\sum_{\tau=1}^{t-1}\vu_1^{(\tau)}[a_1']
\right)
-
\vu_1^{(t)}[a_1(k)]
\\
&
\;=\;
\vu_1^{(t)}[a_1(k-1)]-\vu_1^{(t)}[a_1(k)].
\end{align*}
Expanding the utilities gives
\[
\Delta \gap{1}{t}{a_1(k)}
=
\vx_2^{(t)}[a_2(k-1)] (k-1)
+
\vx_2^{(t)}[a_2(k-2)] (k-2)
-
\vx_2^{(t)}[a_2(k)]\,k
-
\vx_2^{(t)}[a_2(k+1)] (k+1).
\]
Since $k$ is even, \cref{prop:matrixA_prop3_odd} in \cref{lem:matrixA_props} implies $a_2(k-1)=a_2(k)$.
We now derive a crude but sufficient bound.
The terms satisfy $\vx_2^{(t)}[a_2(k-2)]\ge 0$ and $\vx_2^{(t)}[a_2(k-1)]\le 1$.
Moreover, by the definition of the period, Player~2 assigns probability at most $\delta$ to any action other than $a_2(k-2), a_2(k-1)$, hence $\vx_2^{(t)}[a_2(k+1)]\le \delta$.
Substituting these bounds yields
\[
\Delta \gap{1}{t}{a_1(k)}
\ge
(k-1)-k-\delta(k+1)
=
-1-\delta(k+1)
\ge
-1-\delta(2m-1),
\]
where the last inequality uses $k+1\le 2m-1$.
Thus, throughout period~$k-1$, the gap $\gap{1}{t}{a_1(k)}$ can decrease by at most $1+\delta(2m-1)$ per round, and therefore
\[
\gap{1}{\thigh{k-1}}{a_1(k)}
\ge
\gap{1}{\thigh{k-2}}{a_1(k)}
-
\bigl(1+\delta(2m-1)\bigr)\,T_{k-1}.
\]
The odd-$k$ case follows by an analogous argument.
\end{proof}


\subsection{Exponential lower bound}

Here we combine the lemmas from the preceding subsection to obtain an exponential lower bound. In particular, we show that the combined duration of periods $2m-4$ and $2m-3$ is already exponential in $m$.

\begin{lemma}[Lower bound on $T_4$]
\label{lem:period_2_lower_bound}
Under \cref{asmpt:initial_conditions}, it holds that
\[
T_4 \;\ge\; \frac{\gamma_1}{2}.
\]
\end{lemma}

\begin{proof}
By \cref{lem:lower_bound_period_k}, for $k=4$ (even) we have
\(
T_4 \ge \frac{1}{1+\delta(2m+1)}\,\gap{1}{\thigh{3}}{a_1(4)}.
\)
By \cref{asmpt:initial_conditions}, period~$3$ is the single round~$1$, so
\(
\gap{1}{\thigh{3}}{a_1(4)}=\gap{1}{1}{a_1(4)}\ge \gamma^{(1)}.
\)
Substituting yields
\[
T_4 \ge \frac{\gamma_1}{1+\delta(2m+1)}.
\]
Finally, since $\delta=\frac{1}{4m}$, it follows that
\(
1+\delta(2m+1)\le 1+\frac{2m+1}{4m}\le 2,
\)
and therefore \(T_4 \ge \gamma_1/2\), as claimed.
\end{proof}

\recurs*

\begin{proof}
We present the argument for even~$k$; the case of odd~$k$ is analogous.
By \cref{lem:lower_bound_period_k}, the length of period~$k$ satisfies
\begin{equation}
T_k
\;\ge\;
\frac{1}{1 + \delta(2m-1)}\,
\gap{1}{\thigh{k-1}}{a_1(k)}.
\label{eq:rec1}
\end{equation}
To lower bound the gap at time $\thigh{k-1}$, we first invoke \cref{lem:lower_bound_regret_gap}, which yields
\begin{equation}
\gap{1}{\thigh{k-2}}{a_1(k)}
\;\ge\;
\sum_{\ell=4}^{k-2}
\Bigl[
(1-\delta)(\ell-1)
-
\delta(2m-1)
\Bigr]
\,
T_{\ell}.
\label{eq:rec2}
\end{equation}
Next, to relate the gap at times $\thigh{k-2}$ and $\thigh{k-1}$, we apply
\cref{lem:regret_gap_increase_prev_action}, obtaining
\begin{equation}
\gap{1}{\thigh{k-1}}{a_1(k)}
\;\ge\;
\gap{1}{\thigh{k-2}}{a_1(k)}
-
\bigl(1 + \delta(2m-1)\bigr)\, T_{k-1}.
\label{eq:rec3}
\end{equation}
Substituting \eqref{eq:rec2} into \eqref{eq:rec3} and then into \eqref{eq:rec1} gives
\begin{align}
T_k
&\ge
\frac{
\sum_{\ell=4}^{k-2}
\Bigl[
(1-\delta)(\ell-1)
-
\delta(2m-1)
\Bigr]
\,
T_{\ell}
}{
1 + \delta(2m-1)
}
-
T_{k-1}.
\label{eq:rec4}
\end{align}
Rearranging \eqref{eq:rec4} yields
\begin{equation}
T_k + T_{k-1}
\;\ge\;
\frac{
\sum_{\ell=4}^{k-2}
\Bigl[
(1-\delta)(\ell-1)
-
\delta(2m-1)
\Bigr]
\,
T_{\ell}}{1 + \delta(2m-1)}
.
\label{eq:rec5}
\end{equation}

By \cref{lemma:aux_decrease_by_constant}, for $\ell \ge 4$ we have
\[
(1-\delta)(\ell-1) - \delta(2m-1) \;\ge\; \ell-2.
\]
Moreover, since $\delta = \frac{1}{4m}$ (from \cref{asmpt:initial_conditions}), we have
\[
1 + \delta(2m-1) \;\le\; 1 + \frac{2m-1}{4m} \;\le\; 2.
\]
Plugging these bounds into \eqref{eq:rec5} proves that claim
\[
T_k + T_{k-1}
\;\ge\;
\frac{1}{2}
\sum_{\ell=4}^{k-2}
(\ell - 2)
\,
T_{\ell}.
\]

This completes the proof for even~$k$; the odd-$k$ case follows analogously.
\end{proof}

\begin{lemma}[Exponential lower bound]
\label{lem:exponential_lower_bound}
Under \cref{asmpt:initial_conditions}, it holds that
\[
T_{2m-3}+T_{2m-4}
\;\ge\;
\frac{\gamma^{(1)}}{4}\left(\frac{m-3}{e}\right)^{m-3}.
\]
\end{lemma}

\begin{proof}
Define, for $3\le k\le m$,
\[
S_k \;\defeq\; T_{2k-1}+T_{2k-2},
\qquad
P_k \;\defeq\; \sum_{\ell'=3}^{k}(\ell'-2)\,S_{\ell'}.
\]
Applying \cref{lem:lower_bound_k_k1} with index $2k-1$ gives
\[
S_k
=
T_{2k-1}+T_{2k-2}
\;\ge\;
\frac{1}{2}
\sum_{\ell=4}^{2k-3}
(\ell-2)
\,
T_{\ell}.
\]
Grouping the terms $(T_{2\ell-2},T_{2\ell-1})$ for $\ell=4,\dots,k-1$ and using that all coefficients are nonnegative yields a lower bound
\[
\frac{1}{2}
\sum_{\ell=4}^{2k-3}
(\ell-2)
\,
T_{\ell}
\;\ge\;
\frac{1}{2}
\sum_{\ell=3}^{k-1}
((2\ell-2)-2)
\left(
T_{2\ell-1}
+
T_{2\ell-2}
\right)
=
\sum_{\ell'=3}^{k-1}(\ell'-2)\,S_{\ell'}
=
P_{k-1}.
\]
Therefore, for all $k\ge 4$,
\begin{equation}
\label{eq:Sk_ge_Pkminus1}
S_k \;\ge\; P_{k-1}.
\end{equation}
It follows that for every $k\ge 4$,
\[
P_k
=
P_{k-1}+(k-2)S_k
\;\ge\;
P_{k-1}+(k-2)P_{k-1}
=
(k-1)P_{k-1},
\]
where we used \eqref{eq:Sk_ge_Pkminus1}.
Iterating from $k=4$ up to $k=m-2$ yields
\[
P_{m-1}
\;\ge\;
\left(\prod_{j=4}^{m-2}(j-1)\right)P_3
=
\frac{(m-3)!}{2} \,P_3
=
\frac{(m-3)!}{2} \,S_3,
\]
since $P_3=(3-2)S_3=S_3$.
Finally, \eqref{eq:Sk_ge_Pkminus1} with $k=m-1$ gives $S_{m-1}\ge P_{m-2}$, and $S_{m-1}=T_{2m-3}+T_{2m-4}$ by definition, so
\[
T_{2m-3}+T_{2m-4}
=
S_{m-1}
\;\ge\;
P_{m-2}
\;\ge\;
\frac{(m-3)!}{2}
\,
S_3.
\]
Moreover, $S_3=T_5+T_4\ge T_4$, and \cref{lem:period_2_lower_bound} gives $T_4\ge \gamma^{(1)}/2$, proving the factorial bound.

Applying Stirling's lower bound $(m-3)!\ge \bigl(\frac{m-3}{e}\bigr)^{m-3}$ yields
\[
T_{2m-3}+T_{2m-4}
\;\ge\;
\frac{\gamma^{(1)}}{4}
\left(
\frac{m-3}{e}
\right)^{m-3}
.
\]
\end{proof}

\subsection{Period consistency}
\label{sec:valid_period}

Here, we justify the definition of a \emph{period} (\cref{def:transition_period}) and establish its consistency.
To this end, we prove two lemmas.
Throughout, we focus on the case where $k$ is even; the case of odd~$k$ follows by an entirely analogous argument.

We show that the \ftrl dynamics are consistent with the definition of a period.
In particular, if round~$t$ belongs to period~$k$, then round~$t+1$ should belong only to period~$k$, unless
\(
\vx_{1}^{(t+1)}[a_1(k)] \ge 1-\delta,
\)
in which case period~$k$ terminates at round~$t$.
Hence, we have to rule out \emph{period mixing}: since \ftrl induces mixed strategies, it is not a priori clear that the update cannot move directly from period~$k$ to some period other than~$k+1$.
To exclude this, we proceed in two steps.
First, \cref{thm:uniform_small_prob_fixed_action} shows that any action whose gap grows linearly in~$t$ receives at most $\delta/m$ probability mass under the next \ftrl update.
We then leverage this fact in \cref{property:main} to preclude transitions to periods other than $k+1$.

\begin{lemma}[Uniformly small probability under linear gap growth]
\label{thm:uniform_small_prob_fixed_action}
Let $a_{i} \in \cA_i$ and suppose that action a increases its gap by at least $1$ in each round, i.e., $\gap{i}{t+1}{a_{i}} \ge \gap{i}{t}{a_{i}} + 1$ for $t \geq 1$. Then, it holds
\begin{equation}
\vx_{i}^{(t+1)}[a_{i}]
\le
\frac{\regcon}{\eta^{(t)}\,\gap{i}{t}{a_{i}}}
\le
\frac{\delta}{m}, 
\quad
\text{
for every $t\ge 1$
}.
\end{equation}
\end{lemma}

\begin{proof}
The proof relies only on \cref{lem:regret_gap}, which implies that for every $t\ge 1$,
\[
\vx_{i}^{(t+1)}[a_{i}] \;\le\; \frac{\regcon}{\eta^{(t)}\,\gap{i}{t}{a_{i}}}.
\]
By \eqref{eq:gamma_init} in \cref{asmpt:initial_conditions}, the initial gap $\gamma^{(1)}$ satisfies
$\gap{i}{1}{a_{i}} = \gamma^{(1)} \ge \frac{\regcon m}{\delta}$, which ensures the desired bound on $\vecx{i}{2}[a_i]$ at round $2$.
To control all subsequent rounds under $\eta^{(t)}=t^{-\alpha}$, it is enough to guarantee that
\begin{equation}
\label{eq:gap_lb_eq1}
\gap{i}{t}{a_i} \;\ge\; \frac{\regcon m}{\delta}\,t^\alpha
\qquad \forall\, t\ge 1,
\end{equation}
since this immediately implies $\vx^{(t+1)}[a_i] \le \delta/m$ for all $t\ge 1$.
To establish \eqref{eq:gap_lb_eq1}, we use the assumed linear growth of the gap:
\(
\gap{i}{t}{a_i}
\ge
\gap{i}{1}{a_i} + (t-1)
=
\gamma^{(1)} + (t-1)
\text{for all } t\ge 1.
\)
Thus, a sufficient condition for \eqref{eq:gap_lb_eq1} is
\(
\gamma^{(1)} + (t-1) \ge \frac{\regcon m}{\delta}\,t^\alpha
\;
\)
for all $t\ge 1$.
Although the linear term on the left eventually dominates the sublinear term on the right, we need the bound to hold uniformly from the start; this is ensured by $\gamma^{(1)}$ from \eqref{eq:gamma_init} in \cref{asmpt:initial_conditions}. Since $t-1\le t$ for all $t\ge 1$, it suffices to require
\begin{equation}
\label{eq:uniform_sufficient}
\gamma^{(1)} \;\ge\; \frac{\regcon m}{\delta}\,t^\alpha - t
\qquad \forall\, t\ge 1.
\end{equation}
The tightest sufficient condition in \eqref{eq:uniform_sufficient} is obtained by maximizing the right-hand side over $t\ge 1$.
Let $c \defeq \regcon m/\delta$ and define $f(t)\defeq c t^\alpha - t$ for $t\ge 1$.
We view $t$ as a continuous variable and upper bound $\sup_{t\ge 1} f(t)$ by maximizing $f$ over $t\in[1,\infty)$.
Since $\alpha\in [0,1)$, $f$ is concave on $(0,\infty)$ and hence attains its maximum at a unique critical point.
Differentiating yields $f'(t)=c\alpha t^{\alpha-1}-1$, so the maximizer satisfies
\(
t^\star = (c\alpha)^{\frac{1}{1-\alpha}}.
\)
Evaluating $f$ at $t^\star$, we obtain
\[
f(t^\star)
=
c(t^\star)^\alpha - t^\star
=
c (c\alpha)^{\frac{\alpha}{1-\alpha}}
-
(c\alpha)^{\frac{1}{1-\alpha}}
=
\left(
c^{\frac{1}{1-\alpha}}
\right)
\alpha^{\frac{\alpha}{1-\alpha}}
(1-\alpha)
\]
Define $\gamma^\star \defeq f(t^\star)$. Then, by \eqref{eq:gamma_init} in \cref{asmpt:initial_conditions}, we have $\gamma^{(1)} \ge \gamma^\star$, and therefore \eqref{eq:uniform_sufficient} holds. This implies \eqref{eq:gap_lb_eq1}, and consequently
\[
\vx_{i}^{(t+1)}[a_i] \le \frac{\delta}{m}
\qquad\text{for all } t\ge 1.
\]


\end{proof}



\MainProp*

\begin{proof}

Assume that $k$ is even. We prove by induction on the period index that the \ftrl dynamics satisfy the requirements of \cref{def:transition_period}.

By \cref{def:transition_period}, period~$3$ consists of the single round~$1$, \textit{i.e.}, $\tlow{3}=\thigh{3}=1$. This case is already verified and is included purely by convention. At round~$2$, \cref{asmpt:initial_conditions} shows us 
\begin{equation}
\vx_{1}^{(2)}[a_1(3)]
+ 
\vx_{1}^{(2)}[a_1(4)] 
\;\ge\;
1-\delta,
\qquad
\vx_{2}^{(2)}[a_2(4)] 
\;\ge\;
1-\delta.
\label{eq:period3_base}
\end{equation}
Hence, period~$4$ begins at round $\tlow{4}=2$, and the defining conditions are satisfied at its start.

Assume now that periods $3,4,\dots,k-1$ satisfy \cref{def:transition_period}, and consider the initial round $\tlow{k}$ of period~$k$.
By \cref{def:transition_period}, we have
\begin{align}
\vx_{2}^{(\tlow{k})}[a_2(k)] \;\ge\; 1-\delta,
\label{eq:periodk_start2}
\\
\vx_{1}^{(\tlow{k})}[a_1(k-1)] + \vx_{1}^{(\tlow{k})}[a_1(k)] \;\ge\; 1-\delta.
\label{eq:periodk_start3}
\end{align}
It therefore suffices to show that the period-$k$ conditions continue to hold at every round
$t\in\{\tlow{k},\tlow{k}+1,\dots,\thigh{k}\}$.
Our main tools are \cref{thm:uniform_small_prob_fixed_action} and the gap-to-probability bound \cref{lem:regret_gap}.

\smallskip

\noindent\textbf{Player~1.}
By \cref{lem:regret_gap_increase_actions_1} and \cref{lemma:aux_decrease_by_constant}, every action $a_1 \in \cA_1(k+2)$ increases its gap by at least $(k-2)$ at each round $t \in [1,\tlow{k}]$.
Consequently, \cref{thm:uniform_small_prob_fixed_action} implies that every $a_1 \in \cA_1(k+2)$ receives negligible probability mass at the next update:
\[
\vx_1^{(\tlow{k}+1)}[a_1] \;\le\; \frac{\delta}{m}.
\]

Next, consider actions $a_1 \in [m]\setminus \cA_1(k-2)$.
Although these actions may have been active in earlier periods, they have remained \emph{inactive} (\textit{i.e.}, assigned probability at most $\delta/m$) throughout the last two periods, $(k-2)$ and $(k-1)$.
By \cref{lem:regret_gap_increase_actions_1} and \cref{lemma:aux_decrease_by_constant}, each such $a_1$ increases its gap by at least $(k-2)$ per round over the interval $[\tlow{k-2},\thigh{k-1}]$.
Moreover, \cref{lem:lower_bound_period_k} ensures that the length of this interval is linear in the elapsed horizon:
\[
T_{k-1} + T_{k-2}
\;\ge\;
\frac{1}{4} \sum_{\ell=3}^{k-1} T_{\ell}
=
\frac{1}{4}\,\thigh{k-1}.
\]
Therefore, the cumulative gap of any $a_1 \in [m]\setminus \cA_1(k-2)$ is still linear in time, and another application of \cref{thm:uniform_small_prob_fixed_action} yields the uniform bound
\[
\vx_1^{(\tlow{k}+1)}[a_1] \;\le\; \frac{\delta}{m}.
\]


It remains to control the evolution of the competing actions $a_1(k-1)$ and $a_1(k)$.
By the definition of the gap, we have
\begin{align}
\Delta \gap{1}{\tlow{k}}{a_1(k)}
&\defeq
\vu_1^{(\tlow{k})}[a_1(k-1)] - \vu_1^{(\tlow{k})}[a_1(k)]
\nonumber\\
&=
\left(
\vx_2^{(\tlow{k})}[a_2(k-1)] (k-1)
+
\vx_2^{(\tlow{k})}[a_2(k-2)] (k-2)
\right)
\notag
\\
&\quad
-
\left(
\vx_2^{(\tlow{k})}[a_2(k)] k
+
\vx_2^{(\tlow{k})}[a_2(k+1)] (k+1)
\right)
.
\label{eq:gap_update_raw}
\end{align}
Since $k$ is even, \cref{prop:matrixA_prop3_odd} in \cref{lem:matrixA_props} implies $a_2(k-1) = a_2(k)$, and therefore \eqref{eq:gap_update_raw} simplifies to
\begin{equation}
\Delta \gap{1}{\tlow{k}}{a_1(k)}
=
-\vx_2^{(\tlow{k})}[a_2(k)]
-\vx_2^{(\tlow{k})}[a_2(k+1)] (k+1)
+\vx_2^{(\tlow{k})}[a_2(k-2)] (k-2).
\label{eq:gap_update_simplified}
\end{equation}
By \eqref{eq:periodk_start2}, $\vx_2^{(\tlow{k})}[a_2(k)] \ge 1-\delta$, so the total probability mass on columns other than $a_2(k)$ is at most $\delta$.
Using \eqref{eq:gap_update_simplified} and dropping the negative term $-\vx_2^{(\tlow{k})}[a_2(k+1)](k+1)$, we obtain
\begin{align}
\Delta \gap{1}{\tlow{k}}{a_1(k)}
&\le
-\vx_2^{(\tlow{k})}[a_2(k)] + \delta (k-2)
\nonumber\\
&\le
-(1-\delta) + \delta(k-2)
=
-1 + \delta(k-1).
\nonumber
\end{align}
Under the bounds $k \le 2m+1$ and $\delta = \frac{1}{4m}$, we have $\delta(k-1)\le \frac12$, and therefore
\begin{equation}
\Delta \gap{1}{\tlow{k}}{a_1(k)} \;\le\; -\frac12.
\label{eq:val_eq4}
\end{equation}
In particular, as long as $\vx_2^{(\tlow{k})}[a_2(k)] \ge 1-\delta$, the gap of $a_1(k)$ decreases by at least $1/2$ per round at time $\tlow{k}$.

\smallskip

\noindent\textbf{Player~2.}
The same reasoning applies symmetrically to Player~2.
In particular, every action in $\cA_2 \setminus \{a_2(k),a_2(k+1)\}$ continues to have probability at most $\delta/m$ by another application of \cref{thm:uniform_small_prob_fixed_action}.
The argument is identical to the one for Player~1 after splitting the actions into (i) the \emph{future} actions $\cA_2(k+2)$ and (ii) the actions \emph{active} in \emph{earlier} periods, namely $[m] \setminus  \cA_2(k-2)$.

The remaining nontrivial case concerns the competing actions $a_2(k)$ and $a_2(k+1)$.
Since $a_2(k)$ has the largest probability mass, \cref{lem:order_preservation} implies that it also attains the maximal cumulative utility among actions in $\cA_2$.
Thus, the one-step update of $\Gap^{(\tlow{k})}[a_2(k+1)]$ is
\begin{align*}
\Delta \Gap^{(\tlow{k})}[a_2(k+1)]
&\;\defeq\;
\vu_2^{(\tlow{k})}[a_2(k)]
-
\vu_2^{(\tlow{k})}[a_2(k+1)]
\\
&\;=\;
\vx_1^{(\tlow{k})}[a_1(k)] k
+
\vx_1^{(\tlow{k})}[a_1(k-1)] (k-1)
\\
&\quad
-
\vx_1^{(\tlow{k})}[a_1(k+1)] (k+1)
-
\vx_1^{(\tlow{k})}[a_1(k+2)] (k+2).
\end{align*}
Since $k$ is even, \cref{prop:matrixA_prop2_even} in \cref{lem:matrixA_props} implies $a_1(k)=a_1(k+1)$, and therefore
\begin{equation}
\Delta \Gap^{(\tlow{k})}[a_2(k+1)]
=
- 
\vx_1^{(\tlow{k})}[a_1(k)]
+
\vx_1^{(\tlow{k})}[a_1(k-1)] (k-1)
-
\vx_1^{(\tlow{k})}[a_1(k+2)] (k+2).
\label{eq:val_player2eq}
\end{equation}

The sign of $\Delta \Gap^{(\tlow{k})}[a_2(k+1)]$ is not immediate, as it depends on the relative magnitudes of $\vx_1^{(\tlow{k})}[a_1(k)]$ and $\vx_1^{(\tlow{k})}[a_1(k-1)]$, while $\vx_1^{(\tlow{k})}[a_1(k+2)]$ is uniformly small.
Nevertheless, we can extract a meaningful lower bound.
As long as $\vx_1^{(\tlow{k})}[a_1(k-1)] \ge \vx_1^{(\tlow{k})}[a_1(k)]$ (which holds at round $\tlow{k}$ by the definition of the period), \eqref{eq:periodk_start3} implies
\[
\vx_1^{(\tlow{k})}[a_1(k-1)]
\;\ge\;
\frac{1}{2}(1-\delta).
\]
Using this in \eqref{eq:val_player2eq} and the bound $\vx_1^{(\tlow{k})}[a_1(k+2)] \le \delta$ yields
\begin{equation*}
\Delta \Gap^{(\tlow{k})}[a_2(k+1)]
\;\ge\;
\frac{1}{2}(1-\delta) \, (k-2)
-
\delta (k+2).
\end{equation*}
Under our standing choice of parameters and \cref{lemma:aux_decrease_by_constant}, the right-hand side is at least $(k-3)/2$.
Hence, since $\vx_1^{(\tlow{k})}[a_1(k-1)] \ge \vx_1^{(\tlow{k})}[a_1(k)]$, the gap increases by at least the constant $(k-3)/2$, \textit{i.e.},
\begin{equation}
\Delta \Gap^{(\tlow{k})}[a_2(k+1)]
\;\ge\;
\frac{(k-3)}{2}.
\label{eq:val_eq5}
\end{equation}



\paragraph{Probability mass transfer from $a_1(k-1)$ to $a_1(k)$.}
As in \eqref{eq:val_eq5}, we can show that as long as
$\vx_1^{(t)}[a_1(k-1)] \ge \vx_1^{(t)}[a_1(k)]$, the gap of
$a_2(k+1)$ increases by at least $\frac{k-3}{2}$. This inequality, however,
cannot hold for the entire period. Indeed, by \eqref{eq:val_eq4}, the gap of
$a_1(k)$ decreases by at least $\frac12$ in every round in which $\vx_2^{(t)}[a_2(k)] \ge 1-\delta$,
and therefore the probability mass on $a_1(k)$ increases over time. As a result,
there exists a first round $\hat{t}_k$ at which the cumulative utility of $a_1(k)$ overtakes
that of $a_1(k-1)$, \textit{i.e.},
\[
\sum_{t=1}^{\hat{t}_k} \vu_1^{(t)}[a_1(k)]
\;\ge\;
\sum_{t=1}^{\hat{t}_k} \vu_1^{(t)}[a_1(k-1)].
\]

After $\hat{t}_k$,  an identical argument to \eqref{eq:val_eq4} implies that the gap $\gap{1}{t}{a_1(k-1)}$ increases by at least $1/2$ per round. Our goal is to show that when $\vx_1^{(\tlow{k+1})}[a_1(k)] \ge 1-\delta$, we still have $\vx_2^{(\tlow{k+1})}[a_2(k)] \ge 1-\delta$.
Starting from round $\hat{t}_k$, consider $\tau$ additional rounds until round $\hat{t}_k + \tau$.
To track the evolution, we outline the timeline of the gaps for $a_1(k-1), a_1(k)$; see \cref{fig:gaps}.

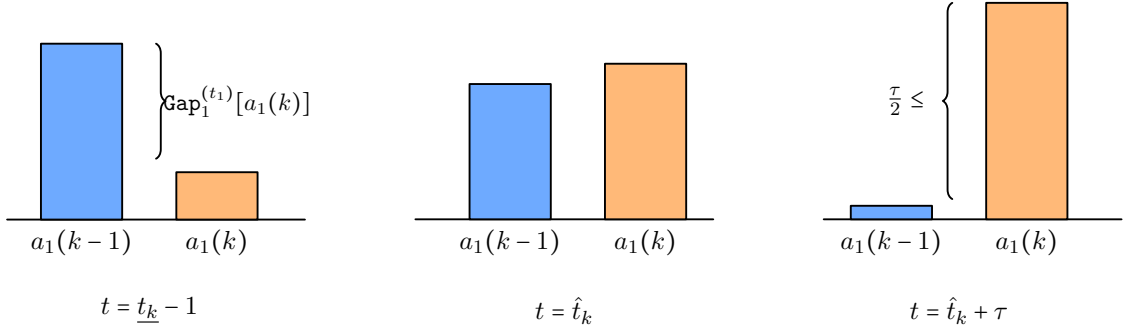
\begin{figure*}[t]
  \centering
  \begin{subfigure}{0.22\linewidth}\centering
    \scalebox{0.9}{
\definecolor{LeftBar}{RGB}{110,170,255}   
\definecolor{RightBar}{RGB}{255,185,120}  
\tikzset{
  leftbar/.style={draw=black, fill=LeftBar},
  rightbar/.style={draw=black, fill=RightBar},
}

\begin{tikzpicture}[x=1cm,y=1cm, line cap=round, line join=round, thick]
  \draw (-0.2,0) -- (4.2,0);

  \draw[leftbar]  (0.3,0) rectangle (1.5,2.6);
  \draw[rightbar] (2.3,0) rectangle (3.5,0.7);

  \draw[decorate,decoration={brace,amplitude=4pt,mirror}] (2,0.9) -- (2,2.6)
    node[midway,xshift=1.2cm] {\small{$\texttt{Gap}_{1}^{(t_1)}[a_1(k)]$}};

  \node[below] at (0.9,0) {$a_{1}(k-1)$};
  \node[below] at (2.9,0) {$a_{1}(k)$};

  \node[below] at (1.9,-1.0) {$t = \underline{t_k}-1$};
\end{tikzpicture}}
  \end{subfigure}
  \hspace{4em}
  \begin{subfigure}{0.22\linewidth}\centering
    \scalebox{0.9}{
\definecolor{LeftBar}{RGB}{110,170,255}   
\definecolor{RightBar}{RGB}{255,185,120}  
\tikzset{
  leftbar/.style={draw=black, fill=LeftBar},
  rightbar/.style={draw=black, fill=RightBar},
} 

\begin{tikzpicture}[x=1cm,y=1cm, line cap=round, line join=round, thick]
  \draw (-0.2,0) -- (4.2,0);

  \draw[leftbar]  (0.6,0) rectangle (1.8,2);
  \draw[rightbar] (2.6,0) rectangle (3.8,2.3);

  \node[below] at (1.2,0) {$a_{1}(k-1)$};
  \node[below] at (3.2,0) {$a_{1}(k)$};

  \node[below] at (2.0,-1.0) {$t = \hat{t}_k$};
\end{tikzpicture}}
  \end{subfigure}
  \hspace{4em}
  \begin{subfigure}{0.22\linewidth}\centering
    \scalebox{0.9}{
\definecolor{LeftBar}{RGB}{110,170,255}   
\definecolor{RightBar}{RGB}{255,185,120}  
\tikzset{
  leftbar/.style={draw=black, fill=LeftBar},
  rightbar/.style={draw=black, fill=RightBar},
}

\begin{tikzpicture}[x=1cm,y=1cm, line cap=round, line join=round, thick]
  \draw (-0.2,0) -- (4.2,0);

  \draw[leftbar]  (0.2,0) rectangle (1.4,0.2);
  \draw[rightbar] (2.2,0) rectangle (3.4,3.2);

  \node[below] at (0.8,0) {$a_{1}(k-1)$};
  \node[below] at (2.8,0) {$a_{1}(k)$};

  \draw[decorate,decoration={brace,amplitude=4pt}] (1.7,0.3) -- (1.7,3.2)
    node[midway,xshift=-0.5cm] {$ \frac{\tau}{2} \leq \ \ \ $};

  \node[below] at (1.8,-1.0) {$t = \hat{t}_k + \tau$};
\end{tikzpicture}}
  \end{subfigure}
  \caption{Probability mass shift between the two competing actions. Transitioning from $a_1(k-1)$ to $a_1(k)$ takes a long time when the initial gap is large. Eventually $a_2(k)$ will surpass $a_1(k)$, and this will trigger a shift in the other player's cumulative utility. The consistency of the argument requires bounding the time it takes for $a_2(k)$ to be played with high probability.}
  \label{fig:gaps}
\end{figure*}
\begin{enumerate}
\item \textbf{Time $\tlow{k}-1$.} The cumulative utility of $a_1(k)$
is smaller than that of $a_1(k-1)$.

\item \textbf{Time $\hat{t}_k$.} The cumulative utility of $a_1(k)$
first exceeds that of $a_1(k-1)$, so $a_1(k)$ becomes the action of maximal
cumulative utility.

\item \textbf{Time $\hat{t}_k + \tau$.} The gap of $a_1(k-1)$ has increased to at least $\frac{\tau}{2}$.
\end{enumerate}
Then,
\[
\gap{1}{\hat{t}_k + \tau}{a_1(k-1)}
\;\ge\;
\gap{1}{\hat{t}_k}{a_1(k-1)} + \frac{\tau}{2}
\;\ge\;
\frac{\tau}{2}.
\]
We next bound the smallest $\tau$ such that, after $\tau$ additional rounds, $a_1(k-1)$ becomes inactive, \textit{i.e.},
$\vx_1^{(t_2+\tau+1)}[a_1(k-1)] \le \frac{\delta}{m}$.
By \cref{lem:regret_gap}, it suffices to require
\begin{equation}
\vx_1^{(\hat{t}_k+\tau+1)}[a_1(k-1)]
\;\le\;
\frac{\regcon}{\eta^{(\hat{t}_k+\tau)}\,\gap{1}{\hat{t}_k}{a_1(k-1)}}
\;\le\;
\frac{\regcon}{\eta^{(\hat{t}_k+\tau)}\,(\tau/2)}
\;\le\;
\frac{\delta}{m}.
\label{eq:prob_small_sufficient}
\end{equation}
Setting $c \defeq \frac{\regcon m}{\delta}$ and using $\eta^{(t)}=t^{-\alpha}$, \eqref{eq:prob_small_sufficient} is equivalent to
\begin{equation}
\frac{(\hat{t}_k+\tau)^\alpha}{\tau} \;\le\; \frac{1}{2c}.
\label{eq:period_val_eq3}
\end{equation}
A sufficient condition is obtained by bounding
\(
(\hat{t}_k+\tau)^\alpha
\le
2^\alpha
\left(
\hat{t}_k^\alpha+\tau^\alpha
\right)
\le
2
\left(
\hat{t}_k^\alpha+\tau^\alpha
\right),
\)
for $\alpha \in [0,1)$.
Substituting into \eqref{eq:period_val_eq3} concludes the sufficient requirement
\(
(\hat{t}_k^\alpha / \tau)
+
\tau^{\alpha-1}
\le
(1/4c).
\)
To simplify the algebra, we impose the stronger pair of inequalities
$\hat{t}_k^\alpha / \tau \le 1/(8c)$ and $\tau^{\alpha-1} \le 1/(8c)$.
These are equivalent to
\[
\tau \ge 8c \, \hat{t}_k^\alpha
\qquad
\text{and}
\qquad
\tau \ge (8c)^{\frac{1}{1-\alpha}},
\]
since $\alpha-1<0$ implies $\tau^{1-\alpha}\ge 8c$.
Therefore,
\[
\tau 
\;=\;
\max\!\left\{
(8c)\,\hat{t}_k^\alpha,\; (8c)^{\frac{1}{1-\alpha}}
\right\}
.
\]
Moreover, \cref{lem:period_2_lower_bound} together with \eqref{eq:gamma_init} implies
$\hat{t}_k \ge \tlow{k} \ge T_4 \ge \gamma^{(1)} / 2 \ge (64c)^{\frac{1}{1-\alpha}}$.
Hence, the maximum is attained by the first term.
In particular, this shows that it takes at most $\tau$ rounds for $\vx_1^{(t)}[a_1(k-1)]$ to become inactive, where
\begin{equation}
\tau 
\defeq
(8c)\,\hat{t}_k^\alpha
=
\left(\frac{8\,\regcon m}{\delta}\right)\hat{t}_k^\alpha.
\label{eq:tau_def}
\end{equation}

It remains to show that, even after these $\tau$ rounds, Player~2 still assigns probability at least $1-\delta$ to action $a_2(k)$.
As argued previously for Player~2, every action other than $a_2(k+1)$ accumulates a gap that is linear in $\hat{t}_k + \tau$; hence, by \cref{thm:uniform_small_prob_fixed_action}, each such action remains assigned probability at most $\delta/m$.
Therefore, if $\vx_2^{(\hat{t}_k + \tau)}[a_2(k)]$ were to drop below $1-\delta$, then only $a_2(k+1)$ could have increased its probability.

To this end, we lower bound the gap of the competing action $a_2(k+1)$.
As noted above, the implication in \eqref{eq:val_eq5} does not hold after round $\hat{t}_k$. We therefore bound the worst-case one-step change in the gap of $a_2(k+1)$. From \eqref{eq:val_player2eq}, for any $t\ge \hat{t}_k$,
\begin{align}
\Delta \Gap^{(t)}[a_2(k+1)]
&=
- 
\vx_1^{(t)}[a_1(k)]
+
\vx_1^{(t)}[a_1(k-1)] (k-1)
-
\vx_1^{(t)}[a_1(k+2)] (k+2)
\notag
\\
&\geq
- 
\vx_1^{(t)}[a_1(k)]
-
\vx_1^{(t)}[a_1(k+2)] (k+2)
\notag
\\
& 
\geq
-
1
- (\delta / m) (k+2)
\geq -2
.
\label{eq:val_eq6}
\end{align}
By \cref{lem:lower_bound_regret_gap}, \eqref{eq:val_eq5} and \eqref{eq:val_eq6}, we obtain
\begin{align}
\gap{2}{\hat{t}_k + \tau}{a_2(k+1)}
&=
\gap{2}{\hat{t}_k}{a_2(k+1)} 
+
\sum_{t=\hat{t}_k+1}^{\hat{t}_k + \tau}
\Delta \Gap^{(t)}[a_2(k+1)]
\notag\\
&\ge
\gap{2}{\tlow{k}-1}{a_2(k+1)} 
+
\sum_{t=\tlow{k}}^{\hat{t}_k}
\Delta \Gap^{(t)}[a_2(k+1)]
+
\sum_{t=\hat{t}_k+1}^{\hat{t}_k + \tau}
(-2)
& (\text {From \eqref{eq:val_eq6}})
\notag\\
&\ge
\gap{2}{\tlow{k}-1}{a_2(k+1)}
+
\sum_{t=\tlow{k}}^{\hat{t}_k}\frac{(k-3)}{2}
+
\tau (-2)
& (\text {From \eqref{eq:val_eq5}})
\notag\\
&\ge
\sum_{\ell=4}^{k-1} (\ell-2) T_{\ell}
+
\frac{(\hat{t}_k - \tlow{k} + 1)}{2}
-2 \tau
& (\text {From \eqref{lem:lower_bound_regret_gap}})
\notag\\
&\ge
\sum_{\ell=4}^{k-1} 
T_{\ell}
+
\frac{(\hat{t}_k - \tlow{k} + 1)}{2}
-
2 \tau
\notag\\
&\geq
\frac{\tlow{k}-1}{2}
+
\frac{(\hat{t}_k - \tlow{k} + 1)}{2}
-
2 \tau
\notag\\
&=
\frac{\hat{t}_k}{2}
-
(16 c ) \,\hat{t}_k^{\alpha},
\label{eq:gap_lower_t2_minus_sublinear}
\end{align}
where in the last step we used $\tau = (8 c) \,\hat{t}_k^\alpha$ from \eqref{eq:tau_def}.

The correction term $(8c)\,\hat{t}_k^\alpha$ is sublinear in $\hat{t}_k$ for $\alpha \in [0,1)$, so after a particular round the linear term should dominates the second term. It suffices to find $t$ such as the following is true
\begin{equation}
\frac{t}{2} - (16c)\,t^\alpha \;\ge\; \frac{t}{4},
\label{eq:val_eq7}
\end{equation}
which is equivalent to
\(
\frac{t}{4} 
\ge
(16c)\,t^\alpha
\Leftrightarrow
t 
\ge
(64c)^{\frac{1}{1-\alpha}}.
\)
Therefore, \eqref{eq:val_eq7} holds whenever $\hat{t}_k \ge (64c)^{\frac{1}{1-\alpha}}$. Moreover, due to \cref{lem:period_2_lower_bound} and \eqref{eq:gamma_init}
\begin{equation}
\hat{t}_k 
\geq
\sum_{\ell=3}^{k-1} T_{\ell} 
\geq
T_4 
\geq
\frac{\gamma^{(1)}}{2} 
\geq
(64c)^{\frac{1}{1-\alpha}},
\label{eq:val_eq8}
\end{equation}
so from \eqref{eq:gap_lower_t2_minus_sublinear} we conclude
\begin{equation}
\gap{2}{\hat{t}_k + \tau}{a_2(k+1)} \;\ge\; \frac{\hat{t}_k}{4}.
\label{eq:val_eq10}
\end{equation}
%
The final step is to verify that $\vx_2^{(\hat{t}_k + \tau)}[a_2(k+1)] \le \delta/m$. This ensures that the conditions of \cref{def:transition_period} are satisfied, and hence periods do not blend despite the propensity of \ftrl to mix actions; in particular, period $k$ can only be followed by period $k+1$. To this end, we invoke \cref{lem:regret_gap} together with \eqref{eq:val_eq10}.
\begin{align}
\vecx{2}{\hat{t}_k + \tau}[a_2(k+1)] 
&\;\le\;
\frac{\regcon}{\eta^{(\hat{t}_k + \tau)} \gap{2}{\hat{t}_k + \tau}{a_2(k+1)}}
\;\le\;
\frac{\regcon}{(\hat{t}_k + \tau)^{-\alpha} (\hat{t}_k/4)}
=
\frac{4R \; (\hat{t}_k + \tau)^{\alpha}}{\hat{t}_k}
\leq 
\frac{\delta}{m}
\label{eq:val_eq9}
\end{align}

Hence, it suffices to show that
\(
\frac{(\hat{t}_k+\tau)^{\alpha}}{\hat{t}_k} \le \frac{1}{4c},
\) where $c = \frac{\regcon m}{\delta}$.
From \eqref{eq:period_val_eq3}, we already have $\frac{(\hat{t}_k+\tau)^\alpha}{\tau} \le \frac{1}{2c}$. Therefore,
\[
\frac{(\hat{t}_k+\tau)^{\alpha}}{\hat{t}_k}
=
\frac{(\hat{t}_k+\tau)^{\alpha}}{\tau}
\frac{\tau}{\hat{t}_k}
\le
\frac{1}{2c}
\frac{\tau}{\hat{t}_k}.
\]
Thus, it remains to ensure $\tau/\hat{t}_k \le 1/2$, which, in turn, implies \eqref{eq:val_eq9} and concludes the proof.
By \eqref{eq:tau_def}, $\tau = (8c)\,\hat{t}_k^{\alpha}$, and hence
\(
\tau/\hat{t}_k = 8c\,\hat{t}_k^{-(1-\alpha)} = 8c/\hat{t}_k^{\,1-\alpha}.
\)
Therefore, the condition $\tau/\hat{t}_k \le 1/2$ is equivalent to
\[
\frac{\tau}{\hat{t}_k}
\le
\frac12
\;\Longleftrightarrow\;
\frac{8c}{\hat{t}_k^{\,1-\alpha}}\le \frac12
\;\Longleftrightarrow\;
\hat{t}_k^{\,1-\alpha}\ge 16c
\;\Longleftrightarrow\;
\hat{t}_k \ge (16c)^{\frac{1}{1-\alpha}},
\]
where the last step uses $1-\alpha>0$. Since \eqref{eq:val_eq8} guarantees \(\hat{t}_k \ge (64c)^{\frac{1}{1-\alpha}}\), the above condition is satisfied; hence \eqref{eq:val_eq9} holds as well, completing the proof.
\end{proof}

\subsection{Proof of \cref{theorem:negative}}

Having established in \cref{property:main} that the \ftrl iterates traverse the periods sequentially (from $k=3$ up to $k=2m-1$), we can now invoke the preceding results to prove the main theorem of \cref{sec:lowerbounds}, namely \cref{theorem:negative}. We begin with a lemma showing that, in every period except the last, the \ftrl iterates do not reach an $\epsilon$-Nash equilibrium for any $\epsilon \le \frac{1}{8m}$.

Throughout, we work with the identical-payoff game $(\mA,\mA)$ defined in \cref{subsec:init}, and we fix $\delta$ and $\gamma^{(1)}$ as in \cref{asmpt:initial_conditions}.

\begin{lemma}
\label{lem:beneficial_deviation}
For every $k\in\{3,\dots,2m-2\}$, throughout period $k$ the \ftrl iterates $(\vx_1^{(t)},\vx_2^{(t)})$ are not an $\epsilon$-Nash equilibrium for $\epsilon=\frac{1}{8m}$.
\end{lemma}

\begin{proof}
Fix a round $t$ in period $k$, where $3\le k\le 2m-2$ and $k$ is even.
By \cref{def:transition_period},
\[
\vx_{1}^{(t)}[a_1(k-1)] + \vx_{1}^{(t)}[a_1(k)] \ge 1-\delta,
\qquad
\vx_{2}^{(t)}[a_2(k)] \ge 1-\delta .
\]
Then,
\[
\sum_{a_1\notin\{a_1(k-1),a_1(k)\}} \vx_{1}^{(t)}[a_1]\le \delta,
\qquad
\sum_{a_2\neq a_2(k)} \vx_{2}^{(t)}[a_2]\le \delta .
\]
Using the uniform bound $\mA[a_1,a_2]\le 2m-1$ from \cref{prop:matrixA_prop1_max} in \cref{lem:matrixA_props}, we upper bound the realized payoff as
\begin{align}
\langle \vx_1^{(t)}, \mA \vx_2^{(t)} \rangle
&=
\sum_{a_1,a_2}\vx_{1}^{(t)}[a_1]\vx_{2}^{(t)}[a_2]\mA[a_1,a_2]
\nonumber
\\
&\le \delta^2(2m-1)
+ \vx_2^{(t)}[a_2(k)]
\Bigl(
    \vx_1^{(t)}[a_1(k-1)]\,\mA[a_1(k-1),a_2(k)] 
\nonumber
\\
&\hphantom{\le \delta^2(2m-1)
+ \vx_2^{(t)}[a_2(k)]\Bigl(}
  + \vx_1^{(t)}[a_1(k)]\,\mA[a_1(k),a_2(k)]
\Bigr).
\label{eq:payoff-upper-clean-pre}
\end{align}
Since $k$ is even, \cref{prop:matrixA_prop3_odd} in \cref{lem:matrixA_props} implies $a_2(k-1)=a_2(k)$, and hence
\[
\mA[a_1(k-1),a_2(k)]
=
\mA[a_1(k-1),a_2(k-1)]
=
k-1.
\]
Substituting this identity and $\mA[a_1(k),a_2(k)]=k$ into \eqref{eq:payoff-upper-clean-pre} yields
\begin{align}
\langle \vx_1^{(t)}, \mA \vx_2^{(t)} \rangle
&\le
\delta^2(2m-1)
+
\vx_{2}^{(t)}[a_2(k)]
\left(
\vx_{1}^{(t)}[a_1(k-1)] \, (k-1)
+
\vx_{1}^{(t)}[a_1(k)] \, k
\right).
\label{eq:payoff-upper-clean}
\end{align}

Consider Player~1 deviating to the pure action $a_1(k)$. By definition,
$\mA[a_1(k),a_2(k)]=k$ and all other entries, but the $m$th entry, are nonnegative, hence
\begin{equation}
\innerprod{\ve_{a_1(k)}}{\mA \vecx{2}{t}}
\ge 
\vecx{2}{t}[a_2(k)] k 
+
\vecx{2}{t}[m] \, 
\left(
- 2 \gamma^{(1)}
\right).
\label{eq:p1-dev-clean2}
\end{equation}

The probability assigned to the $m$th action can easily be bounded. The key observation is that the gap of action $m$ grows $\gamma^{(1)}$ times faster than the gap of any action whose gap increases linearly with time. In particular,
\[
\gap{2}{t-1}{m}
\;\ge\;
\gamma^{(1)}\,\gap{2}{t-1}{a_2(2m-1)}.
\]
By \cref{lem:regret_gap}, we have
\begin{align}
\vx_2^{(t)}[m]
&\le
\frac{\regcon}{\eta^{(t-1)}\,\gap{2}{t-1}{m}}
\le
\frac{\regcon}{\eta^{(t-1)}\,\gamma^{(1)}\,\gap{2}{t-1}{a_2(2m-1)}}
\le
\frac{1}{\gamma^{(1)}}\cdot \frac{\delta}{m},
\label{eq:p1-dev-clean3}
\end{align}
where the last inequality uses \cref{thm:uniform_small_prob_fixed_action} applied to $a_2(2m-1)$, whose gap is linear in $t-1$.

Combining this bound with \eqref{eq:p1-dev-clean2} yields
\begin{equation}
\langle \ve_{a_1(k)}, \mA \vx_2^{(t)} \rangle
\;\ge\;
\vx_2^{(t)}[a_2(k)] \, k
-
\frac{2\delta}{m}.
\label{eq:p1-dev-clean}
\end{equation}

Since $\vecx{1}{t}[a_1(k-1)] + \vecx{1}{t}[a_1(k)] \le 1$,
\eqref{eq:payoff-upper-clean} can be rewritten as
\begin{equation*}
\langle \vx_1^{(t)}, \mA \vx_2^{(t)} \rangle
\le \delta^2(2m-1)
+ \vx_{2}^{(t)}[a_2(k)]
\Bigl((k-1) + \vx_{1}^{(t)}[a_1(k)]\Bigr).
\end{equation*}

Subtracting it
from \eqref{eq:p1-dev-clean} yields
\begin{align}
\innerprod{\ve_{a_1(k)}}{\mA \vecx{2}{t}}
-
\innerprod{\vecx{1}{t}}{\mA \vecx{2}{t}}
&\ge
-
\delta^2(2m-1)
-
\frac{2\delta}{m}
+
\vecx{2}{t}[a_2(k)]
\left(
1
-
\vecx{1}{t}[a_1(k)]
\right).
\label{eq:p1-gain-master}
\end{align}

If $\vecx{1}{t}[a_1(k)] \le 1-\frac{1}{k+1}$, then $1-\vecx{1}{t}[a_1(k)] \ge \frac{1}{k+1}$ and $\vecx{2}{t}[a_2(k)] \ge 1-\delta$, so
\begin{equation}
\innerprod{\ve_{a_1(k)}}{\mA \vecx{2}{t}}
-
\innerprod{\vecx{1}{t}}{\mA \vecx{2}{t}}
\ge
-\delta^2(2m-1)
-
\frac{2\delta}{m}
+
(1-\delta)\frac{1}{k+1}.
\label{eq:case1-const22}
\end{equation}
Substituting $\delta=\frac{1}{4m}$ and using $k+1\le 2m$ yields
\[
(1-\delta)\frac{1}{k+1}\;\ge\;(1-\delta)\frac{1}{2m}.
\]
Hence
\begin{align*}
-\delta^2(2m-1)-\frac{2\delta}{m}+(1-\delta)\frac{1}{k+1}
&\;\ge\;
-\delta^2(2m-1)-\frac{2\delta}{m}+(1-\delta)\frac{1}{2m}
\nonumber
\\
&\;=\;
-\frac{2m-1}{16m^2}-\frac{1}{2m^2}+\Bigl(1-\frac{1}{4m}\Bigr)\frac{1}{2m}
\nonumber
\\
&\;=\;
-\frac{2m-1}{16m^2}-\frac{1}{2m^2}+\frac{1}{2m}-\frac{1}{8m^2}
\nonumber
\\
&\;=\;
\frac{3}{8m}-\frac{9}{16m^2}
=
\frac{6m-9}{16m^2}
\geq 
\frac{1}{8m}, & \text{(for $m >2$)}
\end{align*}

\begin{equation}
\innerprod{\ve_{a_1(k)}}{\mA \vecx{2}{t}}
-
\innerprod{\vecx{1}{t}}{\mA \vecx{2}{t}}
\ge
\frac{1}{8m}.
\label{eq:case1-const}
\end{equation}

Otherwise, $\vx_1^{(t)}[a_1(k)] > 1-\frac{1}{k+1}$, and therefore
\[
\vx_1^{(t)}[a_1(k-1)]
\;\le\;
1-\vx_1^{(t)}[a_1(k)]
\;<\;
\frac{1}{k+1}.
\]
Consider player~2 deviating to $a_2(k+1)$. Since $k$ is even, \cref{prop:matrixA_prop2_even} (in \cref{lem:matrixA_props}) implies $a_1(k)=a_1(k+1)$. Hence
\[
\mA[a_1(k),a_2(k+1)]
=
\mA[a_1(k+1),a_2(k+1)]
=
k+1.
\]
Moreover, all payoff entries appearing in the deviation comparison are nonnegative, except the $m$th entry, so
\begin{equation}
\innerprod{\vecx{1}{t}}{\mA \ve_{a_2(k+1)}}
\ge 
\vecx{1}{t}[a_1(k)]
\,
(k+1)
+
\vecx{1}{t}[m]
\, 
\left( 
-2 \gamma^{(1)}
\right)
.
\label{eq:p2-dev-clean}
\end{equation}
As before, we can bound $\vx_1^{(t)}[m]\le \delta/(\gamma^{(1)}m)$ (analogously to the bound on $\vx_2^{(t)}[m]$). Using this fact, subtracting \eqref{eq:payoff-upper-clean} from $\langle \vx_1^{(t)}, \mA \ve_{a_2(k+1)}\rangle$, and using $\vx_2^{(t)}[a_2(k)]\le 1$, we obtain
\begin{align}
\langle \vx_1^{(t)}, \mA \ve_{a_2(k+1)} \rangle
-
\langle \vx_1^{(t)}, \mA \vx_2^{(t)} \rangle
&\ge
\vx_1^{(t)}[a_1(k)](k+1)
-
\frac{2\delta}{m}
-
\delta^2(2m-1)
\nonumber\\
&\quad
-
\Bigl(
\vx_1^{(t)}[a_1(k-1)](k-1)
+
\vx_1^{(t)}[a_1(k)]k
\Bigr)
\nonumber\\
&=
-\delta^2(2m-1)
-
\frac{2\delta}{m}
+
\vx_1^{(t)}[a_1(k)]
-
(k-1)\vx_1^{(t)}[a_1(k-1)].
\label{eq:p2-gain-master}
\end{align}
Using $\vx_1^{(t)}[a_1(k)] > 1-\frac{1}{k+1}$ and $\vx_1^{(t)}[a_1(k-1)] \le \frac{1}{k+1}$ yields
$\vx_1^{(t)}[a_1(k)] - (k-1) \vx_1^{(t)}[a_1(k-1)] \ge \frac{1}{k+1}$, and hence from \eqref{eq:p2-gain-master}
\begin{align}
\innerprod{\vecx{1}{t}}{\mA \ve_{a_2(k+1)}}
-
\innerprod{\vecx{1}{t}}{\mA \vecx{2}{t}}
&\ge
-
\delta^2(2m-1)
-
\frac{2\delta}{m}
+
\frac{1}{k+1}
\nonumber
\\
&
\ge
-
\delta^2(2m-1)
-
\frac{2\delta}{m}
+
\frac{(1-\delta)}{k+1}.
\label{eq:case2-const2}
\end{align}

In \eqref{eq:case1-const22} and \eqref{eq:case1-const} we lower bounded the right-hand side by $1/(8m)$. The same bound applies to \eqref{eq:case2-const2}, yielding
\begin{equation}
\innerprod{\vecx{1}{t}}{\mA \ve_{a_2(k+1)}}
-
\innerprod{\vecx{1}{t}}{\mA \vecx{2}{t}}
\;\ge\;
\frac{1}{8m}.
\label{eq:case2-const}
\end{equation}

Combining \eqref{eq:case1-const} and \eqref{eq:case2-const}, every round $t$ of every even period
$k\le 2m-2$ admits a unilateral deviation that improves the payoff by at least $1/(8m)$.
The odd-$k$ case is analogous (with the roles of the players swapped), so no iterate that lies in a period
$k\le 2m-2$ can be an $\epsilon$-Nash equilibrium for any $\epsilon \le \frac{1}{8m}$.
\end{proof}

We now prove the formal version of \cref{theorem:negative} stated in \cref{sec:introduction}.

\begin{theorem}
\label{theorem:negative_formal}
Consider the identical-payoff game $(\mA,\mA)$, and suppose both players run \ftrl with a permutation-invariant regularizer and learning rate $\eta^{(t)}=t^{-\alpha}$ for $\alpha\in[0,1)$. Then, for any $\epsilon \le \frac{1}{8m}$, the \ftrl iterates do not constitute an $\epsilon$-Nash equilibrium for at least $2^{\Omega(m\log m)}$ rounds.
\end{theorem}

\begin{proof}
The claim follows by combining \cref{property:main}, \cref{lem:beneficial_deviation}, and \cref{lem:exponential_lower_bound}.
By \cref{property:main}, the \ftrl iterates sequentially pass through periods $k=3,\dots,2m-3$.
For every round $t$ occurring in these periods, \cref{lem:beneficial_deviation} guarantees the existence of a unilateral deviation that increases payoff by at least $1/(8m)$.
Hence, at every such round the iterate fails to be an $\epsilon$-Nash equilibrium for any $\epsilon\le 1/(8m)$.

It remains to lower bound the time elapsed by the end of period $2m-3$.
By \cref{lem:exponential_lower_bound}, the combined duration of periods $2m-4$ and $2m-3$ satisfies
\[
T_{2m-4}+T_{2m-3}
\;\ge\;
\frac{\gamma^{(1)}}{4}
\left(\frac{m-3}{e}\right)^{m-3}.
\]

To express the dominant term in base~$2$, note that
\[
\left(\frac{m-3}{e}\right)^{m-3}
=
2^{(m-3)\log_2\!\left(\frac{m-3}{e}\right)}
=
2^{(m-3)\left(\log_2(m-3)-\log_2 e\right)}.
\]
For $m\ge 10$, we have $m-3\ge m/2$ and
\[
\log_2(m-3)\;\ge\;\log_2(m/2)\;=\;\log_2 m-1.
\]
Therefore,
\begin{align*}
(m-3)\left(\log_2(m-3)-\log_2 e\right)
&\ge
\frac{m}{2}\left((\log_2 m-1)-\log_2 e\right) \\
&=
\frac{m}{2}\left(\log_2 m - (1+\log_2 e)\right).
\end{align*}
Let $C \coloneqq 1+\log_2 e$. For all $m \ge 2^{2C} (< 10)$ we have $\log_2 m - C \ge \tfrac12 \log_2 m$, and hence
\[
(m-3)\left(\log_2(m-3)-\log_2 e\right)
\;\ge\;
\frac{m}{4}\log_2 m.
\]
Plugging this back yields
\[
\left(\frac{m-3}{e}\right)^{m-3}
\;\ge\;
2^{\frac{m}{4}\log_2 m}
=
2^{\Omega(m\log m)}.
\]
\end{proof}

\section{Properties of $\mB_{r,n}$}
\label{sec:matrix_props}

In this section, we introduce the class of games introduced by~\citet{Panageas23:Exponential}.
Specifically, for $m = 4,6,\ldots$ and $r \in \N$, we define the matrix $\mB_{m,r}$ as follows.
\begin{definition}[Recursive payoff matrix $\mB_{m,r}$]
\label{def:payoffB}
Let $m\ge 2$ be even and $r\in\mathbb N$. Define $\mB_{m,r}\in\mathbb R^{m\times m}$ recursively by
\begin{equation}
\mB_{m,r}
\defeq
\begin{bNiceMatrix}
(r+1) & 0   & \cdots & 0   & 0 \\
0   &     &        &     & (r+4) \\
\vdots &  & \mB_{m-2,r+4} & & \vdots \\
0   &     &        &     & 0 \\
(r+2) & 0   & \cdots & 0   & (r+3)
\end{bNiceMatrix},
\qquad
\text{with base case }
\mB_{2,r}
\defeq
\begin{bNiceMatrix}
r+1 & 0 \\
r+2 & r+3
\end{bNiceMatrix}.
\label{mat:payoffB}
\end{equation}
\end{definition}


\begin{lemma}[Support entries of $\mB_{m,r}$]
\label{lem:support_matrixB}
Let $m\ge 2$ and $r\in\mathbb{N}$ be even.
Then the set of nonzero entries of $\mB_{m,r}$ is exactly
\(
\{r+1,r+2,\dots,r+(2m-1)\},
\)
and each value in this set appears exactly once.
\end{lemma}

\begin{proof}
We proceed by induction on $m$.
For the base case $m=2$, the matrix $\mB_{2,r}$ has exactly three nonzero entries, namely
$\{r+1,\,r+2,\,r+(2\cdot 2-1)\}$, each appearing once.
Assume the claim holds for $m-2$, and consider $\mB_{m,r}$ with $m\ge 4$.
By construction, the only nonzero entries on the outer frame are
\[
r+1 \text{ at } (1,1),\qquad
r+2 \text{ at } (m,1),\qquad
r+3 \text{ at } (m,m),\qquad
r+4 \text{ at } (2,m),
\]
and all other outer entries are zero.
The inner block is $\mB_{m-2,r+4}$, supported on rows and columns in $\{2,\dots,m-1\}$.
By the induction hypothesis, its nonzero entries are exactly
\[
\{(r+4)+1,\dots,(r+4)+(2(m-2)-1)\}
=
\{r+5,\dots,r+(2m-1)\},
\]
each appearing exactly once.
Since $\{r+1,r+2,r+3,r+4\}$ is disjoint from $\{r+5,\dots,r+(2m-1)\}$, all nonzero entries of $\mB_{m,r}$ are distinct, and their union is precisely $\{r+1,\dots,r+(2m-1)\}$.
\end{proof}

\begin{definition}[Locator functions $a_1,a_2$]
Let $m\ge 2$ be even and let $r\in\mathbb{N}$.
For each value $s\in\{r+1,\dots,r+(2m-1)\}$ that appears as an entry of $\mB_{m,r}$, define $(a_1(s),a_2(s)) \in [m]\times[m]$ as the unique pair of indices (guaranteed by \Cref{lem:support_matrixB}) such that
\[
\mB_{m,r}\big[a_1(s),a_2(s)\big]=s.
\]
\end{definition}

\begin{proposition}[Structural properties of $\mB_{m,r}$]
\label{lem:matrixB_props}
Let $m\ge 2$ and $r\in\mathbb{N}$ be even. The following properties hold for $\mB_{m,r}$:
\begin{enumerate}[label=(\roman*)]
\item $\max_{i,j} \mB_{m,r}[i,j] = r+(2m-1)$.
\label{prop:matrixB_prop1}

\item For even $k \in \{r+1,\dots,r+(2m-2)\}$,
\(
a_1(k) = a_1(k+1).
\)
\label{prop:matrixB_prop2}

\item For odd $k \in \{r+1,\dots,r+(2m-2)\}$,
\(
a_2(k) = a_2(k+1).
\)
\label{prop:matrixB_prop3}
\end{enumerate}
\end{proposition}

\begin{proof}
\Cref{prop:matrixB_prop1} is immediate. By \Cref{lem:support_matrixB}, the nonzero entries of $\mB_{m,r}$ are exactly
$\{r+1,r+2,\dots,r+(2m-1)\}$, and hence the maximum is $r+(2m-1)$.

We prove \Cref{prop:matrixB_prop2,prop:matrixB_prop3} simultaneously by induction on $m$.
For $m=2$, and since $r$ is even, we have
$a_2(r+1)=a_2(r+2)=1$ and $a_1(r+2)=a_1(r+3)=2$, establishing the claim.

Assume the claim holds for $m-2$ (with $m\ge 4$), and consider $\mB_{m,r}$.
By construction,
\[
r+1 \text{ is at }(1,1),\quad
r+2 \text{ is at }(m,1),\quad
r+3 \text{ is at }(m,m),\quad
r+4 \text{ is at }(2,m),
\]
and the submatrix on rows/cols $\{2,\dots,m-1\}$ is $\mB_{m-2,r+4}$.

Fix $k \in \{r+1,\dots,r+(2m-2)\}$.
\begin{itemize}
\item If $k \in \{r+1,r+2,r+3\}$, the conclusions follow by inspection:
\[
a_2(r+1)=a_2(r+2)=1,\qquad
a_1(r+2)=a_1(r+3)=m,\qquad
a_2(r+3)=a_2(r+4)=m.
\]
Thus, for odd $k \in \{r+1,r+3\}$ we have $a_2(k)=a_2(k+1)$, and for even $k=r+2$ we have $a_1(k)=a_1(k+1)$.

\item If $k=r+4$, then $r+4$ is at $(2,m)$ and $r+5$ is the top-left entry of the inner block, hence at $(2,2)$.
Therefore $a_1(r+4)=a_1(r+5)=2$, which is exactly \Cref{prop:matrixB_prop2} (since $r+4$ is even).

\item If $k \ge r+5$, then $k$ and $k+1$ are inside the inner block $\mB_{m-2,r+4}$.
Let $(\tilde a_1(\cdot),\tilde a_2(\cdot))$ denote the location functions within that $(m-2)\times(m-2)$ block.
Because the inner block is embedded with a $(+1,+1)$ index shift,
\[
a_1(t)=\tilde a_1(t)+1,\qquad a_2(t)=\tilde a_2(t)+1
\qquad\text{for all } \tilde{r} \in\{r+5,\dots,r+(2m-1)\}.
\]
Applying the induction hypothesis to $\mB_{m-2,r+4}$ yields
$\tilde a_1(k)=\tilde a_1(k+1)$ when $k$ is even, and
$\tilde a_2(k)=\tilde a_2(k+1)$ when $k$ is odd.
Shifting by $+1$ gives the desired equalities for $a_1,a_2$ in $\mB_{m,r}$.
\end{itemize}
This completes the induction and proves \Cref{prop:matrixB_prop2,prop:matrixB_prop3}.
\end{proof}

\section{Auxiliary lemmas}
\label{sec:aux}

\begin{lemma}
\label{lemma:aux_prob_comparison}
Let $\vx_i \in \Delta(\cA_i)$. Suppose there exist $a_i,a_i'\in \cA_i$ such that
\(
\vx_i[a_i]+\vx_i[a_i'] \ge 1-\delta
\)
for $\delta = \frac{1}{4m}$. Then
\[
\max\{\vx_i[a_i],\vx_i[a_i']\} \ge \vx_i[\tilde{a}_i] \quad \text{for all } \tilde{a}_i\in \cA_i.
\]
In particular, either $\vx_i[a_i]$ or $\vx_i[a_i']$ is a maximal coordinate of $\vx_i$.
\end{lemma}

\begin{proof}
Since $\vx_i$ is a probability vector,
\[
\sum_{\tilde a_i\in \cA_i\setminus\{a_i,a_i'\}} \vx_i[\tilde{a}_i]
= 1-(\vx_i[a_i]+\vx_i[a_i'])
\le \delta.
\]
Hence, for every $\tilde{a}_i\in \cA_i\setminus\{a_i,a_i'\}$,
\begin{equation}
\vx_i[\tilde{a}_i]
\le
\sum_{\hat{a}_i \in \cA_i\setminus\{a_i,a_i'\}}
\vx_i[\hat{a}_i] 
\le
\delta.
\label{eq:aux_prob_1}
\end{equation}
On the other hand,
\begin{equation}
\max\{\vx_i[a_i],\vx_i[a_i']\}
\ge \frac{\vx_i[a_i]+\vx_i[a_i']}{2}
\ge \frac{1-\delta}{2}.
\label{eq:aux_prob_2}
\end{equation}
Because $\delta = \frac{1}{4m}$, we have $\frac{1-\delta}{2}\ge \delta$, and therefore from \eqref{eq:aux_prob_1} and \eqref{eq:aux_prob_2}
\[
\max\{\vx_i[a_i],\vx_i[a_i']\} \ge \vx_i[\tilde{a}_i]
\quad \text{for all } \tilde{a}_i \in \cA_i\setminus\{a_i,a_i'\}.
\]
The conclusion is trivial for $\tilde{a}_i \in\{a_i, a_i'\}$, so
\[
\max\{\vx_i[a_i],\vx_i[a_i']\} \ge \vx_i[\tilde{a}_i]
\quad \text{for all } \tilde{a}_i \in \cA_i,
\]
which implies that either $\vx_i[a_i]$ or $\vx_i[a_i']$ is a maximal coordinate of $\vx_i$.
\end{proof}


\begin{lemma}
\label{lem:aux_max_min_increment}
For any player $i$ and any $t\ge 1$, it holds that
\[
\max_{a_i \in \cA_i}\sum_{\tau=1}^{t} \vu_i^{(\tau)}[a_i]
-
\max_{a_i \in \cA_i}\sum_{\tau=1}^{t-1} \vu_i^{(\tau)}[a_i]
\;\ge\;
\min\!\left\{
\vu_i^{(t)}[a_i^{(t-1)}],
\;\vu_i^{(t)}[a_i^{(t)}]
\right\},
\]
where
\[
a_i^{(t)} \in \arg\max_{a_i' \in \cA_i} \sum_{\tau=1}^{t} \vu_i^{(\tau)}[a_i'],
\quad
a_i^{(t-1)} \in \arg\max_{a_i' \in \cA_i} \sum_{\tau=1}^{t-1} \vu_i^{(\tau)}[a_i'].
\]
\end{lemma}

\begin{proof}
Define
\[
M_t \defeq \max_{a_i' \in \cA_i}\sum_{\tau=1}^{t} \vu_i^{(\tau)}[a_i'],
\qquad
M_{t-1} \defeq \max_{a_i' \in \cA_i}\sum_{\tau=1}^{t-1} \vu_i^{(\tau)}[a_i'].
\]
By definition of $a_i^{(t)}$ and $a_i^{(t-1)}$, we have
\(
M_t=\sum_{\tau=1}^{t} \vu_i^{(\tau)}[a_i^{(t)}],
\;
M_{t-1}=\sum_{\tau=1}^{t-1} \vu_i^{(\tau)}[a_i^{(t-1)}].
\)
Hence,
\begin{align}
M_t - M_{t-1}
&=
\sum_{\tau=1}^{t} \vu_i^{(\tau)}[a_i^{(t)}]
-
\sum_{\tau=1}^{t-1} \vu_i^{(\tau)}[a_i^{(t-1)}]
\notag\\
&=
\Bigg(\sum_{\tau=1}^{t-1} \vu_i^{(\tau)}[a_i^{(t)}]
-
\sum_{\tau=1}^{t-1} \vu_i^{(\tau)}[a_i^{(t-1)}]\Bigg)
+
\vu_i^{(t)}[a_i^{(t)}].
\label{eq:aux_min_max_eq1}
\end{align}
By optimality of $a_i^{(t-1)}$, we have
\(
\sum_{\tau=1}^{t-1} \vu_i^{(\tau)}[a_i^{(t-1)}]
\ge
\sum_{\tau=1}^{t-1} \vu_i^{(\tau)}[a_i^{(t)}],
\)
so this term in \eqref{eq:aux_min_max_eq1} is nonpositive and therefore
\begin{equation}
M_t - M_{t-1} \;\ge\; \vu_i^{(t)}[a_i^{(t)}].
\label{eq:aux_min_max_eq2}
\end{equation}
Finally,
\begin{equation}
\vu_i^{(t)}[a_i^{(t)}]
\;\ge\;
\min\!\left\{\vu_i^{(t)}[a_i^{(t-1)}],\,\vu_i^{(t)}[a_i^{(t)}]\right\}.
\label{eq:aux_min_max_eq3}
\end{equation}
Combining \eqref{eq:aux_min_max_eq2} and \eqref{eq:aux_min_max_eq3} proves the claim.
\end{proof}


\begin{lemma}
\label{lemma:aux_decrease_by_constant}
Let $m \ge 2$ and $\delta = \frac{1}{4m}$. Then, for every integer $k$ with
$3 \le k \le 2m-1$, it holds that
\[
(1-\delta)(k-1) - \delta(2m-1) \;\ge\; k-2.
\]
\end{lemma}

\begin{proof}
Starting from the left-hand side and subtracting $k-2$, we obtain
\[
(1-\delta)(k-1) - \delta(2m-1) - (k-2)
=
1-\delta\bigl((k-1)+(2m-1)\bigr)
=
1-\delta(k+2m-2).
\]
It therefore suffices to show that $1-\delta(k+2m-2)\ge 0$, equivalently $\delta(k+2m-2)\le 1$.
Using $\delta=\frac{1}{4m}$ and $k\le 2m-1$, we have
\[
\delta(k+2m-2)
=
\frac{k+2m-2}{4m}
\le
\frac{(2m-1)+2m -2}{4m}
=
\frac{4m-3}{4m}
<1.
\]
Hence $1-\delta(k+2m-2)>0$, and thus
\[
(1-\delta)(k-1) - \delta(2m-1) \ge k-2,
\quad
\text{for}
\;
3 \leq k \leq (2m-1),
\]
as claimed.
\end{proof}


\section{Proofs from Section~\ref{sec:multiplayer}}
\label{appendix:multiplayer}

We conclude with the proofs from~\Cref{sec:multiplayer}.

\recursionsnake*

\begin{proof}
    We let $\calT_k$ be the set of time indices corresponding to $T_k$. Let $i$ be the unique player who has a positive best-response gap during $\calT_k$; uniqueness follows from the snake property of $P$ and the construction of the game. We denote by $a_i(k) \in \{0, 1\}$ the action of player $i$ corresponding to payoff $k$, and similarly for $a_i(k+1)$. $\mathcal{T}_k$ lasts as long as it takes so that $\sum_{\tau=1}^t \vu_i^{(\tau)}[a_i(k+1)] > \sum_{\tau=1}^t \vu_i^{(\tau)}[a_i(k)]$. This is so because $\fictp$ picks with probability $1$ the action with the highest cumulative utility; without any loss, we assume for convenience that ties are broken adversarially, so as to maximize the number of steps. 
    
    Now, for all times $t$ during the iterations in $\calT_{k-1}$, we have $\vu_i^{(t)}[a_i(k+1)] = \vu_i^{(t)}[a_i(k)] - (k-1)$. This holds because during those iterations $i$ was obtaining $(k-1)$ by playing $a_i(k)$ while switching to $a_i(k+1)$ would result in $0$ utility since the corresponding action is off the path.
    
    Furthermore, we claim that during all times in $\calT_{k-2}$ and $\calT_{k-3}$, it holds that $\vu_i^{(t)}[a_i(k+1)] \leq \vu_i^{(t)}[a_i(k)] - 1$. Specifically, during $\calT_{k-2}$, there exists some player $i' \neq i$ with positive best-response gap, for otherwise the snake property of $P$ would be violated---there would exist an edge between the vertex corresponding to $k - 2$ and the vertex corresponding to $k$. This means that during those iterations switching to $a_i(k+1)$ would lead to a utility smaller by at least $1$. During $\calT_{k-3}$, it is possible that $i$ has a positive best-response gap, in which case player $i$ eventually transitions from $a_i(k+1) = a_i(k-3)$ (actions are binary and $a_i(k+1) \neq a_i(k)$) to $a_i(k)$, which means that again $a_i(k)$ is the preferred action and satisfies $\vu_i^{(t)}[a_i(k)] \geq \vu_i^{(t)}[a_i(k+1)] + 1$. If it is not player $i$ who transitions, playing action $a_i(k+1)$ leads off the path, so the same inequality applies.
    
    As a result, if $t$ is the last round in $\mathcal{T}_k$,
    \begin{align*}
        \sum_{\tau=1}^t \vu_i^{(\tau)}[a_i(k+1)] - \sum_{\tau=1}^t \vu_i^{(\tau)}[a_i(k)] \leq T_k  -(k-1) T_{k-1} - T_{k-2} - T_{k-3} + \sum_{\kappa = 1}^{k-4} \kappa T_\kappa.
    \end{align*}
    Here we used the fact that $\vu_i^{(t)}[a_i(k+1)] = \vu_i^{(t)}[a_i(k)] + 1$ for all iterations $t$ in $\calT_k$, and during $\calT_\kappa$ we have $\vu_i^{(t)}[a_i(k+1)] \leq \vu_i^{(t)}[a_i(k)] + \kappa$. Since $t$ is the last round in $T_k$, it follows that $\sum_{\tau=1}^t \vu_i^{(\tau)}[a_i(k+1)] - \sum_{\tau=1}^t \vu_i^{(\tau)}[a_i(k)] > 0$. This implies
\[
    T_{k} \geq (k-1) T_{k-1} + T_{k-2} + T_{k-3} - (k-4) T_{k-4} - \dots - 1 T_1,
\]
as claimed.
\end{proof}

\factorial*

\begin{proof}
We define, for $k \ge 2$,
\[
    S_k \defeq T_k - (k-1)T_{k-1}.
\]
It suffices to prove that $S_k \ge 0$ for all $k \ge 2$, as it would imply $T_k \ge (k-1)T_{k-1}$, which in turn yields $T_k \ge (k-1)!$ since $T_3 = 2!$ Using the recurrence relation for $T_k$,
\begin{align}
    S_k &= T_k - (k-1)T_{k-1} \notag \\
        &\ge T_{k-2} + T_{k-3} - \sum_{j=1}^{k-4} j T_j. \label{eq:Sk_bound}
\end{align}
Let $R_k$ denote the right-hand side of \eqref{eq:Sk_bound}. We proceed by strong induction. We assume $R_\kappa \ge 0$ for all $2 \le \kappa < k$, which in turn implies $T_\kappa \ge (\kappa-1)T_{\kappa-1}$ for $2 \leq \kappa < k$. We will show that $R_k \geq 0$. The basis of the induction $k \in \{2, 3, 4 \}$ follows trivially, so we take $k \geq 5$.

We consider the difference
\begin{align*}
    R_k - R_{k-1} &= \left( T_{k-2} + T_{k-3} - \sum_{j=1}^{k-4} j T_j \right) - \left( T_{k-3} + T_{k-4} - \sum_{j=1}^{k-5} j T_j \right) \\
    &= T_{k-2} - T_{k-4} - (k-4)T_{k-4} \\
    &= T_{k-2} - (k-3)T_{k-4}.
\end{align*}

By the inductive hypothesis,
\[
    T_{k-2} \ge (k-3)T_{k-3} \ge (k-3)(k-4)T_{k-4}.
\]
So, $R_{k} - R_{k-1} \geq (k - 3)(k - 5)T_{k-4} \geq 0$. Given that $R_{k-1} \geq 0$, it follows that $R_k \geq 0$, and the claim follows.
\end{proof}

\end{document}